\documentclass[11pt,reqno]{article}

\usepackage{amssymb}
\usepackage{latexsym}
\usepackage{amsmath}

\usepackage{hyperref,cite}

\usepackage{amsthm}
\usepackage{empheq}
\usepackage{bm}
\usepackage[capitalise]{cleveref}
\usepackage{todonotes}

\usepackage{longtable}
\usepackage{multirow,multicol}
\usepackage{caption}
\usepackage{ltablex}\keepXColumns

\newcolumntype{Y}{>{\centering\arraybackslash}X}
\newcolumntype{Z}{>{\scriptsize}Y}
\usepackage{rotating}

\numberwithin{equation}{section}
\theoremstyle{definition}
\newtheorem{theorem}{Theorem}[section]
\newtheorem{corollary}[theorem]{Corollary}
\newtheorem{proposition}[theorem]{Proposition}
\newtheorem{definition}[theorem]{Definition}
\newtheorem{example}[theorem]{Example}
\newtheorem{const}[theorem]{Construction}
\newtheorem{notation}[theorem]{Notation}
\newtheorem{remark}[theorem]{Remark}
\newtheorem{lemma}[theorem]{Lemma}

\newtheorem{conjecture}[theorem]{Conjecture}

\newcounter{alp}
\newcounter{ara}
\newcounter{rom}
\newenvironment{romanlist}{\begin{list}{(\roman{rom})\hfill}{\usecounter{rom}
     \topsep0ex \labelwidth1.7em \leftmargin1.7em \labelsep0cm
     \rightmargin0cm \parsep0ex \itemsep.4ex
     \partopsep1ex}}{\end{list}}
\newenvironment{alphalist}{\begin{list}{(\alph{alp})\hfill}{\usecounter{alp}
     \topsep0ex \labelwidth.6cm \leftmargin.6cm \labelsep0cm
     \rightmargin0cm \parsep0ex \itemsep0ex
     \partopsep0ex}}{\end{list}}
\newenvironment{arabiclist}{\begin{list}{(\arabic{ara})\hfill}{\usecounter{ara}
     \topsep0ex \labelwidth1.5em \leftmargin1.5em \labelsep0cm
     \rightmargin0cm \parsep0ex \itemsep.5ex
     \partopsep0ex}}{\end{list}}
\newcommand\qbin[3]{\left[\begin{matrix} #1 \\ #2 \end{matrix} \right]_{#3}}
\newcommand\qbinsmall[3]{\Big[\textstyle{\begin{matrix} #1 \\ #2 \end{matrix}} \Big]_{#3}}

\newcommand{\colsp}{\textnormal{colsp}}

\newcommand{\numberset}{\mathbb}
\newcommand{\N}{\numberset{N}}
\newcommand{\Z}{\numberset{Z}}

\newcommand{\R}{\numberset{R}}

\newcommand{\F}{\numberset{F}}

\newcommand{\mV}{\mathcal{V}}
\newcommand{\mU}{\mathcal{U}}

\newcommand{\mL}{\mathcal{L}}

\newcommand{\mC}{\mathcal{C}}
\newcommand{\mA}{\mathcal{A}}
\newcommand{\mI}{\mathcal{I}}

\newcommand{\mT}{\mathcal{T}}
\newcommand{\rk}{\textnormal{rk}}
\newcommand{\srk}{\textnormal{srk}}

\newcommand{\tail}{\textnormal{tail}}
\newcommand{\head}{\textnormal{head}}
\newcommand{\nmax}{n^\ast}
\newcommand{\nmin}{n_\ast}

\newcommand{\mP}{\mathcal{P}}

\newcommand{\subspace}[1]{\mbox{$\langle{#1}\rangle$}}

\newcommand{\GaussianD}[2]{\genfrac{[}{]}{0pt}{0}{#1}{#2}}
\newcommand{\SmallMat}[4]{\mbox{\small{$\begin{pmatrix}{#1}&{\!\!\!#2}\\{#3}&{\!\!\!#4}\end{pmatrix}$}}}
\newcommand{\SmallMatTwo}[2]{\mbox{\small{$\begin{pmatrix}{#1}\\{#2}\end{pmatrix}$}}}

\newcommand{\rkL}{\textnormal{rk}_{\mL}}
\newcommand{\muL}{\mbox{$\mu_{\mL}$}}

\newcommand*{\myproofname}{Proof}

\usepackage{authblk}

\usepackage{array}
\newcolumntype{x}[1]{>{\centering\arraybackslash\hspace{0pt}}p{#1}}
\newcolumntype{y}[1]{>{\centering\arraybackslash\hspace{0pt}}m{#1}}

\renewcommand{\to}{\longrightarrow}
\renewcommand{\mapsto}{\longmapsto}

\usepackage{setspace}

\usepackage{tikz,pgfplots}
\pgfplotsset{compat=newest}

\begin{document}
	\title{\textbf{Fundamental Properties of Sum-Rank Metric Codes}}

\author[1]{Eimear Byrne}
\author[2]{Heide Gluesing-Luerssen\thanks{H. Gluesing-Luerssen was partially supported by the grant \#422479 from the Simons Foundation.}}
\author[3]{Alberto Ravagnani}

\affil[1]{School of Mathematics and Statistics, University College Dublin, Ireland}
\affil[2]{Department of Mathematics, University of Kentucky, USA}
\affil[3]{Department of Mathematics and Computer Science, Eindhoven University of Technology, the Netherlands}

	\date{}
	\maketitle
	
	\begin{abstract}\label{sec:Abstract}
This paper investigates the theory of sum-rank metric codes for which the individual matrix blocks may have different sizes. Various bounds on 
the cardinality of a code are derived, along with their asymptotic extensions.  
The duality theory of sum-rank metric codes is also explored,
showing that MSRD codes (the sum-rank analogue of MDS codes) dualize to MSRD codes only if all matrix blocks have the 
same number of columns. In the latter case, duality considerations lead to an upper bound on the number of blocks for MSRD codes.
The paper also contains various constructions of sum-rank metric codes for variable block sizes, illustrating the possible behaviours of these objects with respect to bounds, existence, and duality properties. 
	\end{abstract}

\bigskip

\bigskip

\section{Introduction}

It is well known that codes with the rank metric offer a solution to the problem of error amplification in linear
network coding, both in the one-shot and multi-shot regime; see~\cite{silva2008rank,ravagnani2018adversarial,koetter2008coding}. In the latter scenario however, using codes with the
\textit{sum-rank} metric can significantly reduce the size of the network alphabet~\cite{NoUF10}, which is particularly handy in fast-evolving systems.
The sum-rank metric itself occurred already earlier in space-time coding~\cite{ElGH03,LuKu05,ShKsch20}.
Besides applications in multi-shot linear network coding (see also~\cite{MPK19}) and space-time coding,
sum-rank metric codes can also be utilized for distributed storage~\cite{MPK19a}.
Furthermore, convolutional codes endowed with the sum-rank metric have been considered in, for instance,~\cite{MBK16}  in order to address network streaming problems.

The theory of sum-rank metric codes is still in its beginnings and
to date only a few constructions are known.
A general construction of MSRD codes (the analogue of MDS codes) up to a certain length, along with a Welch-Berlekamp sum-rank decoding algorithm, can be found in~\cite{MP18} (see also~\cite{MPK19a}).
These codes are \emph{linearized Reed-Solomon codes}, i.e., a hybrid between Reed-Solomon codes and Gabidulin codes, and their
duals are again linearized Reed-Solomon codes~\cite{MPK19}.
Furthermore, in~\cite{MPK19a}  sum-rank alternant codes are introduced, and sum-rank BCH codes have been introduced in~\cite{MP20}.
Finally, a generic decoding algorithm for sum-rank metric codes is presented in~\cite{PRR20}.
In all the above mentioned results, the sum-rank metric codes considered can be viewed as linear spaces over an extension field, which in particular requires that all matrix blocks have the same number of columns.

\paragraph{Our Contribution.} \ In this paper, we study the fundamental properties of codes in the sum-rank metric
with respect to bounds, their asymptotics, duality, and existence/optimality.
We do not require that 
all blocks have the same number of columns.
In particular, we only assume that the codes are linear over $\F_q$, and so do not necessarily have a representation as a block code over an extension field (cf.\ Definition~\ref{D-SRMC}).
As we will see, removing these restrictions makes the theory of sum-rank metric codes, and in particular of \emph{maximum sum-rank metric codes} (MSRD codes), significantly richer from a mathematical viewpoint.

After setting up the preliminaries, in Section~\ref{sec:bounds} we derive several bounds.
We start by showing that any bound for codes in the Hamming metric yields a bound for sum-rank 
metric codes, which we call the corresponding \textit{induced bound}. 
We then establish additional bounds for the parameters of a sum-rank metric code and 
show that these are in general not comparable with each other, i.e., each bound is the unique best bound 
for a specific class of parameters. 
We also provide examples of codes meeting these bounds. 
While the induced bounds only take into account the maximum number of columns over all blocks,
the additional bounds depend on all block dimensions; this refinement is evident in the sharpness of the latter bounds compared to those induced from the known Hamming metric bounds.  

In Section~\ref{S-Asymp} we turn to the theory of asymptotic bounds for sum-rank metric codes by letting the number of blocks in the 
ambient space grow (imposing natural constraints on the block sizes) and describing the behaviour of the rate of an optimal code.
This results in the asymptotic version of the bounds mentioned above. 

Next, we concentrate on the duality theory of sum-rank metric codes and their MacWilliams identities.
We introduce the notions of 
\textit{sum-rank}, \textit{rank-list} and \textit{support} \textit{distributions} of a sum-rank metric code.
While the first one does not obey a MacWilliams identity in general, we prove that the latter two do, and 
give a closed formula for the corresponding MacWilliams transformations using a combinatorial approach.
This generalizes the MacWilliams identities for both the Hamming-metric and the rank-metric.

A substantive part of the paper is devoted to the study of MSRD codes, i.e., those sum-rank metric codes
that attain the sum-rank metric analogue of the Singleton Bound. 
We show that when the blocks of the ambient space all have the same number of columns, 
the dual of an MSRD code is again MSRD. 
Interestingly, this duality result is not true in general: if the ambient space supports variable block sizes, then 
the duality result does not necessarily hold.
Moreover,  under the assumption that all blocks have the same number of columns, we explicitly compute
the support distribution of an MSRD code. 
We also investigate for which parameters MSRD codes exist (which extends the problem for which parameter sets MDS codes exist). 
Our contribution is this direction is twofold. 
We first provide a non-existence criterion with the aid of the support distribution, which
allows us to rule out MSRD codes for certain parameters even if none of our bounds rules them out.
Secondly, we use a modification of the sphere-packing bound to establish an upper bound
on the number of blocks of an MSRD code. 
As simple corollaries, we recover classical upper bounds on the length of an MDS code.

We conclude the paper with a section on constructions of optimal codes. We first provide several constructions of MSRD 
codes for specific families of parameters. 
Although our techniques do not extend to arbitrary parameter sets, to our best knowledge they produce the only known 
constructions of MSRD codes for ambient spaces where the blocks have different numbers of columns.
Finally, we present a \textit{lifting construction} that produces a sum-rank metric code combining a Hamming-metric and a rank-metric code. As an application, we obtain a class of sum-rank metric codes that attain the induced Plotkin bound with equality.

\section{Sum-Rank Metric Codes}\label{S-Prelim}
Throughout this paper, let $\N=\{1,2,...\}$ and $\N_0=\N \cup \{0\}$.
For $a \in \N$, denote by $[a]$ the set $\{1,...,a\}$.
Let $q$ denote a prime power and $t, n_1,...,n_t,m_1,...,m_t \ge 1$ be integers such that
\begin{equation}\label{e-nimi}
     n_i\leq m_i\ \text{ for all }\ i\in[t]\ \text{ and }\ m_1 \ge \cdots \ge m_t.
\end{equation}
We consider the product of $t$ matrix spaces
\begin{equation}\label{e-Pi}
  \Pi:=\Pi_q(n_1\times m_1\mid\cdots\mid n_t\times m_t):=\bigoplus_{i=1}^t \F_q^{n_i \times m_i}
\end{equation}
and define the \textbf{sum-rank} of an element $X=(X_1,...,X_t) \in \Pi$  as

\[
   \srk(X):=\sum_{i=1}^t \rk(X_i).
 \]
It is easy to see that the sum-rank induces a metric on~$\Pi$ via $(X,Y)\longmapsto \srk(X-Y)$.

\begin{definition}\label{D-SRMC}
A (\textbf{sum-rank metric}) \textbf{code} is an $\F_q$-linear subspace of the metric space~$\Pi$.
The \textbf{minimum} (\textbf{sum-rank}) \textbf{distance} of a non-zero code $C \le \Pi$ is defined as usual via 
$\srk(C):=\min\{\srk(X) \mid X \in C, \; X \neq 0\}.$
We say that a code $C \le \Pi$ is \textbf{trivial} when $C=\{0\}$ or $C=\Pi$.
For a codeword $(X_1,\ldots,X_t)$  we call the matrix~$X_i$ the \textbf{$i$-th block}.
\end{definition}

Obviously, the definition reduces to rank-metric codes in the case where $t=1$ and to block codes of length~$t$ with the 
Hamming metric in the case where $m_i=n_i=1$ for all $i\in[t]$.

\begin{remark}\label{R-Fqmlinear}
A special class of sum-rank metric codes arises in the case where $m_i=m$ for all $i\in[t]$. 
In this case $\Pi\cong\F_{q^m}^N$, where $N=\sum_{i=1}^t n_i$, and we may consider sum-rank-metric codes as vector codes in $\F_{q^m}^N$. 
In this case the codes may be $\F_q$-linear or even $\F_{q^m}$-linear.
The latter class of codes has been studied  in \cite{MP18,MP19,MPK19,MPK19a}. 
In~\cite[Thm.~1]{NoUF10} it has been shown that a code in $\F_{q^m}^N$ with sum-rank distance~$d$ can correct~$\alpha$ errors 
and~$\beta$ erasures whenever $2\alpha+\beta<d$.
In~\cite[Thm.~1]{MPK19} the converse has been established: any code with $\alpha$-error-$\beta$-erasure correction capability must
have sum-rank distance at least $2\alpha+\beta$.
\end{remark}

Sum-rank metric codes may also be considered as rank-metric codes consisting of matrices supported on a particular profile.
More precisely, let
\begin{equation}\label{e-NM}
    N=n_1+\cdots+n_t \quad \text{ and } \quad M=m_1+\cdots +m_t.
\end{equation}
For any subset $\mP \subseteq [N] \times [M]$ define $\F_q^{N \times M}[\mP]$ as the space of
$N \times M$ matrices over~$\F_q$ supported on $\mP$ (i.e., whose nonzero entries have indices in~$\mP$).
Setting now
\begin{equation}\label{e-calP}
  \mP:=\bigcup_{i=1}^t
   \Big(\Big\{\sum_{j=0}^{i-1}n_j+1,\ldots,\sum_{j=0}^{i}n_j\Big\}\times \Big\{\sum_{j=0}^{i-1}m_j+1,\ldots,\sum_{j=0}^{i}m_j\Big\}\Big),
\end{equation}
where $n_0=m_0=0$, we obtain an $\F_q$-isomorphism
\begin{equation}\label{e-Psi}
  \psi:\Pi\longrightarrow\F_q^{N \times M}[\mP],\quad
   (X_1,\ldots,X_t)\longmapsto \begin{pmatrix}X_1 & & \\ &\ddots& \\ & &X_t\end{pmatrix},
\end{equation}
satisfying $\rk(\psi(X))=\srk(X)$ for all $X\in\Pi(n_1\times m_1\mid\cdots\mid n_t\times m_t)$.
In other words,~$\psi$ is an isometry between the metric spaces $(\Pi,\srk)$ and $(\F_q^{N \times M}[\mP],\rk)$.
Occasionally, it will be helpful to consider sum-rank metric codes as rank-metric codes in $\F_q^{N \times M}[\mP]$. 
When we do so, we will make this explicit.

\medskip

In order to define the  support of elements in~$\Pi$, we need the following.

\begin{definition}\label{D-Lattice}
For $i\in[t]$ let $\mL_i$ be the lattice of subspaces of $\F_q^{n_i}$ partially ordered by inclusion.
Let $\mL$ be the product lattice, that is, $\mL:=\mL_1 \times \cdots \times \mL_t$ endowed with the product order, which we denote  
by~$\subseteq$. 
The rank function on~$\mL$ is given by
$\rk_\mL(\bm{U})=\sum_{i=1}^t \dim(U_i)$ for all $\bm{U}=(U_1,...,U_t) \in \mL$.
We also set $\dim(\bm{U}):=(\dim(U_1),...,\dim(U_t))$ for $\bm{U} \in \mL$.
Finally, the M\"obius function of $\mL$ is given by 
\[
    \mu_\mL(\bm{U},\bm{V})=\prod_{i=1}^t (-1)^{\dim(V_i)-\dim(U_i)} q^{\binom{\dim(V_i)-\dim(U_i)}{2}} \quad 
    \text{ for }\bm{U} \subseteq \bm{V}.
\]
\end{definition}
We refer the reader to \cite{stanleyec} for more details.

The following definitions of support and shortening are straightforward extensions of the rank-metric case; 
see~\cite[Ex.~39]{ravagnani2018duality}.

\begin{definition}\label{D-Supp}
For a matrix $M\in\F^{a\times b}$ let $\colsp(M)\leq\F^a$ denote its column space.
We define the \textbf{support} as
\[
   \sigma:\Pi\longrightarrow\mL,\quad  (X_1,\ldots,X_t)\longmapsto(\colsp(X_1),...,\colsp(X_t)).
\]
Note that $\srk(X)=\rk_\mL(\sigma(X))$ for all $X \in \Pi$, that is, the sum-rank weight is the composition of~$\sigma$ with the rank function of $\mL$.
\end{definition}

\begin{definition}\label{D-Short}
For a code $C \le \Pi$ and $\bm{U} \in \mL$ we define the \textbf{shortening} of $C$ on $\bm{U}$ as
$$C(\bm{U}):=\{X \in C \mid \sigma(X) \subseteq \bm{U}\} \le \Pi.$$
\end{definition}

\begin{remark} \label{R-PiU}
Let $\bm{U}=(U_1,\ldots,U_t)$ and $C=\Pi$ as in~\eqref{e-Pi}. Then
$\Pi(\bm{U})$ is a code in $\Pi$ of dimension $\sum_{i=1}^t m_i \dim(U_i)$. Moreover, every $X \in \Pi(\bm{U})$ satisfies
$\srk(X) \le \sum_{i=1}^t  \dim(U_i) =\rk_\mL(\bm{U})$.
\end{remark}

The next result easily follows from the definitions.

\begin{proposition}\label{R-eqv}
Let~$\Pi$ be as in~\eqref{e-Pi}. For a non-zero code $C \le \Pi$  and $1 \le d \le N=n_1+ \cdots +n_t$ the following are equivalent:
\begin{arabiclist}
\item  $\srk(C) \ge d$;
\item  $|C(\bm{U})|=1$ for all $\bm{U} \in \mL$ with $\rk_\mL(\bm{U}) \le d-1$;
\item  $|C(\bm{U})|=1$ for all $\bm{U} \in \mL$ with $\rk_\mL(\bm{U}) = d-1$.
\end{arabiclist}
\end{proposition}

We close this preliminary section by introducing three partition enumerators of codes in $\Pi$,
which will be studied later on in the paper.

\begin{definition}\label{D-WeightEnum}
Let $C \le \Pi$ be a code. For $r \in \N_0$, $\bm{r}=(r_1,...,r_t) \in \N_0^t$ and $\bm{U} \in \mL$
let
\begin{align*}
W_r(C) &:= \big|\{X \in C \mid \srk(X)=r\}\big|, \\
W_{\bm{r}}(C) &:= \big|\{X \in C \mid \rk(X_i)= r_i \mbox{ for all $i \in [t]$} \}\big|, \\
W_{\bm{U}}(C) &:= \big|\{X \in C \mid \sigma(X)=\bm{U}\}\big|.
\end{align*}
We call the lists $\big(W_r(C)\big)_{r\in\N_0},\ \big(W_{\bm{r}}(C)\big)_{\bm{r}\in \N_0^t}$, and 
$\big(W_{\bm{U}}(C)\big)_{\bm{U} \in \mL}$ the \textbf{sum-rank}, \textbf{rank-list} and
\textbf{support distributions} of $C$, respectively.
\end{definition}

Note that
\begin{equation}\label{e-CUWV}
   |C(\bm{U})|=\sum_{\bm{V}\leq\bm{U}}W_{\bm{V}}(C).
\end{equation}

\section{Bounds for Sum-Rank Metric Codes} \label{sec:bounds}
In this and the following section we will deviate from Definition~\ref{D-SRMC} and consider general (that is, not necessarily $\F_q$-linear) codes.
We will establish various bounds for the cardinality of a sum-rank metric code as a function of its parameters. 
To our knowledge, the only bound for sum-rank metric codes established so far is the Singleton Bound for 
the case where $m_1= \cdots =m_t$; see~\cite[Cor.~2]{MPK19a}.
As before, we will not make this assumption on $m_1,\ldots,m_t$.

For ease of exposition, we first present all bounds and relegate the proofs to Subsection~\ref{sub:proofs}. 
We compare the bounds with each other in Subsection~\ref{sub:compar}.

A simple method to obtain bounds for the dimension of a sum-rank metric code $C\subseteq\Pi$
is to use bounds for non-linear codes in the Hamming metric.
Denote by $\mA_Q(n,d)$ the largest cardinality of a (possibly non-linear) code over an alphabet of
size~$Q$ with length~$n$ and minimum Hamming distance at least~$d$. 
We set $\mA_Q(n,d):=1$ if such a code does not exist. 
The following result shows how any bound for a non-linear Hamming-metric code yields a bound for
sum-rank metric codes. We call these bounds ``induced''.
For details see Section~\ref{sub:proofs}. 

Throughout we assume $\F=\F_q$ and $\Pi= \Pi_q(n_1\times m_1\mid\cdots\mid n_t\times m_t)$ with the specifications as is~\eqref{e-nimi}--\eqref{e-NM}.
In particular, $N=n_1+\ldots+n_t$.

\begin{theorem}\label{th:induced}
Let $m=\max\{m_1,...,m_t\}$ and let $C\subseteq\Pi$ be sum-rank metric code with $|C| \ge 2$ and $\srk(C)=d$.
Then
$|C| \le \mA_{q^{m}}(N,d)$, which leads to the following bounds.
\\[1ex]
\mbox{}\hspace*{-.6em}\begin{tabular}{lp{11cm}}
  \textbf{Induced Singleton Bound:}&  $|C| \le q^{m(N-d+1)}$,\\[1.3ex]
  \textbf{Induced Hamming Bound:}&
       ${\displaystyle |C| \le \bigg\lfloor\frac{q^{mN}}{\sum_{s=0}^{\lfloor(d-1)/2\rfloor}\binom{N}{s}(q^{m}-1)^s}}\bigg\rfloor$,\\[2.5ex]
  \textbf{Induced Plotkin Bound:}& ${\displaystyle |C| \le \Big\lfloor\frac{q^{m}d}{q^md-(q^{m}-1)N}\Big\rfloor, \:\:\text{ if }d > (q^m-1)N/q^m},$\\[2ex]
  \textbf{Induced Elias Bound:}&
  ${\displaystyle |C| \le \Big\lfloor\frac{Nd(q^m-1)}{q^mw^2-2Nw(q^m-1)+(q^m-1)Nd} \cdot \frac{q^{mN}}{V_w(\F_{q^m}^N)}}\Big\rfloor$,
\end{tabular}
\\[1.7ex]
where in the Induced Elias bound,~$w$ is any integer in the interval $[0,\, N(q^m-1)/q^m]$ such that the denominator is 
positive and $V_w(\F_{q^m}^N)$ is the volume of any sphere of radius~$w$ in $\F_{q^m}^N$ with respect to the Hamming distance, that is $V_w(\F_{q^m}^N)=\sum_{i=0}^w \binom{N}{i}(q^m-1)^i$.
\end{theorem}

The next bound is obtained from a projection argument and  can be regarded as the sum-rank analogue of the Singleton Bound.
It generalizes~\cite[Cor.~2]{MPK19a}.

\begin{theorem}[\textbf{Singleton Bound}]\label{th:singl}
Let $C \subseteq \Pi$ be a code with $|C| \ge 2$ and $\srk(C) = d$.
Let $j$ and $\delta$ be the unique integers satisfying $d-1=\sum_{i=1}^{j-1}n_i+\delta$ and $0\leq\delta\leq n_j-1$.
Then
\[
    |C|\leq q^{ \sum_{i=j}^t m_in_i- m_j\delta}.
\]
In the special case where $m_1= \cdots =m_t=:m$, the upper bound simplifies to
$|C| \le q^{m(N-d+1)}$, which agrees with the Induced Singleton Bound in Theorem~\ref{th:induced}.
\end{theorem}

\begin{definition}\label{D-MSRD}
A code $C\subseteq\Pi$ is called \textbf{MSRD} (maximum sum-rank distance) if its cardinality attains the 
Singleton Bound of Theorem \ref{th:singl}, or if $|C|=1$.
\end{definition}

Our next bounds require us to give the cardinality of a sum-rank sphere.

\begin{definition}
	For $r\in\N$ we define 
	\[
	V_r(\Pi):=\sum_{s=0}^r  \ \sum_{\substack{(s_1,...,s_t) \in \N_0^t \\ s_1+ \cdots + s_t=s}} \ 
	\prod_{i=1}^t \qbin{n_i}{s_i}{q} \prod_{j=0}^{s_i-1} (q^{m_i} - q^j).
	\]
\end{definition}

It is not hard to see that $V_r(\Pi)$ is the volume of any sphere in~$\Pi$ of sum-rank radius~$r$.

\begin{lemma}\label{L-sphere}
For all $r \in \N$ we have 	$V_r(\Pi) =  |\{ (X_1,...,X_t) \in \Pi \mid \srk(X_1,...,X_t)\leq r \}|$.
\end{lemma}

We next apply a sphere-packing argument in the metric space $(\Pi,\srk)$. 

\begin{theorem}[\textbf{Sphere-Packing Bound}]\label{th:pack}
A code $C \subseteq\Pi$ with $|C| \ge 2$ and $\srk(C)=d$  satisfies
\[
  |C| \leq \bigg\lfloor\frac{|\Pi|}{V_r(\Pi)}\bigg\rfloor,\ \text{ where }r=\lfloor(d-1)/2\rfloor.
\]
\end{theorem}

The sphere-packing bound can be combined with projection arguments to obtain other bounds for the cardinality of a sum-rank metric  
code. We include one of these results. 
In Section~\ref{sec:msrd}, we will make use of it to obtain non-existence criteria for MSRD codes.

\begin{theorem}[\textbf{Projective Sphere-Packing Bound}]\label{th:prpack}
Let $C\subseteq \Pi$ be a code with $|C| \ge 2$ and minimum distance $3\leq d \leq N$.
Let $\ell\in[t-1]$ and $\delta\in[n_{\ell+1}-1]$ be the unique integers such that 
$d-3=\sum_{j=1}^\ell n_j+\delta$. 
Define $\Pi':= \Pi_q((n_{\ell+1}-\delta) \times m_{\ell+1} \mid n_{\ell+2} \times m_{\ell+2} \mid\cdots\mid n_t\times m_t)$.
Then
$$|C| \le \bigg\lfloor\frac{|\Pi'|}{V_1(\Pi')}\bigg\rfloor.$$
\end{theorem}

Estimating for a given code $C\subseteq\Pi$ the sum of the sum-rank distances between its codewords (the total distance), we arrive 
at our final bound.
Even though its proof in the next section suggests that it is a very coarse bound, 
it turns out to be often the best among all bounds; see Section~\ref{sub:compar}. 

\begin{theorem}[\textbf{Total-Distance Bound}]\label{th:totalwt}
Let $C \subseteq \Pi$ be a code with $|C| \ge 2$ and sum-rank distance $d$. Let $Q=\sum_{i=1}^t q^{-m_i}$.
Then
$$d \le N + \frac{t - |C| Q}{|C|-1}.$$
In particular, if $d> N - Q$, then 
$$|C| \leq \frac{d - N + t}{d-N + Q}. $$
\end{theorem}

As a corollary of Theorem~\ref{th:totalwt},
 we obtain an upper bound for the length $t$ of a sufficiently large sum-rank metric code. 
The proof follows from using $|C|\sum_{i=1}^tq^{-m_i}\leq t|C|q^{-m}$ in the first inequality above. 

\begin{corollary} \label{cor:tlarge}
Let $m:=\max\{m_1,...,m_t\}$ and let $C \subseteq \Pi$ be a sum-rank metric code 
with cardinality $|C| > q^m$ and sum-rank distance~$d$. Then
$$t \le (N-d) q^m \frac{|C|-1}{|C|-q^m}.$$
\end{corollary}
In Section~\ref{sec:msrd} we will improve the upper bound of Corollary~\ref{cor:tlarge}
for the class of MSRD codes and some parameter sets; see Theorem~\ref{coro:nec}.

All of the above bounds apply to general (not necessarily $\F_{q}$-linear) codes.
Since an $\F_q$-linear code with cardinality~$b$ clearly has dimension $\log_q(b)$,
we propose the following expressions.

\begin{definition}\label{D-LinBound}
If~$b$ is the value of any of the above bounds, we call $\lfloor\log_q(b)\rfloor$ the 
\textbf{$\F_q$-linear version} of the bound.
We say that a linear code~$C\leq\Pi$ \textbf{attains the linear version of the bound} if $\dim(C)=\lfloor\log_q(b)\rfloor$.
\end{definition}

We conclude this part with the sum-rank analogue of the Gilbert bound, which follows by the application of the standard sphere-covering argument, i.e. that the union of the spheres of sum-rank radius $d-1$ centered at each codeword of a code of minimum sum-rank distance $d$ is the entire ambient space $\Pi$.

\begin{theorem}[\textbf{Sphere-Covering Bound}] \label{th:cov}
Let $1 \le d \le N$ be an integer. Denote by $k$ the smallest integer such that
\[
    q^k \ge \left\lceil \frac{ |\Pi|}{V_{d-1}(\Pi)} \right\rceil,
\]
where $V_{d-1}(\Pi)$ is as in Theorem~\ref{th:pack}.
Then there exists a linear code $C \leq \Pi$ of dimension~$k$ and sum-rank distance at least~$d$.
\end{theorem}

\subsection{Proofs of the Bounds} \label{sub:proofs}
In this subsection we present the proofs of the various bounds.

\begin{proof}[Proof of Theorem~\ref{th:induced}]
Fix an ordered basis $\{\gamma_1,...,\gamma_{m}\}$ of $\F_{q^{m}}$ over 
$\F_q$. 
For all $i\in[t]$ define the map $f^i: \F_{q}^{n_i \times m_i} \to \F_{q^{m}}^{n_i}$
as $f^i(X)_j:= \sum_{\ell=1}^{m_i} X_{j,\ell}  \, \gamma_\ell$,
for all $X \in \F_{q}^{n_i \times m_i}$ and $j\in[n_i]$.
This gives rise to a map
\begin{equation} \label{def_f}
  f:  \Pi \longrightarrow \F_{q^{m}}^N,\quad 
  (X_1,...,X_t)\longmapsto (f^1(X_1),...,f^t(X_t)).
\end{equation}
Observe that for all $i$ and all $X,X' \in \F_{q}^{n_i \times m_i}$ the Hamming weight
of $f^i(X)-f^i(X')$ equals the number of nonzero rows of~$X-X'$ and thus is
at least the rank of~$X-X'$. Therefore, for $X,X' \in \Pi$, the Hamming distance between
$f(X)$ and $f(X')$ is at least the sum-rank distance between $X$ and~$X'$. In particular, $f$ is injective 
(and bijective when $m_1= \cdots =m_t$).

All of this shows that the image $f(C)$ of $C$ is a code
in $\F_{q^{m}}^N$ with the same cardinality as $C$ and whose minimum Hamming distance is at least $d$. 
Thus $|C| \le \mA_{q^{m}}(N,d)$, as claimed.
The specific bounds are well known, see e.g. \cite[Ch.~2]{pless}.
\end{proof}

For the proof of Theorem~\ref{th:singl} we need the following lemma.

\begin{lemma}\label{L-Maximize}
Let $0\leq r< N$ and $\mU_r=\{(u_1,\ldots,u_t)\in\N_0^t\mid 0\leq u_i\leq n_i,\,\sum_{i=1}^t u_i=r\}$. Set $K:=\max\{\sum_{i=1}^t m_i u_i\mid (u_1,\ldots,u_t)\in\mU_r\}$.
Let $j\in[t]$ and $\delta\in[n_j-1]$ be the unique integers such that $r=\sum_{i=1}^{j-1}n_i+\delta$.
Then 
\[ 
   K=\sum_{i=1}^{j-1}m_in_i+m_j\delta.
\]
\end{lemma}

\begin{proof}
First of all, $u=(n_1,\ldots,n_{j-1},\delta,0,\ldots,0)$ is in $\mU_r$ and thus $K\geq\tilde{K}:=\sum_{i=1}^{j-1}m_in_i+m_j\delta$.
For the converse we proceed as follows. Let $(u_1,\ldots,u_t)\in\mU_r$ and set $\alpha_i=n_i-u_i$ for $i\in[t]$. Then $\alpha_i\geq 0$ and $\sum_{i=1}^tu_i=r=\sum_{i=1}^{j-1}n_i+\delta$ implies
$\delta=\sum_{i=j}^t n_i-\sum_{i=1}^t\alpha_i$. Using $m_1\geq\ldots\geq m_t$ we obtain
\begin{align*}
   \sum_{i=1}^t m_i u_i&= \sum_{i=1}^{j-1}m_in_i-\sum_{i=1}^{j-1}m_i\alpha_i+\sum_{i=j}^t m_iu_i
       \leq\tilde{K}-m_j\delta-m_j\sum_{i=1}^{j-1}\alpha_i+m_j\sum_{i=j}^tu_i\\
      &=\tilde{K}-m_j \left(\sum_{i=j}^tn_i-\sum_{i=1}^t \alpha_i+\sum_{i=1}^{j-1}\alpha_i-\sum_{i=j}^tu_i \right)=\tilde{K}.
\end{align*}
This shows $K\leq\tilde{K}$.
\end{proof}

\begin{proof}[Proof of Theorem~\ref{th:singl}]
Fix any $(u_1,... ,u_t) \in \N_0^t$ with $0 \le u_i \le n_i$ and $\sum_{i=1}^t u_i =d-1$.
For each $i \in \{1,...,t\}$ choose~$u_i$ coordinates of $\{1,...,n_i\}$
and delete the corresponding rows from the matrices in $\F_q^{n_i \times m_i}$. 
This provides us with a projection
\[
     \tau:\Pi\longrightarrow\hat{\Pi}:=\Pi_q\big((n_1-u_1) \times m_1\mid\cdots\mid(n_t-u_t)\times m_t\big).
\]
Consider now a code $C\subseteq\Pi$ with sum-rank distance~$d$.
Its image is $\tau(C) \subseteq\hat{\Pi}$.
For any distinct $X,Y \in C$ we have $\srk(X-Y)\geq d$ and therefore $\srk(\tau(X),\tau(Y))\geq1$.
This shows that~$\tau$ is injective and hence 
$|C|=|\tau(C)| \leq |\hat{\Pi}|=q^{\sum_{i=1}^t(n_i-u_i)m_i}$.
The best bound arises by maximizing $\sum_{i=1}^t u_im_i$ subject to the above conditions on $(u_1,\ldots,u_t)$.
Hence Theorem~\ref{th:singl} follows from Lemma~\ref{L-Maximize}.
\end{proof}

\begin{proof}[Proof of Lemma \ref{L-sphere}]
	Fix any $X \in \Pi$ and an integer $r \ge 0$. 
	The bijection  on~$\Pi$ given by the shift $Y \longmapsto X-Y$ implies that 
	\[
	\big|\{Y \in \Pi \mid \srk(X-Y) \le r\} \big| =  \big|\{Y \in \Pi \mid \srk(Y) \le r\} \big|,
	\]
	and thus it remains to compute the right hand side. 
	Using the well-known formula for the number of $(n_i\times m_i)$-matrices of rank~$s_i$ we easily obtain
	\begin{align*}
		\big|\{Y \in \Pi \mid \srk(Y) \le r\} \big|&=\sum_{s=0}^r  \ \sum_{\substack{(s_1,...,s_t) \in \N_0^t \\ s_1+ \cdots + s_t=s}}
		\big|\{(Y_1,...,Y_t) \in \Pi \mid \rk(Y_i) = s_i \mbox{ for all $i$}\} \big| \\
		&= \sum_{s=0}^r  \ \sum_{\substack{(s_1,...,s_t) \in \N_0^t \\ s_1+ \cdots + s_t=s}} \ \prod_{i=1}^t \qbin{n_i}{s_i}{q} \prod_{j=0}^{s_i-1} (q^{m_i} - q^j).
		\qedhere
	\end{align*}
\end{proof}

\begin{proof}[Proof of Theorem~\ref{th:prpack}]
Consider the code~$C$ given in Theorem~\ref{th:prpack} and recall the isomorphism~$\psi$ from~\eqref{e-Psi}.
Deleting from all matrices in $\psi(C)$ the first $d-3$ rows provides us with a code $\tilde{C}\subseteq\F_q^{(N-d+3)\times M}$.
In terms of matrix tuples $(X_1,\ldots,X_t)\in C$, this means that we delete a total of $d-3$ rows of the individual matrices 
starting from the left.
As a result $X_1,\ldots,X_{\ell}$ are entirely deleted and the topmost $n_{\ell+1}-\delta$ rows of $X_{\ell+1}$ are deleted, where~$\ell$ 
and~$\delta$ are as in Theorem~\ref{th:prpack}.
Thus $C':=\psi^{-1}(\tilde{C})$ is a sum-rank metric code in the space
$\Pi' :=\Pi_q((n_{\ell+1}-\delta) \times m_{\ell+1}\mid\cdots\mid n_{t} \times m_{t})$.

Since the code~$C$ has sum-rank distance~$d$, the map $C\longrightarrow C'$ is injective and thus
$|C'|=|C|$.
Moreover, $\srk(C') \ge 3$ because deleting one row in one block causes a drop of at most one in the sum-rank distance. 
We can therefore apply the sphere-packing bound of Theorem \ref{th:pack} to~$C'$ with distance~$3$ and arrive at the desired bound
\[
        |C| = |C'| \le \frac{|\Pi'|}{V_1(\Pi')}
         = \frac{q^{\sum_{i=\ell+1}^{t} n_i m_i -\delta m_{\ell+1}}}{1+ \sum_{i=\ell+2}^{t} (q^{n_{i}}-1)\qbinsmall{m_i}{1}{q} + (q^{n_{\ell+1}-\delta}-1)\qbinsmall{m_{\ell+1}}{1}{q}}.
        \qedhere
\]
\end{proof}

\begin{proof}[Proof of Theorem~\ref{th:totalwt}] 
Define the total distance of~$C$, that is,
\[
    S:=\sum_{\substack{X,Y\in C \\ X\neq Y}} \srk(X-Y)= \sum_{i=1}^t \, \sum_{\substack{X,Y\in C \\ X\neq Y}}  \rk(X_i-Y_i).
\]
We have $S \geq d|C|(|C|-1)$. 
In order to derive an upper bound for~$S$ let
\begin{align*}
   K_{i,1}&:=|\{(X,Y) \in C^2 \mid \rk(X_i-Y_i) = n_i\}|,\\
   K_{i,2}&:=|\{(X,Y) \in C^2 \mid X\neq Y, \, \rk(X_i-Y_i) \leq n_i - 1\}|.
\end{align*}
Then $K_{i,1}+K_{i,2} = |C|(|C|-1)$ for each $i \in [t]$. 
Furthermore,
\begin{equation}\label{e-SBound}
    S \leq \sum_{i=1}^t (n_i K_{i,1} +(n_i-1)K_{i,2}) =   \sum_{i=1}^t\big(n_i|C|(|C|-1) - K_{i,2}\big)
          =  N|C|(|C|-1) - \sum_{i=1}^t K_{i,2}.
\end{equation}
For our goal to derive a lower bound on $K_{i,2}$ we need the following notation.
For $X=(X_1,...,X_t)$ write $X_i^1$ for the first row of $X_i$.
Then distinct $X,Y \in C$ clearly satisfy $\rk(X_i-Y_i) \leq  n_i -1$ whenever $X_i^1=Y_i^1$. 
Therefore, the number of pairs of distinct $X,Y \in C$ such that $X_i^1=Y_i^1$ is a lower 
bound for the number of such pairs satisfying $\rk(X_i-Y_i) \leq  n_i -1$. 
Setting for $A \in \F_q^{1 \times m_i}$
\[
    N_{i,A}:=|\{ X \in C \mid X_i^1 =A\}|,
\]
we arrive at
\[
      \sum_{A \in \F_q^{1\times m_i}} N_{i,A} = |C|\ \text{ and }\
K_{i,2} \geq \displaystyle \sum_{A \in \F_q^{1\times m_i}} N_{i,A}(N_{i,A}-1).
\]
Together with~\eqref{e-SBound} this gives us
\begin{eqnarray*}
 d|C|(|C|-1) \leq S(C) &\leq & N|C|(|C|-1) - \sum_{i=1}^t K_{i,2} \\ &\leq& N|C|(|C|-1) - \sum_{i=1}^t \sum_{A \in \F_q^{1\times m_i}} N_{i,A}(N_{i,A}-1),\\
       & = & N|C|(|C|-1) - \sum_{i=1}^t \sum_{A \in \F_q^{1 \times m_i}} N_{i,A}^2 + \sum_{i=1}^t \sum_{A \in \F_q^{1\times m_i}} N_{i,A}, \\
       & = & N|C|(|C|-1) + t|C| - \sum_{i=1}^t \sum_{A \in \F_q^{1\times m_i}} N_{i,A}^2, \\
       & \leq & N|C|(|C|-1) + t |C| - \sum_{i=1}^t q^{-m_i} \bigg(\sum_{A \in \F_q^{1\times m_i}} N_{i,A}\bigg)^2 ,\\
       & = & N|C|(|C|-1) + t |C| - |C|^2 Q,
\end{eqnarray*}
where the penultimate step follows from the Cauchy-Schwarz inequality.
Thus we have $d(|C|-1) \leq  N(|C|-1) + t - |C| Q$, 
which in turn yields
\[
    |C| \leq \frac{d - N + t}{d-N + Q},
\]
as long as $d>N - Q$. 
\end{proof}

\begin{proof}[Proof of Theorem~\ref{th:cov}] 
Towards a contradiction, suppose that the largest dimension, $k$, of a linear code $C \le \Pi$ with sum-rank distance $d$ 
satisfies 
\[
q^k<\frac{|\Pi|}{V_{d-1}(\Pi)}.
\]
Then there must exist $X \in \Pi$ such that $\srk(X-Y) \ge d$ for all $Y \in C$. 
It is easy to see that $C \oplus \langle X \rangle \le \Pi$ is a code of sum-rank distance $d$ that strictly contains $C$. 
This contradicts the maximality of~$C$.
\end{proof}

\subsection{Comparisons} \label{sub:compar}
In this subsection we provide for each of the bounds, except for the Induced Hamming Bound, an 
example of an ambient space~$\Pi$ and a distance~$d$ for which that bound is superior to all other bounds.
This is true for linear as well as general sum-rank metric codes.
In the case where $m=m_1= \cdots =m_t$, the Induced Hamming Bound (Theorem~\ref{th:induced}) is clearly inferior 
to the sphere-packing bound (Theorem~\ref{th:pack}) because the Hamming sphere of radius~$r$ in $\F_{q^{m}}^N$ is contained in 
the sum-rank sphere of the same radius in~$\Pi$, where we identify $\F_{q^{m}}^N$ and $\Pi$ via the map~$f$ defined in~\eqref{def_f}.
For the case of general~$m_i$, the embedding~$f$ and subsequent comparison between the Hamming weight and sum-rank weight
suggests that in fact the Sphere-Packing Bound is always superior to the Induced Hamming Bound.

The examples below show, in contrast to MRD codes, that MSRD codes do not exist for all possible parameter 
sets.
This is not surprising because the sum-rank is a hybrid of the Hamming distance and the rank distance and thus, 
just like MDS codes, MSRD codes may not exist over arbitrary fields.

\begin{example}\label{E-Sing-PSP-TW}
Let $\F=\F_2,\,t=13$, and 
$\Pi=\Pi_2(2\times2\mid\underbrace{1\times2\mid\cdots\mid1\times2}_{7\text{ blocks}}\mid\underbrace{1\times1\mid\cdots\mid1\times1}_{5\text{ blocks}})$.
We have the following data for various values~$d$ of the sum-rank distance (where we provide the floor function of the values). 
The best bound is given in bold and is unique in each case. 

\begin{center}
\begin{tabular}{|c|c|c|c|c|}
\hline
$d$ & 8 &9 &11\\ \hline\hline
\text{Singleton}& 512 & \textbf{128} & 16\\ 
\text{Induced Plotkin} & -- & -- &  22 \\
\text{Induced Elias} &9748 & 2036 & 43 \\
\text{Sphere-Packing} & 1502 & 232 & 50 	\\
\text{ Proj.\ Sphere-Packing} & \textbf{455} & 136 & 14 \\
\text{Total-Distance} & -- & -- & \textbf{6} \\ \hline
\end{tabular}
\end{center}

\noindent For $d=9$ the best bound is the Singleton Bound.
Later with the aid of Corollary~\ref{exclMW}, however, we will see that there is no linear code meeting that bound;
see Example~\ref{E-Sing-PSP-TW2}.
Finally note that for $d=8$ and $d=11$ the best bounds also provide the unique best bound for linear codes.
\end{example}

\begin{example}\label{E-SP-IndPl-Elias}
Let $\F=\F_2$. Let $\Pi=\Pi_2(2\times 2\mid\cdots\mid2\times 2)$ with~$t$ blocks, and where~$t$ and the rank-distance~$d$ 
are specified in the table.
We have the following data.

\begin{center}
\begin{tabular}{|c|c|c|c|c|c|c|}
\hline
$t/d$ & 4/5 & 6/8 &7/10 &9/14 &17/32\\ \hline\hline
\text{Singleton}& 256 &1024 & 1024 & 1024 & 64\\ 
\text{Induced Plotkin} &-- & -- & -- &  \textbf{28} & \textbf{4} \\
\text{Induced Elias} &366 &721 & \textbf{391} & 56 & 10\\
\text{Sphere-Packing}& \textbf{119} & 958 & 863 & 833 &418	\\
\text{ Proj.\ Sphere-Packing} & 146 &\textbf{528} & 528 & 528 & 46\\
\text{Total-Distance}& -- & -- & -- & -- & 6\\ \hline
\end{tabular}
\end{center}

\noindent For $t=4,7,9$ the unique best bound also provides the unique best bound for linear codes.
For $t=17$, the above tells us that a linear code~$C$ with distance $32$ satisfies $\dim(C) \le 2$.
A code attaining this bound (and in fact with minimum distance $2t$) is given by the repetition code
$\{(M,M,...,M) \mid M \in C'\}$, where $C' \le \F^{2 \times 2}$ is an MRD code of minimum rank distance~$2$.
\end{example}

We conclude this section with various linear codes that meet some of the bounds and thereby illustrate the tightness of the linear version of 
those bound.

The first two examples are MSRD codes. 
It is worth noting that in the first example, the number of blocks,~$t$, exceeds $|\F|+1$.
This stands in contrast to MDS codes where the length is believed to be upper bounded by $|\F|+1$ (minus some exceptional cases).
Later in Theorem~\ref{coro:nec} we will derive an upper bound for~$t$ for MSRD codes with certain parameters.

\begin{example}\label{E-MSRD8}
Consider $\Pi=\Pi_q(1\times 2\mid1\times 2\mid1\times 2\mid1\times 2\mid1\times 2\mid1\times 1\mid1\times 1\mid1\times 1)$.
Then the $3$-dimensional code
\[
  C=\left\langle 
\begin{array}{c}
    \big((1,0),\,(1,0),\,(1,0),\,(1,0),\,(1,0),\,(1),(0),(0)\big),\\[.7ex]
    \big((1,0),\,(0,1),\,(0,1),\,(0,1),\,(0,1),\,(0),(1),(0)\big),\\[.7ex]
   \big((0,1),\,(1,0),\,(0,1),\,(1,1),\,(1,1),\,(0),(0),(1)\big)\;\mbox{}
\end{array}\right\rangle\leq\Pi
\]
is an MSRD code of sum-rank distance~$6$ for any finite field~$\F$.  
For $\F=\F_2$ the code  also attains the linear version of the Total-Distance Bound and of the Projective 
Sphere-Packing Bound, whereas the linear version of any other bound is at least~$5$.    
Later in Example~\ref{E-MSRD111Ext} we will provide a generalization of this construction.
\end{example}

\begin{example}\label{E-MSRD6}
Consider $\F=\F_3$ and $\Pi=\Pi_3(2\times2\mid1\times2\mid1\times2\mid1\times2)$. Then the code
\[
 C=\left\langle 
      \begin{array}{ll}
      \bigg(\SmallMat{2}{2}{1}{0},\,(2,1),\,(1,0),\,(0,0)\bigg),&
      \bigg(\SmallMat{0}{2}{2}{1},\,(1,1),\,(0,1),\,(0,0)\bigg),\\[2ex]
      \bigg(\SmallMat{1}{0}{1}{2},\,(2,1),\,(0,0),\,(1,0)\bigg),&
      \bigg(\SmallMat{2}{1}{1}{0},\,(1,0),\,(0,0),\,(0,1)\bigg)\\[2ex]
\end{array}\right\rangle\leq\Pi
\]
is an MSRD code with sum-rank distance $d=4$. 
The code  also attains the linear version of the Projective Sphere-Packing Bound, whereas any other bound is at least~$5$.
A non-existence criterion for MSRD codes derived later in this paper (Corollary~\ref{exclMW}) shows that for $\F=\F_2$ no such MSRD code exists.
\end{example}

\newpage

\begin{example}\label{E-MSRD7}
Consider $\F=\F_2$ and $\Pi=\Pi_2(2\times2\mid2\times2\mid1\times2\mid1\times2)$.
Then the linear code $C\leq \Pi$ generated by the~$7$ matrix tuples
\[
      \begin{array}{ll}
      \bigg(\SmallMat{0}{1}{1}{0},\,\SmallMat{0}{0}{0}{0},\,(1,0),\,(0,0)\bigg),&
      \bigg(\SmallMat{1}{0}{0}{0},\,\SmallMat{1}{0}{1}{0},\,(0,1),\,(0,0)\bigg),\\[2ex]
      \bigg(\SmallMat{0}{1}{0}{1},\,\SmallMat{1}{1}{0}{0},\,(0,0),\,(1,0)\bigg),&
      \bigg(\SmallMat{0}{1}{1}{0},\,\SmallMat{0}{0}{0}{1},\,(0,0),\,(0,1)\bigg),\\[2ex]
      \bigg(\SmallMat{0}{1}{0}{0},\,\SmallMat{1}{1}{1}{0},\,(0,0),\,(0,0)\bigg),&
      \bigg(\SmallMat{1}{0}{0}{1},\,\SmallMat{0}{1}{1}{0},\,(0,0),\,(0,0)\bigg),\\[2ex]
      \bigg(\SmallMat{1}{1}{1}{0},\,\SmallMat{0}{0}{1}{1},\,(0,0),\,(0,0)\bigg)&
   \end{array}
\]
attains the linear version of the Sphere-Packing Bound for $d=3$ (and thus also of the Projective Sphere-Packing Bound), 
while all other bounds are at least~$8$ (in the linear version).
\end{example}

In Section~\ref{S-Constr} we will return to some of these examples and provide generalizations.

\section{Asymptotic Bounds}\label{S-Asymp}
In this section we derive asymptotic versions for the bounds presented in Section~\ref{sec:bounds} 
as $t\rightarrow\infty$, where as usual,~$t$ is the number of blocks.
As the reader will see, the Singleton Bound of Theorem~\ref{th:singl}, the Projective Sphere-Packing Bound of Theorem~\ref{th:prpack} and the Induced Singleton Bound of Theorem \ref{th:induced} are all asymptotically the same. The Sphere-Packing Bound of Theorem \ref{th:singl} and the Total-Distance Bound of Theorem \ref{th:totalwt} are asymptotically the best upper bounds currently known.

Recall from~\eqref{e-nimi} that we always assume that $m_i\geq m_{i+1}$.
As a consequence, we will now impose that $(m_i)$ is a non-increasing sequence of positive integers, which in turn
implies that the sequence stabilizes. 

Throughout this section we fix the following notation.

\begin{notation}\label{Nota}
	Let $(m_i)_{i\in\N},\,(n_i)_{i\in\N}$ be sequences with the following properties and further data. 
	\begin{romanlist}
		\item $1\leq n_i\leq m_i$ and $n_i,\,m_i\in\N$ for all $i\in\N$.
		\item The sequence $(m_i)$ is monotonically non-increasing. 
		\item Let $s\in\N$ be the least integer such that $m_i=m_s$ for all $i>s$. We write $\hat{m}:=m_s$. 
		\item Set $\nmax=\max\{n_i\mid i>s\}$ and $\nmin=\min\{n_i\mid i>s\}$.
		\item For any $t\in\N$ we define $\Pi^{(t)}=\Pi_q(n_1\times m_1\mid\cdots\mid n_t\times m_t)$ and 
		\[
		N_t=\sum_{i=1}^t n_i,\quad D_t=\dim(\Pi^{(t)})=\sum_{i=1}^t n_i m_i.
		\]
	\end{romanlist}
\end{notation}

\begin{definition}\label{D-Aqd}
	Let $d\in\R_{\geq0}$. We define
	\[
	A_q(d):=A_q\big(t,(n_i),(m_i),d\big):=\max\big\{|C|\,\big|\, C\subseteq\Pi^{(t)},\, |C| \ge 2, \, \srk(C)\geq d\big\},
	\]
where the maximum over the empty set is taken to be 1.
	Furthermore, for any $\eta\in[0,1]$ we define
	\[
	\alpha_q(\eta):=\limsup_{t\rightarrow\infty}\frac{\log_q\big(A_q(\eta N_t)\big)}{D_t}.
	\]
	The quantity $\frac{\log_q(A_q(d))}{D_t}$ is the \textbf{rate} of an optimal code $C \subseteq \Pi^{(t)}$ with sum-rank distance at least~$d$, provided that such a code exists (and is zero otherwise).
\end{definition}

Clearly, the induced bounds from Theorem~\ref{th:induced} lead to the asymptotic bounds as known for the Hamming metric; 
see for instance~\cite[Sec.~2.10]{pless}.
For the sake of comparison, and the convenience of the reader, we summarize these below. 
\begin{remark}
	In applying the asymptotic induced upper bounds on the rate of a code for the sum-rank distance, the parameter $m$ that appears in the bounds denotes the maximal value of the $m_i$ (which given the assumptions that the $m_i$ are non-increasing, is $m_1$.) Recall that we identify a code in $\Pi^t$ with one in $\F_{q^m}^{N_t}$, as outlined in the proof of Theorem \ref{th:induced}. Therefore, in applying the asymptotic versions of the induced bounds we use only the datum $m_1$ provided by the sequences~$(m_i)$ and~$(n_i)$. 
\end{remark}
First recall that for any prime power $Q$, we denote by $h_Q(x)$ the Hilbert entropy function defined on $[0,1]$ by
$$h_Q(x) = 	\left\{ \begin{array}{ll} 0 & \text{ if } x = 0,\\
	x\log_Q(Q-1)-x\log_Q(x) -(1-x)\log_Q(1-x)& \text{ if } 0 < x < 1-Q^{-1}.
\end{array}\right.$$
It is well-known \cite[Lem.~2.10.3]{pless} that 
$h_Q(x) = \lim_{t \rightarrow \infty} t^{-1}V_{xt}(\F_Q^t)$, where  $V_{xt}(\F_Q^t)$ is the volume of the Hamming sphere (defined in Theorem~\ref{th:induced}).

\begin{theorem}[\textbf{Asymptotic Induced Bounds}]\label{th:asinduced}
Let $m=m_1$ and $\eta \in [0,1]$. We have the following upper bounds on $\alpha_q(\eta)$.
\\[1ex]
\mbox{}\hspace*{-.6em}\begin{tabular}{lp{11cm}}
  \textbf{Induced Singleton Bound:}&  $\alpha_q(\eta) \leq 1-\eta$.\\[1.3ex]
  \textbf{Induced Hamming Bound:}& $\alpha_q(\eta) \leq 1-h_{q^m}(\eta/2)$ for $\eta \in (0,1-q^{-m})$.\\[1.3ex]
   \textbf{Induced Plotkin Bound:}& $\alpha_q(\eta) \leq 
		\left\{ \begin{array}{ll} 1-\eta(1-q^{-m})^{-1} & \text{ if } \eta \in [0,1-q^{-m}],\\
			0 & \text{ otherwise. }
		\end{array}\right.$ \\[3ex]
  \textbf{Induced Elias Bound:}&
  $\alpha_{q}(\eta) \leq 1- h_{q^m}\big(r - \sqrt{r(r-\eta)}\,\big)$ for $\eta \in (0,r)$ where $r=1-q^{-m}$. 
\end{tabular}
\end{theorem}

We now focus on the other bounds of Section~\ref{sec:bounds} and derive asymptotic versions.

It would appear, intuitively, that the asymptotic bounds we derive shortly do not depend on the parameters of the first~$s$ blocks, where~$s$ is as in Notation~\ref{Nota}~(iii),
and therefore we should be able to assume that $(m_i)_{i\in\N}$ is a constant sequence. 
However, while it is straightforward to puncture on columns and retain information about the distance, the converse process is not obvious. 
Therefore, we will start with the general setup of Notation~\ref{Nota}, and the independence of the asymptotic bounds on the first~$s$ 
blocks will become clear only a posteriori.

Note that for any $t\geq s$ we have 
\begin{align}
    &D_s+\hat{m}\nmax(t-s)\geq D_t=D_s+\hat{m}(N_t-N_s)\geq D_s+\hat{m}\nmin(t-s)\label{e-Dtineq},\\[.7ex]
    &N_s+\nmax(t-s)\geq N_t\geq N_s+\nmin(t-s)\label{e-Ntineq}.
\end{align}
Using that $D_s$ and $N_s$ are constant, we thus have the following limits:
\begin{align}
	&\lim_{t\rightarrow\infty}\frac{D_t}{N_t}=\hat{m},\label{e-limDtNt}\\[.6ex]
	&\frac{1}{\hat{m}\nmax}\leq \liminf_{t\rightarrow\infty}\frac{t}{D_t}\leq
	\limsup_{t\rightarrow\infty}\frac{t}{D_t}\leq\frac{1}{\hat{m}\nmin}\label{e-limtDt},\\[.6ex]
	&\frac{1}{\nmax}\leq \liminf_{t\rightarrow\infty}\frac{t}{N_t}\leq
	\limsup_{t\rightarrow\infty}\frac{t}{N_t}\leq\frac{1}{\nmin}.\label{e-limtNt}
\end{align}

We are now ready to present the Asymptotic Singleton Bound.

\begin{theorem}[\textbf{Asymptotic Singleton Bound}]\label{T-ASympSingl}
	Let $\eta\in[0,1]$. Then 
	\[
	\alpha_q(\eta)\leq 1-\eta.
	\]
\end{theorem}

\begin{proof}
	Let $\eta>0$.
	We apply the Singleton Bound (Theorem \ref{th:singl}) to codes $C\subseteq\Pi^{(t)}$ with $\srk(C)\geq d=\lceil\eta N_t\rceil$ for any~$t$.
	Let $j(t)\in\N_0$ be the unique integer satisfying 
	\begin{equation}\label{e-intervaletaNt}
		\lceil\eta N_t\rceil-1\in[N_{j(t)-1},N_{j(t)}).
	\end{equation}
	Then \cref{th:singl} implies $\log_q(|C|)\leq\sum_{i=j(t)}^t m_in_i$. 
	Since $\lim_{t\rightarrow\infty}\eta N_t=\infty$, we also have $\lim_{t\rightarrow\infty}j(t)=\infty$. 
	Clearly, $j(t)\leq t$.
	Consider now~$t$ large enough so that $j(t)>s$. 
	With the aid of~\eqref{e-intervaletaNt} we arrive at
	\begin{align*}
		\log_q(|C|)&\leq \sum_{i=j(t)}^t m_in_i
		=\hat{m}(N_t-N_{j(t)}+n_{j(t)})\\
		&\leq\hat{m}(N_t-N_{j(t)}+\hat{m})\leq \hat{m}(N_t-\eta N_t+1+\hat{m}).
	\end{align*}
	Hence, applying \eqref{e-limDtNt}, we have:
	\[
	\alpha_q(\eta)\leq\limsup_{t\rightarrow\infty}\frac{\hat{m}(N_t-\eta N_t+1+\hat{m})}{D_t}
	=\limsup_{t\rightarrow\infty}\Big(\hat{m}(1-\eta)\frac{N_t}{D_t}+\frac{\hat{m}(1+\hat{m})}{D_t}\Big)=1-\eta.
	\qedhere
	\]
\end{proof}

We now turn to the Asymptotic Total-Distance Bound. We need the following lemma.
Recall Notation~\ref{Nota}.

\begin{lemma}\label{L-limQtNt}
Let $Q_t:=\sum_{i=1}^t q^{-m_i}$.
Then 
\[
   \frac{1}{\nmax q^{\hat{m}}}\leq \liminf_{t\rightarrow\infty} \frac{Q_t}{N_t}\leq
    \limsup_{t\rightarrow\infty} \frac{Q_t}{N_t} \leq \frac{1}{\nmin q^{\hat{m}}}.
\]
Furthermore, if there exist $s'\in\N$ and $\hat{n}\in\N$ such that  $n_i=\hat{n}$ for all $i>s'$, then 
$\lim_{t\rightarrow\infty} \frac{Q_t}{N_t} = \frac{1}{\hat{n}q^{\hat{m}}}$.
\end{lemma}

\begin{proof}
Using  $Q_t = Q_s+q^{-\hat{m}}(t-s) $ and~\eqref{e-Ntineq} 
we obtain
\begin{equation}\label{eq:QtNt}
	   \frac{Q_s+q^{-\hat{m}}(t-s)}{N_s+\nmax(t-s)} \leq \frac{Q_t}{N_t} 
	   \leq \frac{Q_s+q^{-\hat{m}}(t-s)}{N_s+\nmin(t-s)}.	
\end{equation}
Since $Q_s$ and $N_s$ are constant, this leads to the desired result.
Finally, if $n_i=\hat{n}$ for $i\geq s'$, the above proof with~$s'$ and~$\hat{n}$ in place of~$s$ and~$\nmin$ (or $\nmax$), respectively,
leads to equalities in (\ref{eq:QtNt}) and the stated result follows.
\end{proof}

\begin{theorem}[\textbf{Asymptotic Total-Distance Bound}]\label{T-ASympPlotkin}
Let $\eta\in[0,1]$. Then
\begin{alphalist}
\item If $\eta>1-\frac{1}{\nmax q^{\hat{m}}}$, then $\alpha_q(\eta)=0$.
\item If $\eta\leq 1-\frac{1}{\nmax q^{\hat{m}}}$, then 
	$\alpha_q(\eta)\leq 1 - \eta\big(1-\frac{1}{n^*q^{\hat{m}}}\big)^{-1}$.
\item Suppose there exists $s'\in\N$ such that  $n_i=\hat{n}$ for all $i>s'$.  Then 
	\[
	    \alpha_q(\eta)\leq 1 - \eta\big(1-\frac{1}{\hat{n}q^{\hat{m}}}\big)^{-1}
	     \text{ for all }\eta \leq 1 - \frac{1}{\hat{n}q^{\hat{m}}}
	      \ \ \text{ and }\ \
	      \alpha_q(\eta)=0\text{ otherwise.}
	\]
\end{alphalist}
\end{theorem}

\begin{proof}
(a) Let $\eta> 1-\frac{1}{\nmax q^{\hat{m}}}$ and let $Q_t:=\sum_{i=1}^t q^{-m_i}$.
Then $\eta > 1-\frac{Q_t}{N_t}$  for sufficiently large $t$ thanks to Lemma \ref{L-limQtNt}, 
and so $\eta N_t > N_t - Q_t$.
Therefore, we can apply \cref{th:totalwt} to a code $C\subseteq\Pi^{(t)}$ with minimum sum-rank distance at least 
$\eta N_t$ to obtain:
\[
     |C| \leq \frac{\eta N_t - N_t + t}{\eta N_t-N_t + Q_t} = \frac{\eta -1 + t/N_t}{\eta-1 + Q_t/N_t}. 
\]
Now $\limsup_{t\rightarrow\infty} |C| \leq \frac{\eta-1+\nmin^{-1}}{\eta -1 + (\nmax q^{\hat{m}})^{-1}}$ follows from
Lemma \ref{L-limQtNt} and~\eqref{e-limtNt}.
Thus, $\limsup_{t\rightarrow\infty}A_q(\eta N_t)$ is finite and hence $\alpha_q(\eta)=0$.
	
(b) Now suppose that $\eta\leq 1-\frac{1}{\nmax q^{\hat{m}}}$. 
Let $C\subseteq \Pi^{(t)}$ have $\srk(C)\geq\eta N_t$ and $|C|=A_q(\eta N_t)$.
We proceed in several steps.
We first derive a general form of puncturing that will allow us to apply the Total-Distance Bound. 
In a second step we specify the parameters for the puncturing such that the asymptotics of the bound can be determined.
The final computation is carried out in Sept~3.

\noindent\underline{Step 1: Deriving a Punctured Code:}
For given~$t$ choose any integers $r_1,\ldots,r_t$ such that $ 0\leq r_i\leq n_i$ and 
\begin{equation}\label{eq-ri}
		\sum_{i=1}^t(n_i-r_i)-\sum_{i: r_i<n_i}q^{-m_i}+1\leq \eta N_t .
\end{equation}
In this step we will puncture a subset the code~$C$ by removing~$r_i$ rows from block~$i$.
The specific choice of the parameters~$r_i$, depending on~$t$, will be made in Step~2 of this proof. 
For our general choice subject to~\eqref{eq-ri} define the index sets (which depend on~$t$)
\[
    \mI_1=\{i\in[t]\mid r_i>0\}\ \text{ and }\ \mI_2=\{i\in[t]\mid r_i<n_i\}
\]
and the spaces
\[
\Pi_1^{(t)}=\bigoplus_{i\in\mI_1}\F_q^{r_i\times m_i}\ \text{ and }\  \Pi_2^{(t)}=\bigoplus_{i\in\mI_2}\F_q^{(n_i-r_i)\times m_i}.
\]
Clearly, $\Pi_1^{(t)}\oplus\Pi_2^{(t)}\cong\Pi^{(t)}$.
Set $L_1=\sum_{i\in\mI_1}r_im_i$.
Consider the projections
\begin{align*}
        f_1:&\Pi^{(t)}\longrightarrow\Pi_1^{(t)},\quad (X_1,\ldots,X_t) \longmapsto (X'_i\mid i\in\mI_1) ,\\
        f_2:&\Pi^{(t)}\longrightarrow\Pi_2^{(t)},\quad (X_1,\ldots,X_t) \longmapsto (X''_i\mid i\in\mI_2),
\end{align*}
where~$X'_i$ consists of the first~$r_i$ rows of~$X_i$ and~$X''_i$ consists of the last $n_i-r_i$ rows of~$X_i$.
Since $|\Pi_1^{(t)}|=q^{L_1}$, there exists a matrix tuple $Y\in\Pi_1^{(t)}$ such that 
$C_Y:=\{X\in C\mid f_1(X)=Y\}$ has cardinality at least $|C|/q^{L_1}$. 
Set
\[
	C':=f_2(C_Y)\subseteq\Pi_2^{(t)}.
\]
Then $|C'|=|C_Y|\geq |C|/q^{L_1}$ and $\srk(C')\geq\eta N_t$. 
Note that the code~$C'$ has length~$|\mI_2|$ (and all blocks have a positive number of rows).
Write $N_{2,t}=\sum_{i\in\mI_2}n_i,\,R_{2,t}=\sum_{i\in\mI_2}r_i$, and $Q_{2,t}=\sum_{i\in\mI_2}q^{-m_i}$.
Then~\eqref{eq-ri} reads as $N_{2,t}-R_{2,t}-Q_{2,t}+1\leq\eta N_t$, which means
we can apply Theorem \ref{th:totalwt} to~$C'$.
We obtain
\begin{equation}\label{e-CardC}
	   |C|\leq |C'|q^{L_1}
	      \leq q^{L_1} \frac{\eta N_t-(N_{2,t}-R_{2,t})+|\mI_2|}{\eta N_t-(N_{2,t}-R_{2,t})+Q_{2,t}}
	     \leq q^{L_1} \big(\eta N_t-(N_{2,t}-R_{2,t})+|\mI_2|\big),
\end{equation}
where the last step follows from the fact that the denominator in the previous term is at least~$1$ thanks to~\eqref{eq-ri}.
Setting $R_t=\sum_{i=1}^t r_i$, we have $N_{2,t}-R_{2,t}=N_t-R_t$ and thus
\[
    \eta N_t-(N_{2,t}-R_{2,t})+|\mI_2|\leq(\eta -1)N_t + R_t +t\leq(\eta -1)N_t + N_t +t\leq \eta N_t +t.
\]
With the aid of~\eqref{e-limDtNt} and~\eqref{e-limtDt} we arrive at
\[
        \limsup_{t\rightarrow\infty} \frac{\eta N_t-(N_{2,t}-R_{2,t})+|\mI_2|}{D_t} \leq \limsup_{t\rightarrow\infty} \frac{\eta N_t +t}{D_t} 
        \leq \frac{1}{\hat{m}}\Big(\eta+\frac{1}{\nmin}\Big),
\]
which in turn yields $\limsup_{t\rightarrow\infty}\frac{\log_q(\eta N_t-(N_{2,t}-R_{2,t})+|\mI_2|) }{D_t}=0$.
Together with~\eqref{e-CardC} all of this shows that
\begin{align}
		\alpha_q(\eta)&=\limsup_{t\rightarrow\infty}\frac{\log_q(|C|)}{D_t}\leq 
		\limsup_{t\rightarrow\infty}\frac{\sum_{i\in\mI_1}^t r_im_i+\log_q(\eta N_t-(N_{2,t}-R_{2,t})+|\mI_2|) }{D_t}\nonumber\\
		&= \limsup_{t\rightarrow\infty}\frac{\sum_{i=1}^t r_im_i}{D_t}. \label{e-alphaq}
\end{align} 

\noindent\underline{Step 2: Determining Suitable Parameters $r_i$:}
We seek values for~$r_i$ satisfying~\eqref{eq-ri} and such that $\sum_{i=1}^t r_im_i$ in~\eqref{e-alphaq}
can be evaluated asymptotically.
At the same time we want~$r_i$ such that the sum is small. 
Since $(m_i)$ is non-increasing, this means to put more weight on~$r_i$ for large value of~$i$.
(Note that actually minimizing $\sum_{i=1}^t r_im_i$ subject to the constraint~\eqref{eq-ri} would be an instance of linear programming, 
which is much harder and may lead to results that cannot be evaluated asymptotically.)
We proceed as follows.
First set $T\in\N$ such that $\eta N_t\geq 1$ for all $t\geq T$.
Define, as before, $Q_t=\sum_{i=1}^tq^{-m_i}$. 
Since the sequence $(N_t-Q_t)$ is strictly increasing, there exist, for all $t\geq T$, a maximal index $z(t)\in[t]$ 
such that
\begin{equation}\label{e-ztmax}
   N_{z(t)-1}-Q_{z(t)-1}+1\leq\eta N_t
\end{equation}
(and where we set $N_0=Q_0=0$). 
The maximality of~$z(t)$ implies that $\lim_{t\rightarrow\infty}z(t)\rightarrow\infty$.
Observe that~\eqref{e-ztmax} forms an instance of~\eqref{eq-ri}. Indeed, it corresponds to the choice
\begin{equation}\label{e-richoice}
   r_i=\begin{cases} 0,&\text{for }i=1,\ldots,z(t)-1,\\ n_i,&\text{for }i=z(t)+1,\ldots,t,\end{cases}
\end{equation}
for which we have $\mI_2=[z(t)-1]$ and thus $Q_{2,t}=Q_{z(t)-1}$, as desired.

In preparation of the final part of the proof we need the lower bound
\begin{equation}\label{e-liminfNtzt}
     \liminf_{t\rightarrow\infty}\frac{N_{z(t)}}{N_t}\geq\eta\Big(1-\frac{1}{n^*q^{\hat{m}}}\Big)^{-1}.
\end{equation}
To establish this, note first that our assumption on~$\eta$ implies $\eta\big(1-(n^*q^{\hat{m}})^{-1}\big)^{-1}\leq1$, which
covers the case $z(t)=t$ in above inequality.
Let now $\mT=\{t\geq T\mid z(t)<t\}$ and suppose $|\mT|=\infty$.
Then for $t\in\mT$ the maximality of $z(t)$ subject to~\eqref{e-ztmax} implies 
$\eta N_t< N_{z(t)}-Q_{z(t)}+1$.
With the aid of Lemma~\ref{L-limQtNt} this yields
\[
  \liminf_{t\in\mT,\, t\rightarrow\infty}\frac{N_{z(t)}}{N_t}
  \geq\eta+\liminf_{t\in\mT,\, t\rightarrow\infty}\frac{Q_{z(t)}}{N_t}
  =\eta+\liminf_{t\in\mT, \, t\rightarrow\infty}\frac{Q_{z(t)}}{N_{z(t)}}\frac{N_{z(t)}}{N_t}
  \geq\eta+\frac{1}{n^*q^{\hat{m}}}\liminf_{t\in\mT,\, t\rightarrow\infty}\frac{N_{z(t)}}{N_t},
\]
from which $$\liminf_{t\in\mT,\, t\rightarrow\infty}\frac{N_{z(t)}}{N_t}\geq \eta\big(1-(n^*q^{\hat{m}})^{-1}\big)^{-1}$$ follows.
All of this establishes~\eqref{e-liminfNtzt}.

\underline{Step 3: Determining the Asymptotics:}
With the choice of $r_1,\ldots,r_t$ in~\eqref{e-richoice} we now return to~\eqref{e-alphaq}. 
For sufficiently large~$t$ we have
$\sum_{i=1}^t r_i m_i=\sum_{i=z(t)+1}^t n_im_i=\hat{m}(N_t - N_{z(t)})$,
and therefore~\eqref{e-alphaq} turns into 
\begin{align*}
		\alpha_q(\eta) &\leq  \limsup_{t\rightarrow\infty} \frac{\hat{m}(N_t - N_{z(t)})}{D_t}
		=   \limsup_{t\rightarrow\infty} \frac{N_t - N_{z(t)}}{N_t}\\
		& =  1 -\liminf_{t\rightarrow\infty} \frac{N_{z(t)}}{N_t}
		\leq  1 - \eta\Big(1-\frac{1}{n^*q^{\hat{m}}}\Big)^{-1},
\end{align*}
where the second step follows from~\eqref{e-limDtNt} and the last one from~\eqref{e-liminfNtzt}.
 
(c) Let now $n_i=\hat{n}$ for all $i>s'$.
Then in~\eqref{e-limtDt} and~\eqref{e-limtNt} we have $\lim_{t\rightarrow\infty}(t/D_t)=(\hat{m}\hat{n})^{-1}$ and
$\lim_{t\rightarrow\infty}(t/N_t)=(\hat{n})^{-1}$.
Thus if $\eta\geq 1 - \frac{1}{\hat{n}q^{\hat{m}}}$, the proof of (a) with~$\hat{n}$ instead of~$\nmax$ and~$\nmin$ 
results in $\alpha_q(\eta)=0$, and if $\eta< 1 - \frac{1}{\hat{n}q^{\hat{m}}}$ the proof of (b) with~$\hat{n}$ instead of 
$\nmax$ results in $\alpha_q(\eta)  \leq  1 - \eta\big(1-(\hat{n} q^{\hat{m}})^{-1}\big)^{-1}$, as desired.
\end{proof}

Next, we show that the Projective Sphere-Packing Bound provides the same asymptotic bound as the Singleton Bound.

\begin{theorem}[\textbf{Asymptotic Projective Sphere-Packing Bound}]\label{th:as_proj_sphere_pack}
Let $\eta \in [0,1]$. Then
\[
        \alpha_q(\eta) \leq 1-\eta.
\]
\end{theorem}

\begin{proof}
Let $C\subseteq \Pi^{(t)}$ have $\srk(C)\geq\eta N_t$ and $|C|=A_q(\eta N_t)$.
Consider~$t$ sufficiently large so that $\eta N_t\geq3$.
For these~$t$ let $\ell(t)\in\Z$ and $\delta(t) \in [0,n_{\ell(t)+1}-1]$ be the unique integers satisfying 
\begin{equation}\label{eq:Ntelldelta}
		\lceil \eta N_t\rceil -3 = \sum_{j=1}^{\ell(t)} n_j + \delta(t) = N_{\ell(t)}+\delta(t).
\end{equation}
Clearly, $\ell(t)\leq t$ and $\lim_{t\rightarrow\infty}\ell(t)=\infty$.
By Theorem \ref{th:prpack}, we have 
\begin{equation}\label{eq:aprojspb} 
\left.\begin{split}
    	\log_q(|C|)& \leq \sum_{j=\ell(t)+1}^t n_j m_j - \delta(t)m_{\ell(t)+1}\\[.7ex]
	 &\quad -\log_q\bigg( 1+\sum_{i=\ell(t)+2}^t (q^{n_i}-1) 
	            \frac{q^{m_i}-1}{q-1}+(q^{n_{\ell(t)+1}-\delta(t)} - 1)\frac{q^{m_{\ell(t)+1}}-1}{q-1}\bigg).
\end{split}\quad\right\}
\end{equation}	
We compute $\limsup \log_q(|C|)/D_t$ by considering the two terms individually. First,
\begin{align}
    	\limsup_{t\rightarrow\infty}  \frac{\sum_{j=\ell(t)+1}^t n_j m_j - \delta(t)m_{\ell(t)+1}}{D_t} 
    	&=\limsup_{t\rightarrow\infty}  \frac{D_t - D_{\ell(t)} - \delta(t)m_{\ell(t)+1}}{D_t},  \nonumber \\
    	&= 1 -\liminf_{t\rightarrow\infty} \frac{D_{\ell(t)}}{D_t} 
    	= 1 -\liminf_{t\rightarrow\infty} \frac{N_{\ell(t)}}{N_t} \\
	&\leq 1-\liminf_{t\rightarrow\infty}\frac{\eta N_t-3-\delta(t)}{N_t} = 1-\eta, \label{e-1-eta}
\end{align}
where the third step follows from \eqref{e-limDtNt} and the last two from \eqref{eq:Ntelldelta} and the fact that $\delta(t)$ is bounded.
As for the second term we first note that for sufficiently large $t$ 
\begin{equation}\label{e-second}
 \sum_{i=\ell(t)+2}^t (q^{n_i}-1) \frac{q^{m_i}-1}{q-1}+(q^{n_{\ell(t)+1}-\delta(t)} - 1)\frac{q^{m_{\ell(t)+1} }-1}{q-1} 
   	\geq \frac{q^{\hat{m}}-1}{q-1}\big( (q^{\nmin}-1)(t-\ell(t)-1)\big)
\end{equation}
since $n_{\ell(t)+1}-\delta(t)\geq0$.
The only term in the rightmost part depending on~$t$ is $t-\ell(t)$.
We show now that the sequence $(\frac{t-\ell(t)}{D_t})_{t\in\N}$ is bounded.
Using $N_{\ell(t)}\leq N_s+\nmax(\ell(t)-s)$ we obtain from~\eqref{eq:Ntelldelta} 
\[
  \ell(t)\geq \frac{\eta N_t-3-N_s-\delta(t)}{\nmax}+s.
\]
Therefore, 
\[
    \limsup_{t\rightarrow\infty} \frac{t-\ell(t)}{D_t} \leq 
    \limsup_{t\rightarrow\infty} \frac{t}{D_t} -\liminf_{t\rightarrow\infty}\bigg(\frac{\eta N_t-3-N_{s}-\delta(t)}{\nmax D_t}+\frac{s}{D_t}\bigg)
    \leq\frac{1}{\nmin\hat{m}}-\frac{\eta}{\nmax\hat{m}},
\]
thanks to~\eqref{e-limtDt} and~\eqref{e-limDtNt}.
All of this shows that $\liminf_{t\rightarrow\infty} \frac{\log_q (t-\ell(t))}{D_t} = 0$, and together 
with~\eqref{eq:aprojspb} --~\eqref{e-second} this establishes $\alpha_q(\eta)\leq 1-\eta$, as stated.
\end{proof}

As the proof has shown, the first term in~\eqref{eq:aprojspb} fully determines the asymptotic bound, see also~\eqref{e-1-eta}, and the second 
term does not provide an improvement.
This explains why the Asymptotic Sphere-Packing Bound is identical to the Asymptotic Singleton Bound.

We now turn to the Asymptotic Sphere-Packing and Sphere-Covering Bound. 
For this we need an asymptotic estimate for the cardinality of the spheres of sum-rank radius $t\rho$ as $t\longrightarrow \infty$; see 
Theorem~\ref{th:pack}. 
This has been established in the literature for more general alphabets and applies to our case if all matrix blocks have the same size.
Therefore we assume for the remainder of this section that $m_i=m$ and $n_i=n$ for all $i\in\N$
and furthermore set $\pi:=\F_{q}^{n \times m}$. Thus $\Pi=\pi^t$.
In~\cite{Loeliger94} the author considers a general alphabet~${\mathcal A}$ and a weight function~$w$ on~${\mathcal A}$ that extends to a weight function on ${\mathcal A}^t$ via $w(u_1,...,u_t) = \sum_{i=1}^t w(u_i)$. 
In the setting of this paper, we have ${\mathcal A}=\pi$ with the rank as weight function, and the induced weight function on~${\mathcal A}^t$ is 
precisely the sum-rank.

Now we are ready to state the following sum-rank adaptation of~\cite{Loeliger94} (see also \cite[Theorem 4.1]{grefsull04}).

\begin{theorem} \label{T-Loeliger}
Define the generating function
\[
     f(z):=\sum_{A \in \pi} z^{\rk(A)}  = \sum_{i=0}^n \qbin{n}{i}{q} \prod_{j=0}^{i-1}(q^m-q^j) z^i
\]
and the average rank weight $\epsilon :=\frac{1}{|\pi|} \sum_{A \in \pi} \rk(A)$.
Then for all $\rho \in (0,\epsilon]$ we have
\[
     \lim_{t \longrightarrow \infty} \frac{1}{t} \log_{|\pi|} \Big(V_{t \rho}(\pi^t)\Big) = \min_{z \in (0,1]} \log_{|\pi|}\Big( \frac{f(z)}{z^\rho} \Big).
\]
\end{theorem}

An entropy function for the sum-rank is hence defined as
\[
   H(\rho):=\min_{z \in (0,1]} \log_{|\pi|}\Big( \frac{f(z)}{z^\rho} \Big)= \min_{z \in (0,1]} \frac{1}{mn}\log_q\Big( \frac{f(z)}{z^\rho} \Big).
\]
The asymptotic bounds follow now easily.

\begin{corollary}[\textbf{Asymptotic Sphere-Packing Bound \& Sphere-Covering Bound}]
\label{cor:as_sphere_cov_pack}
Let $\eta \in (0,\frac{\epsilon}{n}]$, where
$\epsilon$ is as in Theorem~\ref{T-Loeliger}. Then
\[
     1-H(\eta n) \leq \alpha(\eta) \leq 1-H\left(\frac{\eta n}{2}\right).
\]
\end{corollary}

\begin{proof}
The bounds are immediate consequences of Theorem \ref{th:pack} and Theorem \ref{th:cov}:
\[
    1-H(\eta n)\!=\!\lim_{t\rightarrow\infty}  \frac{mnt-\log_q\big({V_{\eta nt-1}(\pi^t)}\big)}{mnt}\leq \alpha_q(\eta) 
    \leq \lim_{t\rightarrow\infty} \frac{mnt - \log_{q} \big(V_{\frac{\eta nt}{2}}(\pi^t)\big)}{mnt}\! =\!1-H\left(\frac{\eta n}{2}\right). 
\qedhere
\]
\end{proof}

\medskip

We close this section with graphical comparisons of the asymptotic bounds. 
In Figures~\ref{fig:newbds1} and~\ref{fig:newbds2} we compare all bounds other than the induced bounds. 
The upper and lower graphs labelled as Corollary \ref{cor:as_sphere_cov_pack} represent the Sphere-Packing and Sphere-Covering Bounds, respectively. In Figure \ref{fig:newbds1} we observe that the Asymptotic Total-Distance Bound of Theorem \ref{T-ASympPlotkin} is sharper than the Asymptotic Sphere-Packing Bound when $\eta$ exceeds~$0.35$, while this appears in Figure \ref{fig:newbds2} when $\eta$ exceeds $0.64$ (approximately).

\begin{figure}[h!]
	\centering
	\begin{tikzpicture}[scale=1]
		\begin{axis}[
			legend pos = outer north east,
			legend cell align={right},
			xmin=0, xmax=1,
			ymin=0, ymax=1,
			xtick={0,0.345,1},
			ytick={1},
			xmajorgrids=true,
			grid style=dashed,
			every axis plot/.append style={thick},
			xlabel={Value of $\eta \in [0,1]$},
			ylabel={Bound on $\alpha_2(\eta)$}
			]

      \addplot+[color=blue,style = solid,mark=.,mark size=0.1pt]
      coordinates {
      	(0.0000000000, 1.000000000)
      	(0.02000000000, 0.9793548387)
      	(0.04000000000, 0.9587096774)
      	(0.06000000000, 0.9380645161)
      	(0.08000000000, 0.9174193548)
      	(0.1000000000, 0.8967741936)
      	(0.1200000000, 0.8761290322)
      	(0.1400000000, 0.8554838710)
      	(0.1600000000, 0.8348387097)
      	(0.1800000000, 0.8141935484)
      	(0.2000000000, 0.7935483871)
      	(0.2200000000, 0.7729032258)
      	(0.2400000000, 0.7522580645)
      	(0.2600000000, 0.7316129032)
      	(0.2800000000, 0.7109677419)
      	(0.3000000000, 0.6903225806)
      	(0.3200000000, 0.6696774194)
      	(0.3400000000, 0.6490322581)
      	(0.3600000000, 0.6283870967)
      	(0.3800000000, 0.6077419355)
      	(0.4000000000, 0.5870967742)
      	(0.4200000000, 0.5664516129)
      	(0.4400000000, 0.5458064516)
      	(0.4600000000, 0.5251612903)
      	(0.4800000000, 0.5045161290)
      	(0.5000000000, 0.4838709677)
      	(0.5200000000, 0.4632258065)
      	(0.5400000000, 0.4425806451)
      	(0.5600000000, 0.4219354839)
      	(0.5800000000, 0.4012903226)
      	(0.6000000000, 0.3806451613)
      	(0.6200000000, 0.3600000000)
      	(0.6400000000, 0.3393548387)
      	(0.6600000000, 0.3187096774)
      	(0.6800000000, 0.2980645161)
      	(0.7000000000, 0.2774193548)
      	(0.7200000000, 0.2567741935)
      	(0.7400000000, 0.2361290323)
      	(0.7600000000, 0.2154838710)
      	(0.7800000000, 0.1948387097)
      	(0.8000000000, 0.1741935484)
      	(0.8200000000, 0.1535483871)
      	(0.8400000000, 0.1329032258)
      	(0.8600000000, 0.1122580645)
      	(0.8800000000, 0.09161290323)
      	(0.9000000000, 0.07096774194)
      	(0.9200000000, 0.05032258065)
      	(0.9400000000, 0.02967741936)
      	(0.9600000000, 0.009032258065)
      	(0.9700000000, 0.0000000000)
      	(0.9800000000, 0.0000000000)
      	(1.000000000, 0.0000000000)
      };
      \addplot+[color=green,style = solid,mark=.,mark size=0.1pt]
      coordinates {
      	(0.0000000000, 1.000000000)
      	(0.05000000000, 0.9500000000)
      	(0.1000000000, 0.9000000000)
      	(0.1500000000, 0.8500000000)
      	(0.2000000000, 0.8000000000)
      	(0.2500000000, 0.7500000000)
      	(0.3000000000, 0.7000000000)
      	(0.3500000000, 0.6500000000)
      	(0.4000000000, 0.6000000000)
      	(0.4500000000, 0.5500000000)
      	(0.5000000000, 0.5000000000)
      	(0.5500000000, 0.4500000000)
      	(0.6000000000, 0.4000000000)
      	(0.6500000000, 0.3500000000)
      	(0.7000000000, 0.3000000000)
      	(0.7500000000, 0.2500000000)
      	(0.8000000000, 0.2000000000)
      	(0.8500000000, 0.1500000000)
      	(0.9000000000, 0.1000000000)
      	(0.9500000000, 0.05000000000)
      	(1.000000000, 0.0000000000)
      };
      \addplot+[color=red,style = solid,mark=.,mark size=0.1pt]
      coordinates {
      	(0.0000000000, 0.9991901700)
      	(0.02000000000, 0.9685637886)
      	(0.04000000000, 0.9422219636)
      	(0.06000000000, 0.9178098059)
      	(0.08000000000, 0.8946812896)
      	(0.1000000000, 0.8725241256)
      	(0.1200000000, 0.8511523975)
      	(0.1400000000, 0.8304506247)
      	(0.1600000000, 0.8103254350)
      	(0.1800000000, 0.7907163009)
      	(0.2000000000, 0.7715741341)
      	(0.2200000000, 0.7528584477)
      	(0.2400000000, 0.7345357839)
      	(0.2600000000, 0.7165822887)
      	(0.2800000000, 0.6989739825)
      	(0.3000000000, 0.6816928074)
      	(0.3200000000, 0.6647217487)
      	(0.3400000000, 0.6480473066)
      	(0.3600000000, 0.6316565776)
      	(0.3800000000, 0.6155397433)
      	(0.4000000000, 0.5996862992)
      	(0.4200000000, 0.5840882179)
      	(0.4400000000, 0.5687377465)
      	(0.4600000000, 0.5536284161)
      	(0.4800000000, 0.5387541372)
      	(0.5000000000, 0.5241092766)
      	(0.5200000000, 0.5096889919)
      	(0.5400000000, 0.4954889795)
      	(0.5600000000, 0.4815050485)
      	(0.5800000000, 0.4677338167)
      	(0.6000000000, 0.4541717806)
      	(0.6200000000, 0.4408161337)
      	(0.6400000000, 0.4276641518)
      	(0.6600000000, 0.4147133074)
      	(0.6800000000, 0.4019615156)
      	(0.7000000000, 0.3894067669)
      	(0.7200000000, 0.3770472228)
      	(0.7400000000, 0.3648812665)
      	(0.7600000000, 0.3529074378)
      	(0.7800000000, 0.3411244609)
      	(0.8000000000, 0.3295311938)
      	(0.8200000000, 0.3181265938)
      	(0.8400000000, 0.3069097437)
      	(0.8600000000, 0.2958798424)
      	(0.8800000000, 0.2850362055)
      	(0.9000000000, 0.2743782725)
      	(0.9080000000, 0.2701669779)	
      };
      \addplot+[color=black,style = solid,mark=.,mark size=0.1pt]
      coordinates {
      	(0.0000000000, 0.9991901700)
      	(0.02000000000, 0.9422219636)
      	(0.04000000000, 0.8946812896)
      	(0.06000000000, 0.8511523975)
      	(0.08000000000, 0.8103254350)
      	(0.1000000000 ,0.7715741341)
      	(0.1200000000, 0.7345357839)
      	(0.1400000000, 0.6989739825)
      	(0.1600000000, 0.6647217487)
      	(0.1800000000, 0.6316565776)
      	(0.2000000000, 0.5996862992)
      	(0.2200000000, 0.5687377465)
      	(0.2400000000, 0.5387541372)
      	(0.2600000000, 0.5096889919)
      	(0.2800000000, 0.4815050485)
      	(0.3000000000, 0.4541717806)
      	(0.3200000000, 0.4276641518)
      	(0.3400000000, 0.4019615156)
      	(0.3600000000, 0.3770472228)
      	(0.3800000000, 0.3529074378)
      	(0.4000000000, 0.3295311938)
      	(0.4200000000, 0.3069097437)
      	(0.4400000000, 0.2850362055)
      	(0.4600000000, 0.2639055542)
      	(0.4800000000, 0.2435143024)
      	(0.5000000000, 0.2238603634)
      	(0.5200000000, 0.2049430707)
      	(0.5400000000, 0.1867631133)
      	(0.5600000000, 0.1693225325)
      	(0.5800000000, 0.1526248329)
      	(0.6000000000, 0.1366750857)
      	(0.6200000000, 0.1214800942)
      	(0.6400000000, 0.1070486347)
      	(0.6600000000, 0.09339175001)
      	(0.6800000000, 0.08052312621)
      	(0.7000000000, 0.06845957069)
      	(0.7200000000, 0.05722161634)
      	(0.7400000000, 0.04683430142)
      	(0.7600000000, 0.03732818204)
      	(0.7800000000, 0.02874068021)
      	(0.8000000000, 0.02111791782)
      	(0.8200000000, 0.01451728138)
      	(0.8400000000, 0.009011130812)
      	(0.8600000000, 0.004692375583)
      	(0.8800000000, 0.001683275179)
      	(0.9000000000, 0.0001502111691)
      	(0.9080000000, 0.00000009)
      };
      
      \addplot+[color=black, mark=.,mark size=0.1pt, style=densely dotted]
      coordinates {
      	(0.3450,0.000)
      	(0.3450,0.6480473066)
      };
      \addplot+[color=black, mark=.,mark size=0.1pt, style=densely dotted]
      coordinates {
      	(0.000,0.6480473066)
      	(0.3450,0.6480473066)
      };
      \addplot+[color=black, mark=.,mark size=0.1pt, style=densely dotted]
      coordinates {
      	(0.908,0.000)
      	(0.908,1.000)
      };
      \addplot+[color=black, mark=,mark size=0.1pt, style=densely dotted]
      coordinates {
      	(0.96875,0.000)
      	(0.96875,1.000)
      };
      \addplot+[color=black, mark=.,mark size=0.1pt, style=densely dotted]
      coordinates {
      	(0.99902,0.000)
      	(0.99902,1.000)
      };

\legend{\small{Th. \ref{T-ASympPlotkin}},\small{Cor. \ref{T-ASympSingl} \& Th. \ref{th:as_proj_sphere_pack}},
	\small{Cor. \ref{cor:as_sphere_cov_pack}},\small{Cor. \ref{cor:as_sphere_cov_pack}}
}
\end{axis}
\end{tikzpicture}
\caption{\label{fig:newbds1} 
Comparison of bounds on $\alpha_2(\eta)$ for the sequences 
$(m_i) = (4,4,...), (n_i) = (2,2,...)$.\\
}
\end{figure}
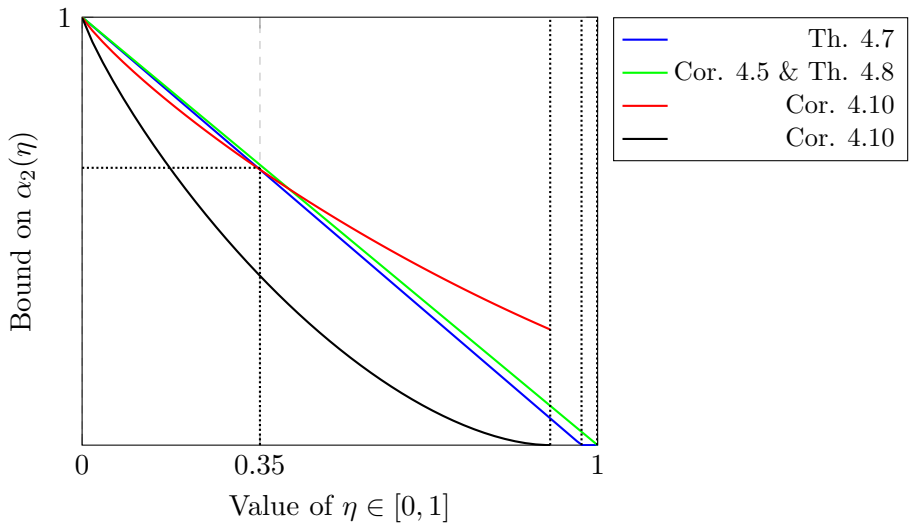

\bigskip

\begin{figure}[h!]
	\centering
	\begin{tikzpicture}[scale=1]
		\begin{axis}[
			legend pos = outer north east,
			legend cell align={right},
			xmin=0, xmax=1,
			ymin=0, ymax=1,
			xtick={0,0.635,1},
			ytick={1},
			xmajorgrids=true,
			grid style=dashed,
			every axis plot/.append style={thick},
			xlabel={Value of $\eta \in [0,1]$},
			ylabel={Bound on $\alpha_2(\eta)$}
			]
			
			\addplot+[color=blue,style = solid,mark=.,mark size=0.1pt]
			coordinates {
				(0.0000000000, 1.000000000)
				(0.02000000000, 0.9796825397)
				(0.04000000000, 0.9593650794)
				(0.06000000000, 0.9390476191)
				(0.08000000000, 0.9187301587)
				(0.1000000000, 0.8984126984)
				(0.1200000000, 0.8780952381)
				(0.1400000000, 0.8577777778)
				(0.1600000000, 0.8374603175)
				(0.1800000000, 0.8171428571)
				(0.2000000000, 0.7968253968)
				(0.2200000000, 0.7765079365)
				(0.2400000000, 0.7561904762)
				(0.2600000000, 0.7358730159)
				(0.2800000000, 0.7155555555)
				(0.3000000000, 0.6952380952)
				(0.3200000000, 0.6749206349)
				(0.3400000000, 0.6546031746)
				(0.3600000000, 0.6342857143)
				(0.3800000000, 0.6139682540)
				(0.4000000000, 0.5936507937)
				(0.4200000000, 0.5733333333)
				(0.4400000000, 0.5530158730)
				(0.4600000000, 0.5326984127)
				(0.4800000000, 0.5123809524)
				(0.5000000000, 0.4920634921)
				(0.5200000000, 0.4717460317)
				(0.5400000000, 0.4514285714)
				(0.5600000000, 0.4311111111)
				(0.5800000000, 0.4107936508)
				(0.6000000000, 0.3904761905)
				(0.6200000000, 0.3701587302)
				(0.6350000000, 0.3549206349)
				(0.6400000000, 0.3498412698)
				(0.6600000000, 0.3295238095)
				(0.6800000000, 0.3092063492)
				(0.7000000000, 0.2888888889)
				(0.7200000000, 0.2685714286)
				(0.7400000000, 0.2482539683)
				(0.7600000000, 0.2279365079)
				(0.7800000000, 0.2076190476)
				(0.8000000000, 0.1873015873)
				(0.8200000000, 0.1669841270)
				(0.8400000000, 0.1466666667)
				(0.8600000000, 0.1263492063)
				(0.8800000000, 0.1060317460)
				(0.9000000000, 0.08571428571)
				(0.9200000000, 0.06539682540)
				(0.9400000000, 0.04507936508)
				(0.9600000000, 0.02476190476)
				(0.9800000000, 0.004444444444)
				(0.9850000000, 0.0000000000)
				(1.000000000, 0.0000000000)
			};
			\addplot+[color=green,style = solid,mark=.,mark size=0.1pt]
			coordinates {
				(0.0000000000, 1.000000000)
				(0.05000000000, 0.9500000000)
				(0.1000000000, 0.9000000000)
				(0.1500000000, 0.8500000000)
				(0.2000000000, 0.8000000000)
				(0.2500000000, 0.7500000000)
				(0.3000000000, 0.7000000000)
				(0.3500000000, 0.6500000000)
				(0.4000000000, 0.6000000000)
				(0.4500000000, 0.5500000000)
				(0.5000000000, 0.5000000000)
				(0.5500000000, 0.4500000000)
				(0.6000000000, 0.4000000000)
				(0.6500000000, 0.3500000000)
				(0.7000000000, 0.3000000000)
				(0.7500000000, 0.2500000000)
				(0.8000000000, 0.2000000000)
				(0.8500000000, 0.1500000000)
				(0.9000000000, 0.1000000000)
				(0.9500000000, 0.05000000000)
				(1.000000000, 0.0000000000)
			};
			\addplot+[color=red,style = solid,mark=.,mark size=0.1pt]
			coordinates {
				(0.0000000000, 0.9979878116)
				(0.02000000000, 0.9652864923)
				(0.04000000000, 0.9356936247)
				(0.06000000000, 0.9081005868)
				(0.08000000000, 0.8818384873)
				(0.1000000000, 0.8565866759)
				(0.1200000000, 0.8321875692)
				(0.1400000000, 0.8084978859)
				(0.1600000000, 0.7854541956)
				(0.1800000000, 0.7629775825)
				(0.2000000000, 0.7410229440)
				(0.2200000000, 0.7195587615)
				(0.2400000000, 0.6985478139)
				(0.2600000000, 0.6779704938)
				(0.2800000000, 0.6578013467)
				(0.3000000000, 0.6380276462)
				(0.3200000000, 0.6186297358)
				(0.3400000000, 0.5995989586)
				(0.3600000000, 0.5809226510)
				(0.3800000000, 0.5625913579)
				(0.4000000000, 0.5445965053)
				(0.4200000000, 0.5269321959)
				(0.4400000000, 0.5095905588)
				(0.4600000000, 0.4925671917)
				(0.4800000000, 0.4758557755)
				(0.5000000000, 0.4594529975)
				(0.5200000000, 0.4433546646)
				(0.5400000000, 0.4275562278)
				(0.5600000000, 0.4120554143)
				(0.5800000000, 0.3968487454)
				(0.6000000000, 0.3819336668)
				(0.6200000000, 0.3673079383)
				(0.6350000000, 0.3565271273)
				(0.6400000000, 0.3529693600)
				(0.6600000000, 0.3389159332)
				(0.6800000000, 0.3251459812)
				(0.7000000000, 0.3116581172)
				(0.7200000000, 0.2984508573)
				(0.7400000000, 0.2855230761)
				(0.7600000000, 0.2728739310)
				(0.7800000000, 0.2605024946)
				(0.7900000000, 0.2544207454)
				(0.7971380000, 0.2501218860)	
			};
			\addplot+[color=black,style = solid,mark=.,mark size=0.1pt]
			coordinates {
				(0.0000000000, 0.9979878116)
				(0.005000000000, 0.9813775916)
				(0.02000000000, 0.9356936247)
				(0.04000000000, 0.8818384873)
				(0.06000000000, 0.8321875692)
				(0.08000000000, 0.7854541956)
				(0.1000000000, 0.7410229440)
				(0.1200000000, 0.6985478139)
				(0.1400000000, 0.6578013467)
				(0.1600000000, 0.6186297358)
				(0.1800000000, 0.5809226510)
				(0.2000000000, 0.5445965053)
				(0.2200000000, 0.5095905588)
				(0.2400000000, 0.4758557755)
				(0.2400000000, 0.4758557755)
				(0.2600000000, 0.4433546646)
				(0.2800000000, 0.4120554143)
				(0.3000000000, 0.3819336668)
				(0.3200000000, 0.3529693600)
				(0.3400000000, 0.3251459812)
				(0.3600000000, 0.2984508573)
				(0.3800000000, 0.2728739310)
				(0.4000000000, 0.2484081597)
				(0.4200000000, 0.2250488791)
				(0.4400000000, 0.2027936717)
				(0.4600000000, 0.1816422104)
				(0.4800000000, 0.1615960553)
				(0.5000000000, 0.1426585639)
				(0.5200000000, 0.1248348018)
				(0.5400000000, 0.1081315819)
				(0.5600000000, 0.09255756540)
				(0.5800000000, 0.07812342301)
				(0.6000000000, 0.06484204925)
				(0.6200000000, 0.05272875835)
				(0.6400000000, 0.04180143034)
				(0.6600000000, 0.03208055468)
				(0.6800000000, 0.02358913905)
				(0.7000000000, 0.01635248789)
				(0.7200000000, 0.01039791105)
				(0.7400000000, 0.005754503137)
				(0.7600000000, 0.002453214555)
				(0.7800000000, 0.0005275010607)
				(0.7900000000, 0.00009198052045)
				(0.7971380000, 0.0000000000)
			};
			\addplot+[color=black, mark=.,mark size=0.1pt, style=densely dotted]
			coordinates {
				(0.797138,0.000)
				(0.797138,1.000)
			};
			\addplot+[color=black, mark=.,mark size=0.1pt, style=densely dotted]
			coordinates {
				(0.63500000000,0.000)
				(0.63500000000,0.3565271273)
			};
			\addplot+[color=black, mark=.,mark size=0.1pt, style=densely dotted]
			coordinates {
				(0.000,0.3565271273)
				(0.63500000000,0.3565271273)
			};
			\addplot+[color=black, mark=,mark size=0.1pt, style=densely dotted]
			coordinates {
				(0.984375,0.000)
				(0.984375,1.000)
			};
			\addplot+[color=black, mark=.,mark size=0.1pt, style=densely dotted]
			coordinates {
				(0.99902,0.000)
				(0.99902,1.000)
			};
			\legend{\small{Th. \ref{T-ASympPlotkin}},\small{Cor. \ref{T-ASympSingl} \& Th. \ref{th:as_proj_sphere_pack}},
				\small{Cor. \ref{cor:as_sphere_cov_pack}},\small{Cor. \ref{cor:as_sphere_cov_pack}}
			}
		\end{axis}
	\end{tikzpicture}
	\caption{\label{fig:newbds2} 
	Comparison of bounds on $\alpha_2(\eta)$ for the sequences 
		$(m_i)=(n_i)=(4,4,...)$.\\
	}
\end{figure}
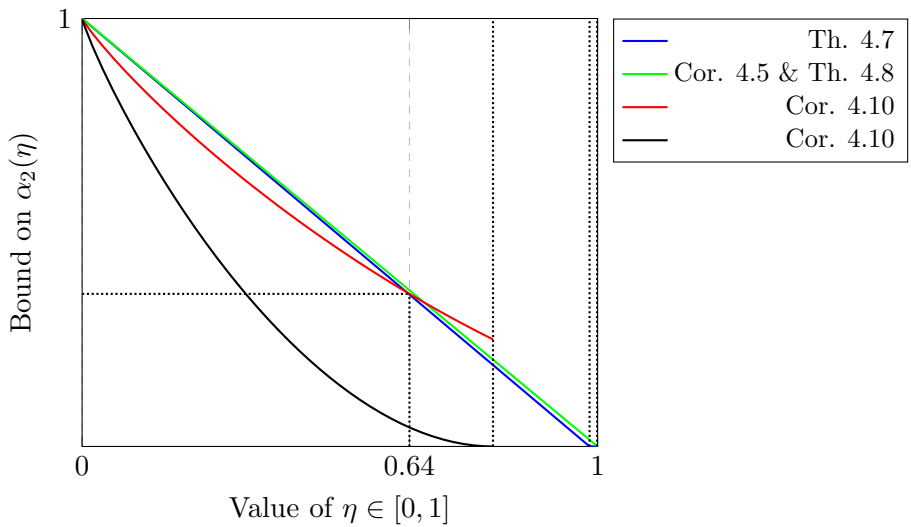

In Figures \ref{fig:allbds1} and \ref{fig:allbds2}, we compare all bounds, where it makes sense to do so. 
In Figure \ref{fig:allbds1}, the Asymptotic Total-Distance Bound is sharpest among the bounds that can be applied for sequences $(m_i)$ and $(n_i)$ satisfying $m_1=10, \hat{m}=4$, and such that $(n_i)$ converges to 2.
In Figure \ref{fig:allbds2}, we assume $m_i=n_i=4$ for each $i$, in which case we see that the Asymptotic Sphere-Packing Bound yields the sharpest bound up to $\eta \approx 0.52$, while when $\eta$ exceeds this value the Induced Plotkin Bound is best.

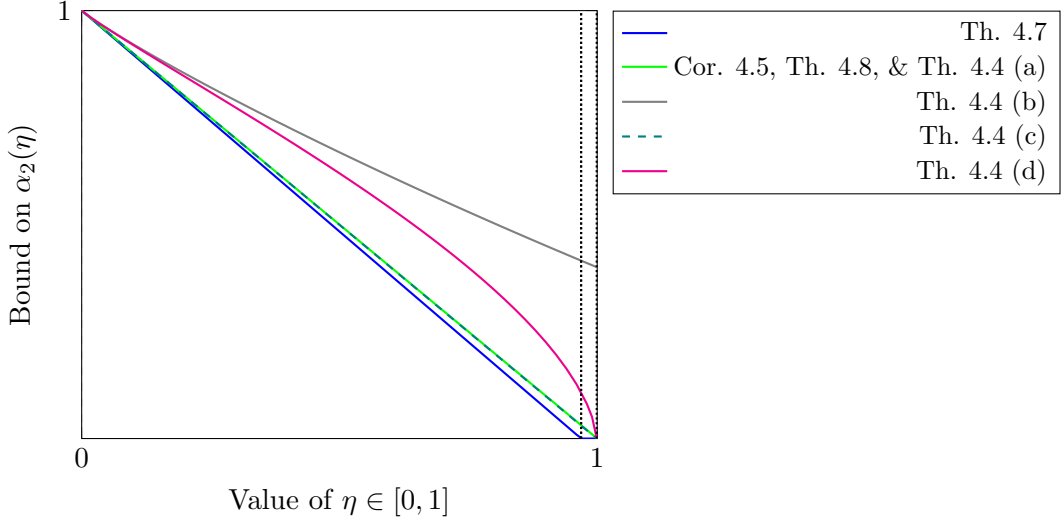
\begin{figure}[h!]
	\centering
	\begin{tikzpicture}[scale=1]
		\begin{axis}[
			legend pos = outer north east,
			legend cell align={right},
			xmin=0, xmax=1,
			ymin=0, ymax=1,
			xtick={0,1},
			ytick={1},
			xmajorgrids=true,
			grid style=dashed,
			every axis plot/.append style={thick},
			xlabel={Value of $\eta \in [0,1]$},
			ylabel={Bound on $\alpha_2(\eta)$}
			]
			
			\addplot+[color=blue,style = solid,mark=.,mark size=0.1pt]
			coordinates {
				(0.0000000000, 1.000000000)
				(0.02000000000, 0.9793548387)
				(0.04000000000, 0.9587096774)
				(0.06000000000, 0.9380645161)
				(0.08000000000, 0.9174193548)
				(0.1000000000, 0.8967741936)
				(0.1200000000, 0.8761290322)
				(0.1400000000, 0.8554838710)
				(0.1600000000, 0.8348387097)
				(0.1800000000, 0.8141935484)
				(0.2000000000, 0.7935483871)
				(0.2200000000, 0.7729032258)
				(0.2400000000, 0.7522580645)
				(0.2600000000, 0.7316129032)
				(0.2800000000, 0.7109677419)
				(0.3000000000, 0.6903225806)
				(0.3200000000, 0.6696774194)
				(0.3400000000, 0.6490322581)
				(0.3600000000, 0.6283870967)
				(0.3800000000, 0.6077419355)
				(0.4000000000, 0.5870967742)
				(0.4200000000, 0.5664516129)
				(0.4400000000, 0.5458064516)
				(0.4600000000, 0.5251612903)
				(0.4800000000, 0.5045161290)
				(0.5000000000, 0.4838709677)
				(0.5200000000, 0.4632258065)
				(0.5400000000, 0.4425806451)
				(0.5600000000, 0.4219354839)
				(0.5800000000, 0.4012903226)
				(0.6000000000, 0.3806451613)
				(0.6200000000, 0.3600000000)
				(0.6400000000, 0.3393548387)
				(0.6600000000, 0.3187096774)
				(0.6800000000, 0.2980645161)
				(0.7000000000, 0.2774193548)
				(0.7200000000, 0.2567741935)
				(0.7400000000, 0.2361290323)
				(0.7600000000, 0.2154838710)
				(0.7800000000, 0.1948387097)
				(0.8000000000, 0.1741935484)
				(0.8200000000, 0.1535483871)
				(0.8400000000, 0.1329032258)
				(0.8600000000, 0.1122580645)
				(0.8800000000, 0.09161290323)
				(0.9000000000, 0.07096774194)
				(0.9200000000, 0.05032258065)
				(0.9400000000, 0.02967741936)
				(0.9600000000, 0.009032258065)
				(0.9700000000, 0.0000000000)
				(0.9800000000, 0.0000000000)
				(0.9900000000, 0.0000000000)
				(1.000000000, 0.0000000000)
			};
			\addplot+[color=green,style = solid,mark=.,mark size=0.1pt]
			coordinates {
				(0.0000000000, 1.000000000)
				(0.05000000000, 0.9500000000)
				(0.1000000000, 0.9000000000)
				(0.1500000000, 0.8500000000)
				(0.2000000000, 0.8000000000)
				(0.2500000000, 0.7500000000)
				(0.3000000000, 0.7000000000)
				(0.3500000000, 0.6500000000)
				(0.4000000000, 0.6000000000)
				(0.4500000000, 0.5500000000)
				(0.5000000000, 0.5000000000)
				(0.5500000000, 0.4500000000)
				(0.6000000000, 0.4000000000)
				(0.6500000000, 0.3500000000)
				(0.7000000000, 0.3000000000)
				(0.7500000000, 0.2500000000)
				(0.8000000000, 0.2000000000)
				(0.8500000000, 0.1500000000)
				(0.9000000000, 0.1000000000)
				(0.9500000000, 0.05000000000)
				(1.000000000, 0.0000000000)
			};
			\addplot+[color=gray,style = solid, mark=.,mark size=0.1pt]
			coordinates {
				(0.0000000000, 1.000000000)
				(0.02000000000, 0.9819220960)
				(0.04000000000, 0.9658587649)
				(0.06000000000, 0.9505650429)
				(0.08000000000, 0.9357764194)
				(0.1000000000, 0.9213673521)
				(0.1200000000, 0.9072639655)
				(0.1400000000, 0.8934175019)
				(0.1600000000, 0.8797933576)
				(0.1800000000, 0.8663657045)
				(0.2000000000, 0.8531145363)
				(0.2200000000, 0.8400239095)
				(0.2400000000, 0.8270808283)
				(0.2600000000, 0.8142745059)
				(0.2800000000, 0.8015958528)
				(0.3000000000, 0.7890371131)
				(0.3200000000, 0.7765915976)
				(0.3400000000, 0.7642534848)
				(0.3600000000, 0.7520176677)
				(0.3800000000, 0.7398796359)
				(0.4000000000, 0.7278353819)
				(0.4200000000, 0.7158813270)
				(0.4400000000, 0.7040142603)
				(0.4600000000, 0.6922312898)
				(0.4800000000, 0.6805298018)
				(0.5000000000, 0.6689074268)
				(0.5200000000, 0.6573620116)
				(0.5400000000, 0.6458915948)
				(0.5600000000, 0.6344943869)
				(0.5800000000, 0.6231687529)
				(0.6000000000, 0.6119131972)
				(0.6200000000, 0.6007263509)
				(0.6400000000, 0.5896069605)
				(0.6600000000, 0.5785538785)
				(0.6800000000, 0.5675660549)
				(0.7000000000, 0.5566425294)
				(0.7200000000, 0.5457824256)
				(0.7400000000, 0.5349849449)
				(0.7600000000, 0.5242493614)
				(0.7800000000, 0.5135750184)
				(0.8000000000, 0.5029613234)
				(0.8200000000, 0.4924077455)
				(0.8400000000, 0.4819138125)
				(0.8600000000, 0.4714791078)
				(0.8800000000, 0.4611032690)
				(0.9000000000, 0.4507859853)
				(0.9200000000, 0.4405269964)
				(0.9400000000, 0.4303260910)
				(0.9600000000, 0.4201831058)
				(0.9800000000, 0.4100979248)
				(0.9900000000, 0.4050769873)
				(0.9990000000, 0.4005704802)
			};
		     	\addplot+[color=teal,style=dashed,mark=.,mark size=0.1pt]
		     coordinates {
		     	(0.0000000000, 1.000000000)
		     	(0.02000000000, 0.9799804497)
		     	(0.04000000000, 0.9599608993)
		     	(0.06000000000, 0.9399413490)
		     	(0.08000000000, 0.9199217986)
		     	(0.1000000000, 0.8999022483)
		     	(0.1200000000, 0.8798826979)
		     	(0.1400000000, 0.8598631476)
		     	(0.1600000000, 0.8398435973)
		     	(0.1800000000, 0.8198240469)
		     	(0.2000000000, 0.7998044966)
		     	(0.2200000000, 0.7797849462)
		     	(0.2400000000, 0.7597653959)
		     	(0.2600000000, 0.7397458456)
		     	(0.2800000000, 0.7197262952)
		     	(0.3000000000, 0.6997067449)
		     	(0.3200000000, 0.6796871945)
		     	(0.3400000000, 0.6596676442)
		     	(0.3600000000, 0.6396480938)
		     	(0.3800000000, 0.6196285435)
		     	(0.4000000000, 0.5996089932)
		     	(0.4200000000, 0.5795894428)
		     	(0.4400000000, 0.5595698925)
		     	(0.4600000000, 0.5395503421)
		     	(0.4800000000, 0.5195307918)
		     	(0.5000000000, 0.4995112414)
		     	(0.5200000000, 0.4794916911)
		     	(0.5400000000, 0.4594721408)
		     	(0.5600000000, 0.4394525904)
		     	(0.5800000000, 0.4194330401)
		     	(0.6000000000, 0.3994134897)
		     	(0.6200000000, 0.3793939394)
		     	(0.6400000000, 0.3593743891)
		     	(0.6600000000, 0.3393548387)
		     	(0.6800000000, 0.3193352884)
		     	(0.7000000000, 0.2993157380)
		     	(0.7200000000, 0.2792961877)
		     	(0.7400000000, 0.2592766373)
		     	(0.7600000000, 0.2392570870)
		     	(0.7800000000, 0.2192375367)
		     	(0.8000000000, 0.1992179863)
		     	(0.8200000000, 0.1791984360)
		     	(0.8400000000, 0.1591788856)
		     	(0.8600000000, 0.1391593353)
		     	(0.8800000000, 0.1191397849)
		     	(0.9000000000, 0.09912023460)
		     	(0.9200000000, 0.07910068426)
		     	(0.9400000000, 0.05908113392)
		     	(0.9600000000, 0.03906158358)
		     	(0.9800000000, 0.01904203324)
		     	(1.000000000, 0.0000000000)
		     };	
		     
			\addplot+[color=magenta,style = solid,mark=.,mark size=0.1pt]
			coordinates {
				(0.0000000000, 1.000000000)
				(0.02000000000, 0.9818380499)
				(0.04000000000, 0.9655399281)
				(0.06000000000, 0.9498682016)
				(0.08000000000, 0.9345606559)
				(0.1000000000, 0.9194921940)
				(0.1200000000, 0.9045883433)
				(0.1400000000, 0.8897990388)
				(0.1600000000, 0.8750878220)
				(0.1800000000, 0.8604265514)
				(0.2000000000, 0.8457924978)
				(0.2200000000, 0.8311666102)
				(0.2400000000, 0.8165324132)
				(0.2600000000, 0.8018752661)
				(0.2800000000, 0.7871818446)
				(0.3000000000, 0.7724397607)
				(0.3200000000, 0.7576372746)
				(0.3400000000, 0.7427630671)
				(0.3600000000, 0.7278060525)
				(0.3800000000, 0.7127552190)
				(0.4000000000, 0.6975994864)
				(0.4200000000, 0.6823275741)
				(0.4400000000, 0.6669278727)
				(0.4600000000, 0.6513883161)
				(0.4800000000, 0.6356962476)
				(0.5000000000, 0.6198382761)
				(0.5200000000, 0.6038001183)
				(0.5400000000, 0.5875664200)
				(0.5600000000, 0.5711205510)
				(0.5800000000, 0.5544443651)
				(0.6000000000, 0.5375179147)
				(0.6200000000, 0.5203191066)
				(0.6400000000, 0.5028232821)
				(0.6600000000, 0.4850026958)
				(0.6800000000, 0.4668258616)
				(0.7000000000, 0.4482567179)
				(0.7200000000, 0.4292535476)
				(0.7400000000, 0.4097675524)
				(0.7600000000, 0.3897409380)
				(0.7800000000, 0.3691042814)
				(0.8000000000, 0.3477728192)
				(0.8200000000, 0.3256410556)
				(0.8400000000, 0.3025746467)
				(0.8600000000, 0.2783976478)
				(0.8800000000, 0.2528713744)
				(0.9000000000, 0.2256568870)
				(0.9200000000, 0.1962420886)
				(0.9400000000, 0.1637807639)
				(0.9600000000, 0.1266601729)
				(0.9800000000, 0.08083818673)
				(0.9900000000, 0.05045602806)
				(0.9990000000, 0.0008005424237)	
			};	
		
			\addplot+[color=black, mark=.,mark size=0.1pt, style=densely dotted]
			coordinates {
				(0.96875,0.000)
				(0.96875,1.000)
			};
			\addplot+[color=black, mark=,mark size=0.1pt, style=densely dotted]
			coordinates {
				(0.96875,0.000)
				(0.96875,1.000)
			};
			\addplot+[color=black, mark=.,mark size=0.1pt, style=densely dotted]
			coordinates {
				(0.99902,0.000)
				(0.99902,1.000)
			};
			
			\legend{\small{Th. \ref{T-ASympPlotkin}},\small{Cor. \ref{T-ASympSingl}, Th. \ref{th:as_proj_sphere_pack}, \& Th. \ref{th:asinduced} (a) },
				\small{Th. \ref{th:asinduced} (b) },
				\small{Th. \ref{th:asinduced} (c) },\small{Th. \ref{th:asinduced} (d) }
			}
		\end{axis}
	\end{tikzpicture}
	\caption{\label{fig:allbds1} 
		Comparison of bounds on $\alpha_2(\eta)$ for 
		$(m_i)$, $(n_i)$, satisfying $m_1=10,\hat{m}=4,\hat{n}=2$.\\
	}
\end{figure}


\bigskip

\begin{figure}[h!]
	\centering
	\begin{tikzpicture}[scale=1]
		\begin{axis}[
			legend pos = outer north east,
			legend cell align={right},
			xmin=0, xmax=1,
			ymin=0, ymax=1,
			xtick={0,0.52,1},
			ytick={1},
			xmajorgrids=true,
			grid style=dashed,
			every axis plot/.append style={thick},
			xlabel={Value of $\eta \in [0,1]$},
			ylabel={Bound on $\alpha_2(\eta)$}
			]
			
			\addplot+[color=blue,style = solid,mark=.,mark size=0.1pt]
			coordinates {
				(0.0000000000, 1.000000000)
				(0.02000000000, 0.9796825397)
				(0.04000000000, 0.9593650794)
				(0.06000000000, 0.9390476191)
				(0.08000000000, 0.9187301587)
				(0.1000000000, 0.8984126984)
				(0.1200000000, 0.8780952381)
				(0.1400000000, 0.8577777778)
				(0.1600000000, 0.8374603175)
				(0.1800000000, 0.8171428571)
				(0.2000000000, 0.7968253968)
				(0.2200000000, 0.7765079365)
				(0.2400000000, 0.7561904762)
				(0.2600000000, 0.7358730159)
				(0.2800000000, 0.7155555555)
				(0.3000000000, 0.6952380952)
				(0.3200000000, 0.6749206349)
				(0.3400000000, 0.6546031746)
				(0.3600000000, 0.6342857143)
				(0.3800000000, 0.6139682540)
				(0.4000000000, 0.5936507937)
				(0.4200000000, 0.5733333333)
				(0.4400000000, 0.5530158730)
				(0.4600000000, 0.5326984127)
				(0.4800000000, 0.5123809524)
				(0.5000000000, 0.4920634921)
				(0.5200000000, 0.4717460317)
				(0.5400000000, 0.4514285714)
				(0.5600000000, 0.4311111111)
				(0.5800000000, 0.4107936508)
				(0.6000000000, 0.3904761905)
				(0.6200000000, 0.3701587302)
				(0.6350000000, 0.3549206349)
				(0.6400000000, 0.3498412698)
				(0.6600000000, 0.3295238095)
				(0.6800000000, 0.3092063492)
				(0.7000000000, 0.2888888889)
				(0.7200000000, 0.2685714286)
				(0.7400000000, 0.2482539683)
				(0.7600000000, 0.2279365079)
				(0.7800000000, 0.2076190476)
				(0.8000000000, 0.1873015873)
				(0.8200000000, 0.1669841270)
				(0.8400000000, 0.1466666667)
				(0.8600000000, 0.1263492063)
				(0.8800000000, 0.1060317460)
				(0.9000000000, 0.08571428571)
				(0.9200000000, 0.06539682540)
				(0.9400000000, 0.04507936508)
				(0.9600000000, 0.02476190476)
				(0.9800000000, 0.004444444444)
				(0.9850000000, 0.0000000000)
				(1.000000000, 0.0000000000)
			};
			\addplot+[color=green,style = solid,mark=.,mark size=0.1pt]
			coordinates {
				(0.0000000000, 1.000000000)
				(0.05000000000, 0.9500000000)
				(0.1000000000, 0.9000000000)
				(0.1500000000, 0.8500000000)
				(0.2000000000, 0.8000000000)
				(0.2500000000, 0.7500000000)
				(0.3000000000, 0.7000000000)
				(0.3500000000, 0.6500000000)
				(0.4000000000, 0.6000000000)
				(0.4500000000, 0.5500000000)
				(0.5000000000, 0.5000000000)
				(0.5500000000, 0.4500000000)
				(0.6000000000, 0.4000000000)
				(0.6500000000, 0.3500000000)
				(0.7000000000, 0.3000000000)
				(0.7500000000, 0.2500000000)
				(0.8000000000, 0.2000000000)
				(0.8500000000, 0.1500000000)
				(0.9000000000, 0.1000000000)
				(0.9500000000, 0.05000000000)
				(1.000000000, 0.0000000000)
			};
			\addplot+[color=red,style = solid,mark=.,mark size=0.1pt]
			coordinates {
				(0.0000000000, 0.9979878116)
				(0.02000000000, 0.9652864923)
				(0.04000000000, 0.9356936247)
				(0.06000000000, 0.9081005868)
				(0.08000000000, 0.8818384873)
				(0.1000000000, 0.8565866759)
				(0.1200000000, 0.8321875692)
				(0.1400000000, 0.8084978859)
				(0.1600000000, 0.7854541956)
				(0.1800000000, 0.7629775825)
				(0.2000000000, 0.7410229440)
				(0.2200000000, 0.7195587615)
				(0.2400000000, 0.6985478139)
				(0.2600000000, 0.6779704938)
				(0.2800000000, 0.6578013467)
				(0.3000000000, 0.6380276462)
				(0.3200000000, 0.6186297358)
				(0.3400000000, 0.5995989586)
				(0.3600000000, 0.5809226510)
				(0.3800000000, 0.5625913579)
				(0.4000000000, 0.5445965053)
				(0.4200000000, 0.5269321959)
				(0.4400000000, 0.5095905588)
				(0.4600000000, 0.4925671917)
				(0.4800000000, 0.4758557755)
				(0.5000000000, 0.4594529975)
				(0.5200000000, 0.4433546646)
				(0.5400000000, 0.4275562278)
				(0.5600000000, 0.4120554143)
				(0.5800000000, 0.3968487454)
				(0.6000000000, 0.3819336668)
				(0.6200000000, 0.3673079383)
				(0.6350000000, 0.3565271273)
				(0.6400000000, 0.3529693600)
				(0.6600000000, 0.3389159332)
				(0.6800000000, 0.3251459812)
				(0.7000000000, 0.3116581172)
				(0.7200000000, 0.2984508573)
				(0.7400000000, 0.2855230761)
				(0.7600000000, 0.2728739310)
				(0.7800000000, 0.2605024946)
				(0.7900000000, 0.2544207454)
				(0.7971380000, 0.2501218860)	
			};
			\addplot+[color=black,style = solid,mark=.,mark size=0.1pt]
			coordinates {
				(0.0000000000, 0.9979878116)
				(0.005000000000, 0.9813775916)
				(0.02000000000, 0.9356936247)
				(0.04000000000, 0.8818384873)
				(0.06000000000, 0.8321875692)
				(0.08000000000, 0.7854541956)
				(0.1000000000, 0.7410229440)
				(0.1200000000, 0.6985478139)
				(0.1400000000, 0.6578013467)
				(0.1600000000, 0.6186297358)
				(0.1800000000, 0.5809226510)
				(0.2000000000, 0.5445965053)
				(0.2200000000, 0.5095905588)
				(0.2400000000, 0.4758557755)
				(0.2400000000, 0.4758557755)
				(0.2600000000, 0.4433546646)
				(0.2800000000, 0.4120554143)
				(0.3000000000, 0.3819336668)
				(0.3200000000, 0.3529693600)
				(0.3400000000, 0.3251459812)
				(0.3600000000, 0.2984508573)
				(0.3800000000, 0.2728739310)
				(0.4000000000, 0.2484081597)
				(0.4200000000, 0.2250488791)
				(0.4400000000, 0.2027936717)
				(0.4600000000, 0.1816422104)
				(0.4800000000, 0.1615960553)
				(0.5000000000, 0.1426585639)
				(0.5200000000, 0.1248348018)
				(0.5400000000, 0.1081315819)
				(0.5600000000, 0.09255756540)
				(0.5800000000, 0.07812342301)
				(0.6000000000, 0.06484204925)
				(0.6200000000, 0.05272875835)
				(0.6400000000, 0.04180143034)
				(0.6600000000, 0.03208055468)
				(0.6800000000, 0.02358913905)
				(0.7000000000, 0.01635248789)
				(0.7200000000, 0.01039791105)
				(0.7400000000, 0.005754503137)
				(0.7600000000, 0.002453214555)
				(0.7800000000, 0.0005275010607)
				(0.7900000000, 0.00009198052045)
				(0.7971380000, 0.0000000000)
			};
			\addplot+[color=gray,style = solid,mark=.,mark size=0.1pt]
			coordinates {
				(0.0000000000, 1.000000000)
				(0.02000000000, 0.9700344895)
				(0.04000000000, 0.9451054114)
				(0.06000000000, 0.9221003561)
				(0.08000000000, 0.9003580468)
				(0.1000000000, 0.8795646283)
				(0.1200000000, 0.8595354113)
				(0.1400000000, 0.8401485019)
				(0.1600000000, 0.8213173905)
				(0.1800000000, 0.8029775074)
				(0.2000000000, 0.7850788367)
				(0.2200000000, 0.7675815191)
				(0.2400000000, 0.7504530658)
				(0.2600000000, 0.7336665094)
				(0.2800000000, 0.7171991263)
				(0.3000000000, 0.7010315265)
				(0.3200000000, 0.6851469875)
				(0.3400000000, 0.6695309550)
				(0.3600000000, 0.6541706618)
				(0.3800000000, 0.6390548318)
				(0.4000000000, 0.6241734465)
				(0.4200000000, 0.6095175588)
				(0.4400000000, 0.5950791415)
				(0.4600000000, 0.5808509649)
				(0.4800000000, 0.5668264944)
				(0.5000000000, 0.5529998066)
				(0.5200000000, 0.5393655181)
				(0.5400000000, 0.5259187258)
				(0.5600000000, 0.5126549557)
				(0.5800000000, 0.4995701202)
				(0.6000000000, 0.4866604805)
				(0.6200000000, 0.4739226142)
				(0.6400000000, 0.4613533879)
				(0.6600000000, 0.4489499327)
				(0.6800000000, 0.4367096231)
				(0.7000000000, 0.4246300590)
				(0.7200000000, 0.4127090491)
				(0.7400000000, 0.4009445967)
				(0.7600000000, 0.3893348878)
				(0.7800000000, 0.3778782798)
				(0.8000000000, 0.3665732918)
				(0.8200000000, 0.3554185968)
				(0.8400000000, 0.3444130137)
				(0.8600000000, 0.3335555017)
				(0.8800000000, 0.3228451542)
				(0.9000000000, 0.3122811945)
				(0.9200000000, 0.3018629718)
				(0.9375000000, 0.2928661586)
			};
		    \addplot+[color=teal,mark=.,mark size=0.1pt]
		    coordinates {
		    	(0.0000000000, 1.000000000)
		    	(0.02000000000, 0.9786666667)
		    	(0.04000000000, 0.9573333333)
		    	(0.06000000000, 0.9360000000)
		    	(0.08000000000, 0.9146666666)
		    	(0.1000000000, 0.8933333333)
		    	(0.1200000000, 0.8720000000)
		    	(0.1400000000, 0.8506666667)
		    	(0.1600000000, 0.8293333334)
		    	(0.1800000000, 0.8080000000)
		    	(0.2000000000, 0.7866666667)
		    	(0.2200000000, 0.7653333333)
		    	(0.2400000000, 0.7440000000)
		    	(0.2600000000, 0.7226666667)
		    	(0.2800000000, 0.7013333333)
		    	(0.3000000000, 0.6800000000)
		    	(0.3200000000, 0.6586666667)
		    	(0.3400000000, 0.6373333333)
		    	(0.3600000000, 0.6160000000)
		    	(0.3800000000, 0.5946666666)
		    	(0.4000000000, 0.5733333333)
		    	(0.4200000000, 0.5520000000)
		    	(0.4400000000, 0.5306666667)
		    	(0.4600000000, 0.5093333333)
		    	(0.4800000000, 0.4880000000)
		    	(0.5000000000, 0.4666666667)
		    	(0.5200000000, 0.4453333333)
		    	(0.5400000000, 0.4240000000)
		    	(0.5600000000, 0.4026666667)
		    	(0.5800000000, 0.3813333333)
		    	(0.6000000000, 0.3600000000)
		    	(0.6200000000, 0.3386666667)
		    	(0.6400000000, 0.3173333333)
		    	(0.6600000000, 0.2960000000)
		    	(0.6800000000, 0.2746666667)
		    	(0.7000000000, 0.2533333333)
		    	(0.7200000000, 0.2320000000)
		    	(0.7400000000, 0.2106666667)
		    	(0.7600000000, 0.1893333333)
		    	(0.7800000000, 0.1680000000)
		    	(0.8000000000, 0.1466666667)
		    	(0.8200000000, 0.1253333333)
		    	(0.8400000000, 0.1040000000)
		    	(0.8600000000, 0.08266666667)
		    	(0.8800000000, 0.06133333333)
		    	(0.9000000000, 0.04000000000)
		    	(0.9200000000, 0.01866666667)
		    	(0.9400000000, 0.0000000000)
		    	(0.9600000000, 0.0000000000)
		    	(0.9800000000, 0.0000000000)
		    	(1.000000000, 0.0000000000)
		    };
			\addplot+[color=magenta,style = solid,mark=.,mark size=0.1pt]
			coordinates {
				(0.0000000000, 1.000000000)
				(0.02000000000, 0.9698925403)
				(0.04000000000, 0.9445868972)
				(0.06000000000, 0.9209955569)
				(0.08000000000, 0.8984684799)
				(0.1000000000, 0.8766984085)
				(0.1200000000, 0.8555047837)
				(0.1400000000, 0.8347682761)
				(0.1600000000, 0.8144038349)
				(0.1800000000, 0.7943475043)
				(0.2000000000, 0.7745492003)
				(0.2200000000, 0.7549684140)
				(0.2400000000, 0.7355714935)
				(0.2600000000, 0.7163298359)
				(0.2800000000, 0.6972186356)
				(0.3000000000, 0.6782159864)
				(0.3200000000, 0.6593022169)
				(0.3400000000, 0.6404593851)
				(0.3600000000, 0.6216708825)
				(0.3800000000, 0.6029211152)
				(0.4000000000, 0.5841952399)
				(0.4200000000, 0.5654789389)
				(0.4400000000, 0.5467582205)
				(0.4600000000, 0.5280192365)
				(0.4800000000, 0.5092481092)
				(0.5000000000, 0.4904307603)
				(0.5200000000, 0.4715527346)
				(0.5400000000, 0.4525990137)
				(0.5600000000, 0.4335538093)
				(0.5800000000, 0.4144003296)
				(0.6000000000, 0.3951205059)
				(0.6200000000, 0.3756946660)
				(0.6400000000, 0.3561011325)
				(0.6600000000, 0.3363157213)
				(0.6800000000, 0.3163110990)
				(0.7000000000, 0.2960559420)
				(0.7200000000, 0.2755138107)
				(0.7400000000, 0.2546416024)
				(0.7600000000, 0.2333873635)
				(0.7800000000, 0.2116870890)
				(0.8000000000, 0.1894598472)
				(0.8200000000, 0.1665999798)
				(0.8400000000, 0.1429638107)
				(0.8600000000, 0.1183450386)
				(0.8800000000, 0.09242356950)
				(0.9000000000, 0.06463848221)
				(0.9200000000, 0.03375348690)
				(0.9375000000, 0.0000000000)
			};	
			
		    \addplot+[color=black, mark=.,mark size=0.1pt, style=densely dotted]
		    coordinates {
		    	(0.797138,0.000)
		    	(0.797138,1.000)
		    };
		    \addplot+[color=black, mark=.,mark size=0.1pt, style=densely dotted]
		    coordinates {
		    	(0.5200000000,0.000)
		    	(0.5200000000, 0.4453333333)
		    };
		    \addplot+[color=black, mark=.,mark size=0.1pt, style=densely dotted]
		    coordinates {
		    	(0.000,0.4453333333)
		    	(0.5200000000, 0.4453333333)
		    };
		    \addplot+[color=black, mark=,mark size=0.1pt, style=densely dotted]
		    coordinates {
		    	(0.984375,0.000)
		    	(0.984375,1.000)
		    };
		    \addplot+[color=black, mark=.,mark size=0.1pt, style=densely dotted]
		    coordinates {
		    	(0.9375,0.000)
		    	(0.9375,1.000)
		    };	
			\legend{\small{Th. \ref{T-ASympPlotkin}},\small{Cor. \ref{T-ASympSingl}, Th. \ref{th:as_proj_sphere_pack}, \& Th. \ref{th:asinduced} (a) },
			\small{Cor. \ref{cor:as_sphere_cov_pack}},\small{Cor. \ref{cor:as_sphere_cov_pack}},\small{Th. \ref{th:asinduced} (b) },
			\small{Th. \ref{th:asinduced} (c) },\small{Th. \ref{th:asinduced} (d) }
		}
		\end{axis}
	\end{tikzpicture}
	\caption{\label{fig:allbds2} 
		Comparison of bounds on $\alpha_2(\eta)$ for the sequences $(m_i)=(n_i)=(4,4,...)$.\\
	}
\end{figure}
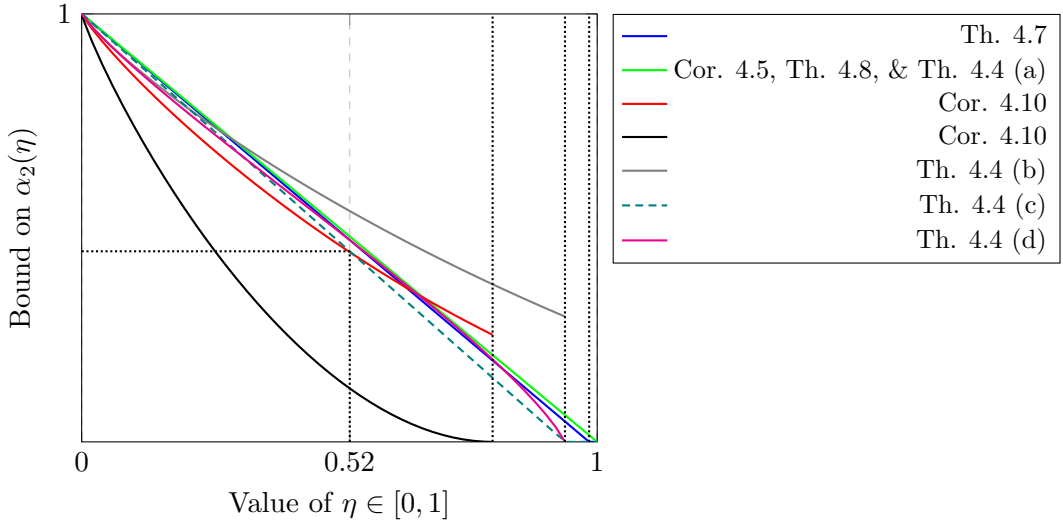

\section{Duality and MacWilliams Identities}\label{S-Duality}
For the rest of the paper we again focus on $\F_{q}$-linear sum-rank metric codes.
In this section we will investigate the distributions introduced in Definition~\ref{D-WeightEnum} under duality.
We show that the sum-rank distribution does not obey a MacWilliams identity, while both the rank-list distribution and the support distribution do. 
Moreover, we give explicit formulas
for the corresponding MacWilliams transformations. We follow the notation of Section~\ref{S-Prelim}; recall in particular
\eqref{e-nimi}, \eqref{e-Pi} and~\eqref{e-NM} and Definition~\ref{D-Lattice}.

\begin{definition}\label{D-Dual}
The \textbf{dual} of a code $C \le \Pi$ is defined as
\[
   C^\perp:=\bigg\{(Y_1,...,Y_t) \in \Pi\,\bigg|\, \sum_{i=1}^t \subspace{ X_i, Y_i } =0 \mbox{ for all $(X_1,...,X_t) \in C$}\bigg\} \le \Pi,
\]
where $\subspace{\, \cdot, \cdot\,}$ denotes the trace-product.
Furthermore, the \textbf{dual} of $\bm{U} \in \mL$ is $\bm{U}^\perp:=(U_1^\perp,...,U_t^\perp)$, where
$U_i^\perp$ is the orthogonal of $U_i$ with respect to the standard inner product of
$\F_q^{n_i}$.
\end{definition}

One easily observes that $\Pi({\bm U}^\perp)=\Pi({\bm U})^\perp$, where $\Pi({\bm U})$ is as in Definition~\ref{D-Short}.
Furthermore, using the isomorphism~$\psi$ from~(\ref{e-Psi})
 and the ordinary trace dual on $\F_q^{N\times M}$ we obtain $\psi(C^\perp)=\psi(C)^\perp\cap\F_q^{N\times M}[\mP]$.
 
Our first observation, shown in the following example, illustrates that the sum-rank distributions of a code~$C$ and its dual~$C^\perp$ do not satisfy
a MacWilliams identity.

\begin{example}
The sum-rank distribution of a code does not determine the sum-rank distribution of its dual code.
Consider e.g. the 1-dimensional codes $C_1$ and $C_2$ in $\Pi=\F_2^{2\times2}\times\F_2^{2\times2}$ 
generated by the matrix tuples
\[
     \left(\begin{pmatrix} 1 & 0 \\ 0 & 1 \end{pmatrix}, \begin{pmatrix} 0 & 0 \\ 0&0\end{pmatrix} \right)
 \ \text{ and } \
     \left(\begin{pmatrix} 1 & 0 \\ 0 & 0 \end{pmatrix}, \begin{pmatrix} 1 & 0\\0&0 \end{pmatrix} \right),
\]
respectively. Then $C_1$ and $C_2$ have the same sum-rank distribution. However, $C_1^\perp$ has
12 elements of sum-rank 1 and $C_2^\perp$ has 10.
\end{example}

On the positive side, the rank-list distribution and the support distribution of sum-rank metric codes
do satisfy MacWilliams identities. 
We establish these results using a combinatorial technique 
that gives an explicit formula for the transformation between the distributions.
In order to do so, we need the following observation, which can be thought of as a sum-rank version of the duality between puncturing and shortening. It extends \cite[Lem.~28]{alb1} from rank-metric codes to sum-rank metric codes.

\begin{proposition} \label{P-CUdual}
Let $C \le \Pi$ and $\bm{U}=(U_1,\ldots,U_t) \in \mL$. Let $u_i:=\dim(U_i)$ for $i \in [t]$. Then
$$|C(\bm{U})| = \frac{|C|}{|\Pi({\bm U})^\perp|} |C^\perp(\bm{U}^\perp)|=\frac{|C|}{q^{\sum_{i=1}^t m_i(n_i-u_i)}} |C^\perp(\bm{U}^\perp)|.$$
\end{proposition}

\begin{proof}
Let $C,A \le \Pi$ be  linear subspaces.
Then
\[
   \dim(C \cap A) = \dim(C)- \dim(A^\perp)+\dim(C^\perp \cap A^\perp),
\]
which can be seen by taking dimensions in the identity
$(C \cap A)^\perp=C^\perp + A^\perp$ and using the dimension formula for the sum of vector spaces.
Specializing to $A=\Pi(\bm{U})$, using $A^\perp =\Pi({\bm U}^\perp)$ and applying Remark~\ref{R-PiU}, we obtain the desired result.
\end{proof}

Now we turn to the support distribution. 
Define the partial order $\leq $ on $\N_0^t$ by
\begin{equation}\label{e-vlessu}
    \bm{v}\leq\bm{u}:\Longleftrightarrow v_i\leq u_i\text{ for all }i\in[t],
\end{equation}
for any ${\bm u}=(u_1,\ldots,u_t),\, {\bm v}=(v_1,\ldots,v_t)\in\N_0^t$. 
The following result shows that the support distribution of~$C^\perp$ is fully determined by the support distribution of~$C$.

\begin{theorem}
Let $C \le \Pi$.
Moreover, let $\bm{U}=(U_1,\ldots,U_t) \in \mL$ and $\dim(\bm{U})=\bm{u}=(u_1,\ldots,u_t)$. Then
\[
   W_{\bm{U}}(C^\perp) = 
   \frac{1}{|C|}\sum_{\bm{H} \in \mL} W_{\bm{H}}(C)\sum_{\bm{v} \le \bm{u}}q^{\sum_{i=1}^t m_iv_i}
   \prod_{i=1}^t (-1)^{u_i-v_i} q^{\binom{u_i-v_i}{2}} 
    \GaussianD{\dim(H_i^\perp \cap U_i)}{v_i}_q.
\]
In particular, the support distribution of~$C^\perp$ is uniquely determined by the support distribution of~$C$ via an invertible linear transformation.
\end{theorem}

\begin{proof}
We follow/extend the approach of~\cite{ravagnani2018duality} based on combinatorics.
Fix $\bm{U} \in \mL$ as in the theorem.
Identity~\eqref{e-CUWV} applied to $C^\perp$ along with M\"obius inversion in the lattice $\mL$ yields
\[
   W_{\bm{U}}(C^\perp) \ = \  \sum_{\substack{\bm{V} \in \mL \\ \bm{V} \le \bm{U}}}
|C^\perp(\bm{V})| \,  \mu_\mL(\bm{V},\bm{U}).
\]
Using Proposition~\ref{P-CUdual} and the explicit expression for the M\"obius function in Definition~\ref{D-Lattice} 
we obtain
\begin{equation} \label{ee1}
W_{\bm{U}}(C^\perp) = \sum_{\bm{v} \le \bm{u}} \;
\frac{|C^\perp|}{q^{\sum_{i=1}^t m_i(n_i-v_i)}}
\prod_{i=1}^t (-1)^{u_i-v_i} q^{\binom{u_i-v_i}{2}} 
\sum_{\substack{\bm{V} \in \mL \\ \bm{V} \le \bm{U} \\ \dim(\bm{V})=\bm{v}}}
|C(\bm{V}^\perp)|.
\end{equation}
Fix $\bm{v} \le \bm{u}$ and consider the set
\[
    \{(\bm{V},X)\in\mL\times C \mid  \bm{V} \le \bm{U}, \, \dim(\bm{V})=\bm{v},\, \sigma(X) \le \bm{V}^\perp\}.
\]
Counting its elements  in two ways, we obtain
\begin{align*}
   \sum_{\substack{\bm{V} \in \mL \\ \bm{V} \le \bm{U} \\ \dim(\bm{V})=\bm{v}}}
   |C(\bm{V}^\perp)|
   & =\sum_{\bm{H} \in \mL} 
   \sum_{\substack{X \in C \\ \sigma(X)=\bm{H}}} \big|\{\bm{V} \in \mL \mid 
   \dim(\bm{V})=\bm{v}, \, \bm{V} \le \bm{H}^\perp\cap \bm{U}\}\big|\\
   &=\sum_{\bm{H} \in \mL} W_{\bm{H}}(C) \prod_{i=1}^t \GaussianD{\dim(H_i^\perp \cap U_i)}{v_i}_q.
\end{align*}
Combining this with \eqref{ee1} and using $|C^\perp|=|\Pi|/|C|=q^{\sum_{i=1}^t m_in_i}/|C|$, we arrive at
$$W_{\bm{U}}(C^\perp) = \frac{1}{|C|}
\sum_{\bm{H} \in \mL} W_{\bm{H}}(C)
 \sum_{\bm{v} \le \bm{u}} \; 
q^{\sum_{i=1}^t m_iv_i}
\prod_{i=1}^t (-1)^{u_i-v_i} q^{\binom{u_i-v_i}{2}} 
 \qbin{\dim(H_i^\perp \cap U_i)}{v_i}{q},$$
as desired.
\end{proof}

Similarly, we obtain a MacWilliams duality for the rank-list distribution.

\begin{theorem}
Let $C \le \Pi$ and $\bm{u} \in \N_0^t$. Then
\[
  W_{\bm{u}}(C^\perp) = \frac{1}{|C|} \sum_{\bm{h} \in \N_0^t} W_{\bm{h}}(C) 
  \sum_{\bm{v} \le \bm{u}} q^{\sum_{i=1}^t m_iv_i}
        \prod_{i=1}^t (-1)^{u_i-v_i} q^{\binom{u_i-v_i}{2}}\qbin{n_i-h_i}{v_i}{q}\qbin{n_i-v_i}{u_i-v_i}{q}.
\]
\end{theorem}

\begin{proof}
Starting from~\eqref{ee1} and summing over all $\bm{U} \in \mL$ with 
$\dim(\bm{U})=\bm{u}$, we obtain
\begin{equation} \label{ee11}
W_{\bm{u}}(C^\perp) = \sum_{\bm{v} \le \bm{u}} \;
\frac{|C^\perp|}{q^{\sum_{i=1}^t m_i(n_i-v_i)}}
\prod_{i=1}^t (-1)^{u_i-v_i} q^{\binom{u_i-v_i}{2}} 
\sum_{\substack{\bm{U} \in \mL \\ \dim(\bm{U})=\bm{u}}}
\sum_{\substack{\bm{V} \in \mL \\ \bm{V} \le \bm{U} \\ \dim(\bm{V})=\bm{v}}}
|C(\bm{V}^\perp)|.
\end{equation}
Fix $\bm{v} \le \bm{u}$ and consider the set
$$S=\{(\bm{U}, \bm{V}, X) \in \mL \times \mL \times C \mid \dim(\bm{U})=\bm{u}, \,
 \dim(\bm{V})=\bm{v}, \,
\bm{V} \le \bm{U}, \, \sigma(X) \le \bm{V}^\perp
\}.$$
On the one hand,
$$|S| = \sum_{\substack{\bm{U} \in \mL \\ \dim(\bm{U})=\bm{u}}}
\sum_{\substack{\bm{V} \in \mL \\ \bm{V} \le \bm{U} \\ \dim(\bm{V})=\bm{v}}}
|C(\bm{V}^\perp)|,$$
and on the other hand,
\begin{align*}
|S| &= \sum_{\bm{h} \in \N_0^t}
\ \sum_{\substack{X \in C \\ \dim(\sigma(X)) = \bm{h}}}
 \ \sum_{\substack{\bm{V} \in \mL \\ \dim(\bm{V})=\bm{v} \\ \bm{V} \le \sigma(X)^\perp}}
   \big|\{\bm{U} \in \mL \mid  \dim(\bm{U})=\bm{u},\,\bm{V} \le \bm{U}\}\big| \\
&= \sum_{\bm{h} \in \N_0^t} W_{\bm{h}}(C) \prod_{i=1}^t \qbin{n_i-h_i}{v_i}{q}
\qbin{n_i-v_i}{u_i-v_i}{q}.
\end{align*}
Therefore
$$\sum_{\substack{\bm{U} \in \mL \\ \dim(\bm{U})=\bm{u}}}
\sum_{\substack{\bm{V} \in \mL \\ \bm{V} \le \bm{U} \\ \dim(\bm{V})=\bm{v}}}
|C(\bm{V}^\perp)| = 
\sum_{\bm{h} \in \N_0^t} W_{\bm{h}}(C) \prod_{i=1}^t \qbin{n_i-h_i}{v_i}{q}
\qbin{n_i-v_i}{u_i-v_i}{q}.
$$
Combining this with~\eqref{ee11} we obtain
$$
W_{\bm{u}}(C^\perp) = \frac{1}{|C|} \sum_{\bm{h} \in \N_0^t} W_{\bm{h}}(C) 
\sum_{\bm{v} \le \bm{u}} \;
q^{\sum_{i=1}^t m_iv_i}
\prod_{i=1}^t (-1)^{u_i-v_i} q^{\binom{u_i-v_i}{2}} 
 \qbin{n_i-h_i}{v_i}{q}
\qbin{n_i-v_i}{u_i-v_i}{q},$$
concluding the proof.
\end{proof}

We conclude this section with the sum-rank metric analogue of the binomial moments of the MacWilliams identities.

\begin{theorem} \label{thm:BM}
Let $C \le \Pi$.
For all $\bm{u}=(u_1,...,u_t) \in \N_0^t$ we have
\[
    \sum_{\bm{h} \in \N_0^t} W_{\bm{h}}(C) \prod_{i=1}^t \qbin{n_i-h_i}{u_i-h_i}{q} = 
    \frac{|C|}{q^{\sum_{i=1}^tm_i(n_i-u_i)}} \sum_{\bm{h} \in \N_0^t} 
           W_{\bm{h}}(C^\perp) \prod_{i=1}^t \qbin{n_i-h_i}{u_i}{q}.
\]
\end{theorem}

\begin{proof}
Fix any tuple $\bm{u} \in \N_0^t$. We start from Proposition~\ref{P-CUdual} and sum over all $\bm{U} \in \mL$
with $\dim(\bm{U})=\bm{u}$, obtaining
\begin{equation} \label{f1}
\sum_{\substack{\bm{U} \in \mL \\ \dim(\bm{U}) = \bm{u}}} |C(\bm{U})|
= \frac{|C|}{q^{\sum_{i=1}^tm_i(n_i-u_i)}} \sum_{\substack{\bm{U} \in \mL \\ \dim(\bm{U}) = \bm{u}}} |C^\perp(\bm{U}^\perp)|.
\end{equation}
The map $\bm{U} \mapsto \bm{U}^\perp$ induces a bijection between the tuples 
$\bm{U} \in \mL$ with $\dim(\bm{U})= \bm{u}$ and those with
$\dim(\bm{U})= \bm{n}-\bm{u}$, where $\bm{n}=(n_1,\ldots,n_t)$.
Therefore~\eqref{f1} can be re-written as
\begin{equation} \label{f2}
\sum_{\substack{\bm{U} \in \mL \\ \dim(\bm{U}) = \bm{u}}} |C(\bm{U})|
= \frac{|C|}{q^{\sum_{i=1}^tm_i(n_i-u_i)}} \sum_{\substack{\bm{U} \in \mL \\ \dim(\bm{U}) = \bm{n}-\bm{u}}} |C^\perp(\bm{U})|.
\end{equation}
Finally, observe that for all $\bm{v} \in \N_0^t$ one has
$$\sum_{\substack{\bm{V} \in \mL \\ \dim(\bm{V})=\bm{v}}} |C(\bm{V})| = 
\sum_{X \in C} |\{\bm{V} \in \mL \mid \dim(\bm{V})=\bm{v}, \, \bm{V} \ge \sigma(X)
 \}| = 
\sum_{\bm{h} \in \N_0^t} W_{\bm{h}}(C) \prod_{i=1}^t \qbin{n_i-h_i}{v_i-h_i}{q}.$$
Using this fact twice in \eqref{f2} 
yields the desired result.
\end{proof}

\begin{remark}
Theorem~\ref{thm:BM} implies 
the binomial moments of the MacWilliams identities for both the Hamming and the rank metric, as in
\cite[eq. (M1), p. 257]{pless} and \cite[Theorem 31]{alb1} respectively.
Indeed, the binomial moments for the rank-metric are easily obtained by setting $t=1$ in Theorem~\ref{thm:BM}.
To recover the binomial moments for the Hamming metric, 
let $(n_1,...,n_t)=(m_1,...,m_t)=(1,..,1)$ and 
fix $0 \le \nu \le t$.
Consider the identity in Theorem~\ref{thm:BM} and sum
over all $\bm{u} \in \{0,1\}^t$
with $|\bm{u}|=t-\nu$.
This results in
\[
    \sum_{\substack{\bm{u} \in \{0,1\}^t \\ |\bm{u}|=t-\nu}} \ \sum_{\substack{\bm{h} \in \{0,1\}^t \\ \bm{h} \le \bm{u}}}
W_{\bm{h}}(C) = \frac{|C|}{q^{\nu}}
\sum_{\substack{\bm{u} \in \{0,1\}^t \\ |\bm{u}|=t-\nu}} \ \sum_{\substack{\bm{h} \in \{0,1\}^t \\ \bm{h} \le (1,...,1)-\bm{u}}}
W_{\bm{h}}(C^\perp).
\]
The latter identity can be re-written as
$$\sum_{i=0}^{t-\nu} W_i(C) \binom{t-i}{\nu} = 
\frac{|C|}{q^\nu}
\sum_{i=0}^{\nu} W_i(C^\perp) \binom{t-i}{t-\nu},$$
where $W_i(C)$ and $W_i(C^\perp)$ denote the number of codewords of Hamming weight $i$ in $C$ and $C^\perp$, respectively. This is precisely \cite[eq.~(M1) on page 257]{pless}.
\end{remark}

\section{Linear MSRD Codes} \label{sec:msrd}
In this section we investigate the structural properties of linear MSRD codes, i.e., the sum-rank metric codes that meet the Singleton Bound
of Theorem~\ref{th:singl}; see also Definition~\ref{D-MSRD}.
We first study duality of MSRD codes.
We show that if $m_1=m_2= \cdots =m_t$, the MSRD property is invariant under dualization, while this is not true for general 
$m_1,\ldots,m_t$. 
Next, the duality result for the case $m_1=m_2= \cdots =m_t$ allows us to explicitly determine the support distribution of MSRD codes 
(and hence the rank-list and sum-rank distributions). 
This in turn leads to necessary conditions for the existence of MSRD codes.
In some cases they turn out to be more powerful than the bounds of Section~\ref{sec:bounds}.
Finally, we provide an upper bound on the length~$t$ of an MSRD code for the case where $n_1=\ldots=n_t$ and $m_1=\ldots=m_t$, improving on Corollary~\ref{cor:tlarge} for certain classes of parameters.
We close the section with discussing puncturing and shortening of MSRD codes.
We will see that doing so in a suitable way results in an MSRD code again. 

We start with the invariance of the MSRD property under dualization for the case where $m_1=\ldots=m_t$. 
For $\F_{q^m}$-linear MSRD codes (see Remark~\ref{R-Fqmlinear}) this result has been derived earlier in~\cite[Thm.~5]{MP19}.
Recall the notation~$C^\perp$ and ${\bm U}^\perp$ from \cref{D-Dual} and $\srk(C)$ from Definition~\ref{D-SRMC}.

\begin{theorem}\label{T-DualMSRD}
Suppose $m_1= \cdots =m_t=m$. If $C \le \Pi$ is MSRD, then
$C^\perp$ is
MSRD as well. Moreover, if~$C$ is a non-trivial MSRD code, then
$\srk(C^\perp)=N-\srk(C)+2$.
\end{theorem}

\begin{proof}
The result is immediate if $C$ is trivial. 
We henceforth assume that~$C$ is MSRD with sum-rank distance $d \ge 2$, and hence has 
dimension $m(N-d+1)$.
We show that $C^\perp$ is MSRD of minimum distance $N-d+2$. Observe that $\dim(C^\perp)=
mN-\dim(C) =m(d-1)$. By Theorem~\ref{th:singl} we then have
$m(d-1) \le m(N-d^\perp+1)$, where $d^\perp$ is the minimum distance of $C^\perp$.
In particular, $d^\perp \le N-d+2$. Therefore it suffices to show that $d^\perp \ge N-d+2$.
We will use the criterion in~\cref{R-eqv} and the notation from \cref{D-Lattice}.
Let $\bm{U} \in \mL$ be arbitrary with
$\rk_\mL(\bm{U})=N-d+1$. Then $\rk_\mL({\bm U}^\perp)=d-1$ and thus
$|\mC(\bm{U}^\perp)|=1$. 
By \cref{P-CUdual} we have
$$1=|C(\bm{U}^\perp)|= \frac{|C|}{q^{m(N-d+1)}} |C^\perp(\bm{U})| = |C^\perp(\bm{U})|.$$
Since $\bm{U}$ was arbitrary with $\rk_\mL(U)=N-d+1$, \cref{R-eqv} implies $d^\perp \ge N-d+2$
as desired.
\end{proof}

If $m_i\neq m_j$ for some $i,j$ with $i \neq j$, Theorem~\ref{T-DualMSRD} is false in general,
 as the following example shows.

\begin{example}\label{E-DualNotMSRD}
Let $\Pi=\Pi_q(n\times n\mid 1\times1)$, where $n\geq 2$ and $d=2$.
Then the Singleton Bound in Theorem~\ref{th:singl} is $n^2-n+1$.
An MSRD code can be constructed as follows.
Let $C_1\leq\F^{n\times n}$ be an MRD code with rank distance~$2$, and let $Z\in\F_q^{n\times n}$ be a matrix of rank~$1$.
Set
\[
    C=\{(A,\,0)\mid A\in C_1\}+\{\lambda(Z,\,1)\mid \lambda\in\F_q\}.
\]
Then $\srk(C)\geq 2$ thanks to $\rk(A+\lambda Z)\geq1$ for all $A\in C_1$ and $\lambda\in\F_q^*$.
Since $\dim(C)=n(n-1)+1$, the code $C$ is MSRD.
The dual code is given by
\[
    C^\perp=\{(B,-\subspace{B,Z})\mid B\in C_1^\perp\}.
\]
Indeed, the right hand side is clearly contained in $C^\perp$ and has dimension equal to $\dim(C_1^\perp)=n=\dim(\Pi)-\dim(C)$.
In order to determine $\srk(C^\perp)$, consider the map $C_1^\perp\longrightarrow\F_q,\ B\longmapsto \subspace{B,Z}$.
It is linear, surjective (because $Z\not\in C_1$) and thus has an $(n-1)$-dimensional kernel. This shows that the sum-rank distance of $C^\perp$ equals the rank distance of $C_1^\perp$, which is~$n$.
But for that distance, the Singleton Bound in Theorem~\ref{th:singl} is $n+1$, and thus $C^\perp$ is not MSRD.
\end{example}

Let us return to the case where $m_1= \cdots =m_t$.
Theorem~\ref{T-DualMSRD} allows us to give a closed formula for the support distribution of any MSRD code.
We need the following simple lemma.
Recall the lattice~$\mL$ and the accompanying notation from Definition~\ref{D-Lattice}.

\begin{corollary} \label{coro:CU}
Suppose $m_1= \cdots =m_t=m$. Let $C \le \Pi$ be a non-zero MSRD code of minimum distance $d$. For all
$\bm{U} \in \mL$ with $\rk_\mL(\bm{U})=u$ we have
\[
     |C(\bm{U})|= \left\{ \begin{array}{ll} 1 & \mbox{if $u<d$,} \\
    q^{m(u-d+1)} & \mbox{otherwise.}   \end{array}  \right.
\]
\end{corollary}

\begin{proof}
The statement is clear for $u<d$. Thus let $u\geq d$.
Then $\rk_\mL(\bm{U}^\perp)=N-u\leq N-d$ while $\srk(C^\perp)=N-d+2$ thanks to \cref{T-DualMSRD}.
Hence $|C^\perp(U^\perp)|=1$, and the result follows from  \cref{P-CUdual}.
\end{proof}

In order to determine the support distribution, we introduce the following notation.
Recall the partial order on~$\N_0^t$ in~\eqref{e-vlessu}.

\begin{notation}
Let $\ell\in\N_0$ and $\bm{U}\in\mL$. We define
\[
   f_\ell(\bm{u})=\sum_{\substack{\bm{v} \le \bm{u} \\ |\bm{v}|=\ell}} \prod_{i=1}^t (-1)^{u_i-v_i} 
         q^{\binom{u_i-v_i}{2}} \qbin{u_i}{v_i}{q}.
\]
\end{notation}

Then
\begin{equation}\label{e-fl}
   f_\ell(\bm{u})=0\text{ if }\ell>|\bm{u}|,\quad f_{|\bm{u}|}(\bm{u})=1,\quad 
    f_{|\bm{u}|-1}(\bm{u})=-\sum_{i=1}^t\GaussianD{u_i}{1}_q=-\frac{\sum_{i=1}^tq^{u_i}-t}{q-1}.
\end{equation}

Now we are ready to present the support distribution of an MSRD code (recall Definition~\ref{D-WeightEnum}).

\begin{theorem} \label{disMSRD}
Suppose $m_1= \cdots =m_t=m$.
Let $C \le \Pi$ be a non-zero MSRD code of minimum distance $d$, and let
$\bm{U} \in \mL\setminus0$ with $\dim(\bm{U})=\bm{u}$. Then
\[
W_{\bm{U}}(C) = \sum_{\ell=d}^{|\bm{u}|}\big(q^{m(\ell-d+1)}-1\big)f_\ell(\bm{u}).
\]
In particular, $W_{\bm{U}}(C)$ only depends on $\bm{u}=\dim(\bm{U})$, but not on the subspace~$\bm{U}$ itself. 
\end{theorem}

\begin{proof}
Set $\mV_{\ell}(\bm{U})=\{\bm{V}\in\mL\mid \bm{V}\leq\bm{U},\rkL(\bm{V})=\ell\}$.
Applying M\"obius inversion to~\eqref{e-CUWV} and using Corollary~\ref{coro:CU}, we obtain
\begin{eqnarray*}
W_{\bm{U}}(C) &=& \sum_{\substack{\bm{V} \in \mL \\ \bm{V} \le \bm{U}}}
|C(\bm{V})| \, \mu_{\mL}(\bm{V},\bm{U}) \\
&=& \sum_{\ell=0}^{d-1} \sum_{\bm{V}\in\mV_{\ell}(U)} \mu_{\mL}(\bm{V},\bm{U}) \ 
+ \sum_{\ell=d}^{|\bm{u}|} 
\sum_{\bm{V}\in\mV_{\ell}(U)}
q^{m(\ell-d+1)} \mu_{\mL}(\bm{V},\bm{U})\\
&=& \sum_{\ell=d}^{|\bm{u}|}\big(q^{m(\ell-d+1)}-1\big)\sum_{\bm{V}\in\mV_{\ell}(\bm{U})}\muL(\bm{V},\bm{U}),
\end{eqnarray*}
where the last identity follows from $\sum_{\bm{V}\leq \bm{U}}\muL(\bm{V},\bm{U})=0$ since $\bm{U}\neq0$.
Finally, one easily observes that 
 $\sum_{\bm{V}\in\mV_{\ell}(\bm{U})}\muL(\bm{V},\bm{U})=f_\ell(\bm{u})$, and this proves the stated identity.
\end{proof}

We remark that examples show that Theorem \ref{disMSRD} does not necessarily hold if we remove the constraint that the $m_i$ are all equal.
 
Now we arrive at the following non-existence criterion for MSRD codes.
It is an immediate consequence of Theorem~\ref{disMSRD}.

\begin{corollary} \label{exclMW}
Consider the space $\Pi=\Pi(n_1\times m\mid\cdots\mid n_t\times m)$ and define the index set
$\mI_{(n_1,\ldots,n_t)}=\{(u_1,\ldots,u_t)\in\N_0^t\mid u_i\leq n_i\text{ for all }i\in[t]\}$.
Suppose there exists an MSRD code $C \le \Pi$ of sum-\- rank distance~$d$. 
Then 
\begin{equation} \label{dise}
 \omega(\bm{u}):=
   \sum_{\ell=d}^{|\bm{u}|} (q^{m(\ell-d+1)}-1) f_\ell(\bm{u}) \ge 0\ \text{ for all }\bm{u}\in\mI_{(n_1,\ldots,n_t)}\setminus\{(0,\ldots,0)\}.
\end{equation}
\end{corollary}

This criterion is indeed quite powerful. 

\begin{example}\label{E-Sing-PSP-TW2}
Consider $\Pi=\Pi_3(3\times3\mid3\times3\mid2\times3)$ and $d=7$. 
Then the linear version of  the Singleton Bound provides a strictly smaller value than all other  bounds. 
In other words, no known bound excludes the existence of an MSRD code.
However, applying the above non-existence criterion 
leads to $\omega(3,3,2)=-52$,  and thus an MSRD code does not exist.
This example shows that MSRD codes of length $t=q$ do not exist in general, not even for fixed column size.
On the other hand, for certain parameters such codes do exist, 
as Example~\ref{E-MSRD6} has shown.
In this context we also wish to remind of \cite[Thm.~4]{MP18} where it has been shown that (even $\F_{q^m}$-linear) MSRD codes always
exist for length~$t$ up to $q-1$.
\end{example}

It is a simple consequence of~\eqref{e-fl} that $\omega(\bm{u})=0$ if $|\bm{u}|<d$ and
$\omega(\bm{u})=q^m-1$ if $|\bm{u}|=d+1$.
Hence it suffices to consider $|\bm{u}|\geq d+1$.
The criterion is obviously conclusive only in the case where $ \omega(\bm{u})<0$ for some $\bm{u}\in\mI_{(n_1,\ldots,n_t)}$, 
but such a $\bm{u}$ may be hard to find within the, generally large, set $\mI_{(n_1,\ldots,n_t)}$, even with the restriction
$|\bm{u}|\geq d+1$.
However, an abundance of examples leads us to the following conjecture.
It tells us that the tests in~\eqref{dise} can be replaced by a single test. 

\begin{conjecture}\label{C-omegau}
Let $(n_1,\ldots,n_t)\in\N_0^t$ and, without loss of generality, $n_1\geq\ldots\geq n_t$. 
Fix $d<N$ and let $s\in[t]$ and $\delta\in[n_{s+1}-1]$ be such that 
$\tilde{\bm{u}}:=(n_1,\ldots,n_s,\delta,0,\ldots,0)\in\mI_{(n_1,\ldots,n_t)}$ satisfies $|\tilde{\bm{u}}|=d+1$.
In other words, $\tilde{\bm{u}}$ is the minimum in $\mI_{(n_1,\ldots,n_t)}$
with respect to the graded reverse lexicographic order and such that $|\tilde{\bm{u}}|\geq d+1$.
Then $\omega(\tilde{\bm{u}})=q^{2m}-1-(q^m-1)/(q-1)\big(\sum_{i=1}^tq^{\tilde{u}_i}-t\big)$ and the conjecture is
\[
     \omega(\tilde{\bm{u}})\geq 0\Longrightarrow \omega(\bm{u})\geq 0\text{ for all }\bm{u}\in\mI_{(n_1,\ldots,n_t)}.
\]
This tells us that the much faster test $\omega(\tilde{\bm{u}})\geq 0$ appears to be as strong as the laborious test in Corollary~\ref{exclMW}.
\end{conjecture}

\begin{example}\label{E-nnt2}
Let $\F=\F_2$ and $\Pi=\Pi_2(n\times n\mid n\times n)$. Then there exists no MSRD code of distance
$d=n+1$. Indeed, in this case $\tilde{\bm{u}}=(n,2)$ and $\omega(\tilde{\bm{u}})=1-2^n$.
\end{example}

Let us return to Corollary~\ref{exclMW}.
Applying the same line of reasoning to the dual code, which is also MSRD by 
Theorem \ref{T-DualMSRD}, we obtain an additional non-existence criterion for MSRD codes.

\begin{corollary} \label{exclMWdual}
Consider $\Pi=\Pi(n_1\times m\mid \cdots\mid n_t\times m)$ and let $N=\sum_{i=1}^t n_i$.
Suppose there exists an MSRD code $C \le \Pi$ of minimum distance $d$. 
Then 
\begin{equation} \label{dise2}
   \hat{\omega}(\bm{u}):= 
   \sum_{\ell=N-d+2}^{|\bm{u}|} (q^{m(\ell-N+d-1)}-1) f_\ell(\bm{u}) \ge 0
 \text{ for all }\bm{u}\in\mI_{(n_1,\ldots,n_t)}\setminus\{(0,\ldots,0)\}.
\end{equation}
\end{corollary}
The same comment as in Conjecture~\ref{C-omegau} applies:  assuming $n_1\geq\ldots\geq n_t$  we believe that
it suffices to test $\hat{\omega}(\tilde{\bm{v}})\geq0$ for $\tilde{\bm{v}}=(n_1,\ldots,n_{s'},\delta',0,\ldots,0)$ where
$|\tilde{\bm v}|=N-d+3$.
It is not hard to find examples where one of Corollary~\ref{exclMW} and Corollary~\ref{exclMWdual}
excludes an MSRD code, but not the other one.

\begin{remark}\label{R-SRDistr}
The above results lead immediately to explicit formulas for the sum-rank distribution 
and the rank-list distribution; see Definition~\ref{D-WeightEnum}. 
Indeed, $W_r(C)=\sum_{\bm{U} \in \mL,\,\rkL(\bm{U})=r}W_{\bm{U}}(C)$ for any $r\in\N_0$ and $W_{\bm{u}}(C)=\sum_{\bm{U} \in \mL,\,\dim(\bm{U})=\bm{u}}W_{\bm{U}}(C)$ for any $\bm{u}\in\N_0^t$.
Thus, Theorem~\ref{disMSRD} generalizes accordingly.
\end{remark}

\medskip
We now turn to the length~$t$ of an MSRD code.
Recall that Corollary~\ref{cor:tlarge} gives an upper bound 
for $t$ that applies to an arbitrary sum-rank metric code.
In this section we restrict to MSRD codes and obtain a much 
stronger bound for the case where $n_1=\ldots=n_t$ and $m_1=\ldots=m_t$.
The proof uses the projective sphere-packing bound of Theorem~\ref{th:prpack}.

\begin{theorem}\label{coro:nec}
Suppose $n=n_1= \cdots =n_t$ and $m=m_1= \cdots =m_t$, and suppose there exists an MSRD code $C \le \Pi$
of minimum distance $d \ge 3$. Then
\[
   t \le  \Big\lfloor \frac{d-3}{n} \Big\rfloor + \Big\lfloor\frac{q^n - q^{n \lfloor (d-3)/n \rfloor +n-d+3} + (q-1)(q^m+1)}{q^n-1}\Big\rfloor
   \leq \Big\lfloor \frac{d-3}{n} \Big\rfloor + 1+ \Big\lfloor\frac{q^m(q-1)}{q^n-1}\Big\rfloor.
\]
In particular, we have the following cases.
\begin{alphalist}
\item If $n\mid(d-3)$, then 
\[
     t \le  \frac{d-3}{n} +\Big\lfloor\frac{(q-1)(q^m+1)}{q^n-1}\Big\rfloor.
\]
\item If $d \le n+2$, then 
\[
  t \le \Big\lfloor\frac{q^n-q^{n-d+3}+(q-1)(q^m+1)}{q^n-1}\Big\rfloor \le 1+ \Big\lfloor\frac{q^m(q-1)}{q^n-1}\Big\rfloor.
\]
If in addition $n=m$, then this implies $t\leq(q^{n+1}-1)/(q^n-1)\leq q+1$ and even $t\leq q$ if $n=m\geq2$.
\end{alphalist}
\end{theorem}

\begin{proof}
We apply the projective sphere-packing bound of Theorem~\ref{th:prpack}.
Write $d-3=\ell n+\delta$ with $0\leq\delta<n$.
Thus $\ell= \lfloor (d-3)/n \rfloor$ and $t'=t-\ell$.
Moreover, 
$n_1'=n-\delta,\,  m_1'=m$, and $(n_i',m_i')=(n,m)$ for $2 \le i \le t'$. 
Since~$C$ is MSRD, the Singleton Bound (Theorem~\ref{th:singl}) tells us that $|C| = q^{m(tn-d+1)}$.
Therefore Theorem \ref{th:prpack} reads as
\[
   q^{m(tn-d+1)}\Big(1+\frac{q^m-1}{q-1}\big((t'-1)(q^n-1)+q^{n-\delta}-1\big)\Big)\leq q^{m(tn-d+3)},
\]
which in turn is equivalent to
\[
  t\leq\frac{(q-1)(q^m+1)-q^{n-\delta}+1}{q^n-1}+\ell+1=\ell+\frac{(q-1)(q^m+1)+q^n-q^{n-\delta}}{q^n-1}
\]
Using $\ell= \lfloor (d-3)/n \rfloor$ and $\delta=d-3-\ell n$, we arrive at the first inequality.
The second one follows from $n-\delta\geq1$, hence $q^n-q^{n-\delta}\leq q^n-q$.
The special cases follow easily from the first inequality.
\end{proof}

\begin{remark}
Suppose $n=n_1= \cdots =n_t$ and $m=m_1= \cdots =m_t$, as in Theorem~\ref{coro:nec}.
\begin{alphalist}
\item For $n=m=1$, the upper bound in Theorem \ref{coro:nec} reads as $t\leq q+d-2$, which is a well-known
upper bound on the length of an MDS code $C \le \F_q^t$ of minimum distance $d \ge 3$; see e.g.~\cite[Cor. 7.4.3(ii)]{pless}.
For $n=1<m$ this generalizes to $t\leq q^m+d-2$ by Part~(a) above.
As we will spell out in Example~\ref{E-CompMat}, it is a consequence of the MDS conjecture that actually $t\leq q^m+1$.
\item The MDS conjecture states that $N\leq q^m+1$ for an $\F_{q^m}$-linear MDS code in $\F_{q^m}^N$ of minimum distance~$d$
   (minus some exceptional cases); see ~\cite[p.~265]{pless}. 
   For an MSRD code this results in $t\leq (q^m+1)/n$.  
   It is straightforward to show that the first bound in Theorem~\ref{coro:nec} yields a tighter bound for~$t$  
   if $d\leq q^m\big(1-n(q-1)/(q^n-1)\big)+4-n$.
   In the special case where $n=m$ and $d\leq n+2$, Theorem~\ref{coro:nec}(b) even provides the much better estimate $t\leq q+1$.
\end{alphalist}
\end{remark}

For relatively small length, MSRD codes do exist for large classes of parameters. 
In \cite[Thm.~4]{MP18} a construction of $\F_{q^m}$-linear  MSRD codes in $\F_{q^m}^N$  is given for the case where 
$m:=m_1=\ldots=m_t$ and where the length satisfies $t\leq q-1$ (and thus $N\leq q^m+1$, in accordance with the MDS conjecture).
Example~\ref{E-nnt2} shows that for $t=q$ MSRD codes may not exist.
On the other hand, in Examples~\ref{E-MSRD8} and~\ref{E-MSRD6} we have seen examples of MSRD codes with $t>q$; 
even for the case where $m_1=\cdots=m_t$.
In the next section we present a few more classes of MSRD codes.

We close this section with a discussion of specific shortening and puncturing of MSRD codes.
A few details are needed in order to specify where to puncture or shorten.
We set up the following notation.

Let~$\Pi$ be as in~\eqref{e-nimi} and~\eqref{e-Pi} and let $C\leq\Pi$ be a linear MSRD code of sum-rank distance~$d$.
As in Theorem~\ref{th:singl} let
\begin{equation}\label{e-jdelta}
   \text{$j\in[t]$ and $\delta\in\{0,\ldots, n_j-1\}$ be the unique integers such that $d-1=\sum_{i=1}^{j-1}n_i+\delta$.}
\end{equation}
Set
\begin{equation}\label{e-niprime}
     n_i'=\begin{cases} n_i,&\text{if }i\neq j,\\ n_j-\delta,&\text{if }i=j,\end{cases}
\end{equation}
and define the space $\Pi'=\Pi(n_j'\times m_j\mid\cdots\mid n_t'\times m_t)$.
Thanks to the Singleton Bound we have $\dim(C)=\sum_{i=j}^tn_i'm_i$ and
the map
\[
   \tau: C\longrightarrow\Pi',\quad
   (X_1,\ldots,X_t)\longmapsto (\hat{X}_j,X_{j+1},\ldots,X_t)
   \quad 
   \left\{\begin{array}{l}\text{where $\hat{X}_j\in\F^{(n_j-\delta)\times m_j}$ consists}\\ \text{of the first $n_j-\delta$ rows of~$X_j$.}
   \end{array}\right.
\]
is an isomorphism.
Indeed,~$\tau$ is injective by~\eqref{e-jdelta} and thus bijective since $\dim(C)=\dim(\Pi')$.
As a consequence,~$C$ has a basis of the form 
\[
     \{B_{a,b}^{(j)}\mid a=1,\ldots,n_j-\delta,\,b=1,\ldots,m_j\}\cup\{B_{a,b}^{(i)}\mid i=j+1,\ldots,t,\,a=1,\ldots,n_i,\,b=1,\ldots,m_i\}
\]
where 
\begin{equation}\label{e-BasisMSRD}
\left.\begin{array}{ccll}
      B_{a,b}^{(j)}&\!\!=\!\!&\Big(\ \ast,\ \ldots,\ast,\SmallMatTwo{\!E_{a,b}^{(j)}\!}{\ast},\quad 0,\  \ldots\ldots,\,0,\ 0\ \Big)
      &\text{for }a\in[n_j'],\,b\in[m_j],\\[2ex]
      B_{a,b}^{(j+1)}&\!\!=\!\!&\Big(\ \ast,\ \ldots,\ast, \SmallMatTwo{\ 0\ }{\ast},\ E_{a,b}^{(j+1)},\ \ldots,\,0,\ 0\ \Big)
      &\text{for }a\in [n_{j+1}'],\,b\in[m_{j+1}],\\
        \vdots          &  &\qquad \vdots \\
     B_{a,b}^{(t)}&\!\!=\!\!&\Big(\underbrace{\ast,\ \ldots,\ \ast}_{j-1 \text{ blocks}}, \SmallMatTwo{\  0\ }{\ast},\quad 0,\ \ldots\ldots,\,0, E_{a,b}^{(t)}\Big)&\text{for }a\in[n_t'],\,b\in[m_t],
\end{array}\ \right\}
\end{equation}
where $E_{a,b}^{(i)}$ denote the standard basis matrices in $\F^{n_i'\times m_i}$ and $\ast$ denote suitable matrices in the respective matrix space.
We define index sets
\begin{align*}
    \tail(C)&=\{(i,a,b)\mid i=j,\ldots,t,\,a\in[n_i'],\,b\in[m_i]\},\\[.6ex]
    \head(C)&=\{(i,a,b)\mid i\in[j-1],\,a\in[n_i'],\,b\in[m_i]\}\cup\{(j,a,b)\mid a=n_j'+1,\ldots,n_j,\,b\in[m_j]\}.
\end{align*}

Note that $|\tail(C)|=\dim(C)$.
Furthermore, the above means that $\tail(C)$ forms an information set for~$C$ in the usual sense.
Note that the MSRD codes in Examples~\ref{E-MSRD8} and~\ref{E-MSRD6} are given in form of a basis as in~\eqref{e-BasisMSRD}; in both cases $\delta=0$.

The above provides us with many options to shorten an MSRD code without compromising the MSRD property.
It should be noted that the shortening in~(a) below is a special case of shortening as introduced in 
Definition~\ref{D-Short} -- up to isomorphism (which consists of the actual deletion of a row).
On the other hand, shortening on a column as considered in~(b) is not an instance of Definition~\ref{D-Short}.

\begin{theorem}[\textbf{Shortening an MSRD Code on a Row or Column}]\label{T-MSRDShorten}
Suppose there exists a linear MSRD code $C\leq\Pi$ of sum-rank distance~$d$ and with data as
in~\eqref{e-jdelta} --~\eqref{e-BasisMSRD}. 
\begin{alphalist}
\item Choose $s\in\{j,\ldots,t\}$ and set
\[
     \tilde{n}_i=\begin{cases} n_i,&\text{if }i\neq s,\\ n_s-1,&\text{if }i=s.\end{cases}
\]
Then there exists an MSRD code in $\tilde{\Pi}:=\Pi(\tilde{n}_1\times m_1\mid\cdots\mid \tilde{n}_t\times m_t)$ with sum-rank distance~$d$.
It is obtained by shortening~$C$ on a row with indices in the tail of~$C$.
\item Choose $s\in\{j+1,\ldots,t\}$ and set
\[
     \tilde{m}_i=\begin{cases} m_i,&\text{if }i\neq s,\\ m_s-1,&\text{if }i=s.\end{cases}
\]
Then there exists an MSRD code in $\tilde{\Pi}:=\Pi(n_1\times\tilde{m}_1\mid\cdots\mid n_t\times \tilde{m}_t)$ with sum-rank distance~$d$.
It is obtained by shortening~$C$ on a column with indices in the tail of~$C$, but not in the $j$-th block.
\end{alphalist}
\end{theorem}

\begin{proof}
(a) Choose any $a\in[n_s']$, where $n_i'$ is as in~\eqref{e-niprime}.
Consider the basis $B$ from~\eqref{e-BasisMSRD} and set
\[
    B'=B\setminus\{B_{a,b}^{(s)}\mid b=1,\ldots,m_s\}.
\]
Let $C'=\subspace{B'}$.
Then clearly $\srk(C')\geq\srk(C)$ and $\dim(C')=\dim(C)-m_s=\sum_{i=j}^t \tilde{n}_im_i-m_j\delta$.
Since all matrix tuples in~$C'$ have a zero row in the $s$-th block at position~$a$, deleting that row
results in a code~$\tilde{C}$ in~$\tilde{\Pi}$ of the same sum-rank distance and dimension as $C'$.
Since $d-1=\sum_{i=1}^{j-1}\tilde{n}_i+\delta$, the code $\tilde{C}$ is MSRD.
\\
(b) Choose any $b\in[m_s]$ and set
$B'=B\setminus\{B_{a,b}^{(s)}\mid a=1,\ldots,n_s\}$.
Let $C'=\subspace{B'}$.
Then clearly $\srk(C')\geq\srk(C)$ and $\dim(C')=\dim(C)-n_s=\sum_{i=j}^t n_i\tilde{m}_i-m_j\delta$.
Deleting column~$b$ in the $s$-th block results in  
an MSRD code~$\tilde{C}$ in~$\tilde{\Pi}$ of the same sum-rank distance and dimension as $C'$.
\end{proof}

We now turn to puncturing. In this case only puncturing of rows leads to an MSRD code in general.

\begin{theorem}[\textbf{Puncturing an MSRD Code on a Row}]\label{T-MSRDPuncture}
Suppose there exists a linear MSRD code $C\leq\Pi$ of sum-rank distance~$d$ and with data as
in~\eqref{e-jdelta} --~\eqref{e-BasisMSRD}.
If $\delta>0$ choose $s\in[j]$ and if $\delta=0$ choose $s\in[j-1]$. Set
\[
     \tilde{n}_i=\begin{cases} n_i,&\text{if }i\neq s,\\ n_s-1,&\text{if }i=s.\end{cases}
\]
Then there exists an MSRD code in $\tilde{\Pi}:=\Pi(\tilde{n}_1\times m_1\mid\cdots\mid \tilde{n}_t\times m_t)$ with sum-rank distance~$d-1$.
It is obtained by puncturing~$C$ on a row with indices in the head of~$C$.
\end{theorem}

\begin{proof}
Consider the map
$\pi:C\longrightarrow \Pi,\ (X_1,\ldots,X_t)\longmapsto(X_1,\ldots,\hat{X}_s,\ldots,X_t)$,
where $\hat{X}_s\in\F^{(n_s-1)\times m_s}$ is obtained from $X_s$ be removing the last row.
Note that if $s=j$, then $\delta>0$ and therefore this last row belongs to the head of~$C$.
Since $\srk(C)\geq2$, the map~$\pi$ is injective. 
Set $C'=\pi(C)$. Then $\srk(C')\geq d-1$.
If $s<j$, we have
\[
  d-2=\sum_{i=1}^{j-1}n_i+\delta-1=\sum_{i=1}^{j-1}\tilde{n}_i+\delta\ \text{ and }\
  \dim(C')=\dim(C)=\sum_{i=j}^t\tilde{n}_im_i-m_j\delta,
\]
and thus~$C'$ is MSRD. 
For $s=j$ we have $\delta>0$ and with $\delta'=\delta-1$ we obtain
\[
  d-2=\sum_{i=1}^{j-1}n_i+\delta'\ \text{ and }\
  \dim(C')=\dim(C)=\sum_{i=j}^tn_im_i-m_j\delta=\sum_{i=j}^t\tilde{n}_im_i-\delta'm_j,
\]
and again~$C'$ is MSRD.
\end{proof}

\begin{remark}\label{R-ShortSamem}
Let us briefly consider the situation where $m_1=\ldots=m_t=:m$. 
In this case we may consider~\eqref{e-jdelta} for any ordering of the blocks. 
This implies that shortening on a row as in Theorem~\ref{T-MSRDShorten} and puncturing as in 
Theorem~\ref{T-MSRDPuncture} can be applied to any row in~$\Pi$ and always leads to an MSRD code.
This agrees with \cite[Cor.~7]{MP19}, where shortening and puncturing (called restriction) are 
defined more generally, but only applied to $\F_{q^m}$-linear codes.
\end{remark}

\section{Constructions of Optimal Codes}\label{S-Constr}

In this section we concentrate on constructions of optimal codes. The main focus is 
on MSRD codes, which we construct in various ad-hoc ways for several parameter sets.
Note that while a general construction of (even $\F_{q^m}$-linear) MSRD codes exists for $t \le q-1$
if $m:=m_1=...=m_t$, see~\cite[Thm.~4]{MP18}, some of our constructions
provide examples of significantly longer MSRD codes if we allow the blocks to have variable number of columns.
Our constructions show that for sum-rank distance equal to $2$ or $N$, MSRD codes exist for all parameters while this is not true for other distances.
Thereafter we will establish the existence of MSRD codes for certain other sum-rank distances in the case that there are 
sufficiently many $1\times1$-blocks involved.

We conclude the section with presenting a lifting construction that produces a sum-rank metric code
by combining a Hamming-metric and a rank-metric code.
As an application, we obtain a family of codes that meet the Induced Plotkin Bound of Theorem~\ref{th:induced}.

\subsection{MSRD Codes}

We start with the following simple example showing that MSRD codes of length $t\leq q^m+1$ in 
$\Pi_q(1\times m\mid\cdots\mid1\times m)$ exist for all distances $d\leq t$. 
They are, up to an isomorphism $\F_{q^m}\cong\F_q^{1\times m}$, simply the $\F_q$-linear MDS codes in $\F_{q^m}^t$.

\begin{const}\label{E-CompMat}
Fix a field~$\F_q$ and consider $\Pi=\Pi_q(1\times m\mid\cdots\mid1\times m)$ with~$t$ blocks. 
Let $d\leq t$. If $t\leq q^m+1$, there exists an ($\F_{q^m}$-linear) MDS code $C\leq\F_{q^m}^t$ of Hamming distance~$d$.
Thus $\dim_{\F_q}(C)=m(t-d+1)$.
Let $\phi:\F_{q^m}\longrightarrow\F_q^{1\times m}$ be an $\F_q$-isomorphism and extend it entrywise to an isomorphism
$\phi:\F_{q^m}^t\longrightarrow\Pi$.
Then the sum-rank metric code $\tilde{C}=\phi(C)\leq\Pi$ is MSRD with sum-rank distance~$d$.
Since conversely every linear MSRD code in~$\Pi$ induces an $\F_q$-linear MDS code in $\F_{q^m}^t$, the MDS conjecture (for nonlinear codes) implies that such an MSRD code with distance not in $\{1,2,t\}$ exists iff $t\leq q^m+1$.
\end{const}

The next example shows that MSRD codes with  sum-rank distance $d=2$ exist for all parameters.
We proceed in two steps.
The construction makes use of MRD codes in the rank metric; see~\cite{del1,gabidulin,roth} among many others.
Recall that for $n\leq m$ an $[n\times m;d]_q$-MRD code is a linear rank-metric code in $\F_q^{n\times m}$ of rank 
distance~$d$ and dimension $m(n-d+1)$. If possible we skip the subscript~$q$.

\begin{const}\label{E-Dual}
\begin{alphalist}
\item 
Let $\Pi=\Pi(n_1\times m\mid\cdots\mid n_t\times m)$. Without loss of generality let $n_1\geq\ldots\geq n_t$. 
We will construct MSRD codes with distance $d=2$.
Let $\hat{C}$ be an $[n_1\times m;2]$-MRD code, thus $\dim(\hat{C})=m(n_1-1)$.
Consider the map  $\phi_i:\F_q^{n_i\times m}\longrightarrow\F_q^{n_1\times m}$ defined by adding $n_1-n_i$ zero
rows to the given matrix. Then $\phi_i$ is rank-preserving and linear.
Define the sum-rank metric code
\[
C=\Big\{\Big(A-\sum_{i=2}^{t}\phi_i(A_i),A_2,\ldots,A_t\Big)\,\Big|\, A_i\in\F_q^{n_i\times m},A\in\hat{C}\Big\}.
\]
Then $\dim(C)=m\sum_{i=2}^{t}n_i+ m(n_1-1)=m(N-1)$, which is the Singleton  bound (Theorem~\ref{th:singl}) for distance~$2$.
It remains to see that $\srk(A-\sum_{i=2}^{t}\phi_i(A_i),A_2,\ldots,A_{t},)\geq 2$ for all nonzero elements in~$C$.
This is clear if either $\srk(A_2,\ldots,A_{t})\geq 2$ or if $(A_2,\ldots,A_{t})=0$.
Thus let $\srk(A_2,\ldots,A_{t})=1$. That means $\rk(A_\ell)=1$ for some~$\ell$ and $A_i=0$ for $i\neq \ell$.
In this case $A-\sum_{i=2}^{t}\phi_i(A_i)=A-\phi_\ell(A_\ell)$, and this matrix is not zero because $\rk(A)\neq 1$.
All of this shows that $C$ is an MSRD code with distance~$2$.
The dual code is given by $C^\perp=\{(B,\psi_2(B),\ldots,\psi_{t}(B))\mid B\in\hat{C}^\perp\}$, where 
$\psi_i:\F^{n_1\times m}\longrightarrow\F^{n_i\times m}$ is the projection onto the first $n_i$ rows.
This follows from the simple identity $\subspace{\psi_i(B),A_i}=\subspace{B,\phi_i(A_i)}$ along with $\dim(C^\perp)=m=\dim\hat{C}^\perp$.
Using that~$\hat{C}^\perp$ has rank distance~$n_1$, we conclude that~$C^\perp$ has sum-rank distance~$N$ and thus is MSRD, 
in agreement with Theorem~\ref{T-DualMSRD}.
If $n_i=n_1$ for all~$i$, then $C^{\perp}$ is just the repetition code $\{(B,\ldots,B)\mid B\in\hat{C}^\perp\}$.
\item The previous construction can be generalized to arbitrary ambient spaces. Let~$\Pi$ be as in\eqref{e-nimi}--~\eqref{e-Pi}.
Define $\hat{\Pi}=\Pi(n_1\times m_1\mid\cdots\mid n_t\times m_1)$ and identify~$\Pi$ with the subspace of~$\hat{\Pi}$ consisting of the
matrix tuples for which the last $m_1-m_i$ columns in the $i$th block are zero for all~$i\in[t]$.
By~(a) there exists an MSRD code $\hat{C}\leq\hat{\Pi}$ with distance $d=2$.
Hence $\dim(\hat{C})=m_1(N-1)$.
Note that the projection
$\rho:\hat{C}\longrightarrow\Pi_q\big((n_1-1)\times m_1\mid n_2\times m_1\mid \cdots\mid n_t\times m_1\big)$,
obtained by deleting the first row in the first block is an isomorphism because $\srk(\hat{C})>1$. 
Consider now $C:=\hat{C}\cap\Pi$. Clearly $\srk(C)\geq2$.
In order to determine its dimension, consider the map
\[
     \tau:\hat{C}\longrightarrow\Pi_q\big(n_2\times(m_1-m_2)\mid \cdots\mid n_t\times(m_1-m_t)\big),
\]
where we project onto the last $m_1-m_i$ columns of each matrix.
Then clearly $\ker(\tau)=C$ and the surjectivity of~$\rho$ implies the surjectivity of~$\tau$.     
This implies $\dim(C)=\sum_{i=2}^{t}m_in_i+m_1(n_1-1)$, which is the Singleton Bound in Theorem~\ref{th:singl} for $d=2$.
\end{alphalist}
\end{const}

We now turn to the other extreme and consider distance $d=N$ and $d=N-1$.
As we will show, MSRD codes with sum-rank distance $d=N$ exist for all parameters, while for $d=N-1$ this is not the case.

\begin{const}\label{E-SimpleMSRD}
\begin{alphalist}
\item Consider $\Pi$ as in \eqref{e-nimi}--\eqref{e-Pi}. We show that there exists an MSRD code with sum-rank
         distance~$N=\sum_{i=1}^t n_i$.
         For each $i$ let $C_i$ be an $[n_i\times m_i;n_i]_q$-MRD code, thus $\dim(C_i)=m_i$. Let
         $A_{i,1},\ldots,A_{i,m_i}$ be a basis of $C_i$.
         Using that $m_t\leq m_i$ for all~$i$, we may construct the code
         \[
             C=\big\langle(A_{1,1},\ldots,A_{t,1}),\ldots, (A_{1,m_t},\ldots,A_{t,m_t})\big\rangle.
          \]
          Then clearly every nonzero element of~$C$ has sum-rank weight~$N$ and $\dim(C)=m_t$, which is the Singleton Bound for $d=N$;
          see Theorem~\ref{th:singl}.
          Thus,~$C$ is an MSRD code.
\item In some cases the previous construction can be generalized to distance $d=N-1$.
         For instance, suppose $n_t\geq2$ and $2m_t\leq m_{t-1}$. 
         For $d=N-1$ the Singleton Bound in Theorem~\ref{th:singl} reads as $\dim(C)\leq 2m_t$.
         Choose now $[n_i\times m_i;n_i]$-MRD codes~$C_i$ for $i=1,\ldots,t-1$ and an $[n_t\times m_t;n_t-1]$-MRD code~$C_t$.
         Then $\dim(C_i)=m_i$ for $i=1,\ldots,t-1$ and $\dim(C_t)=2m_t$.
         Now we can mimic the construction from~(a) to obtain an MSRD code with distance $N-1$.
         The construction generalizes further to distance $d=N-\alpha$, but needs the stronger assumptions $n_t\geq\alpha+1$ 
         and $(\alpha+1)m_t\leq m_{t-1}$ 
\item Let us consider again $d=N-1$ in~(b).  Then MSRD codes do not exist for all ambient spaces~$\Pi$.
         For instance, for $m:=m_1=\ldots=m_t$ and $n_1\geq\ldots\geq n_t$ one can apply Corollary~\ref{exclMW} to rule 
         out, for many cases, the existence of an MSRD code with sum-rank distance $d=N-1$. 
         With the aid of~\eqref{e-fl} one obtains 
         $\omega(\tilde{u})=q^{2m}-q^m\sum_{i=1}^t(q^{n_i}-1)/(q-1)$ for $\tilde{u}=(n_1,\ldots,n_t)$.
         Thus, for instance, $\omega(\tilde{u})< q^{2m}-q^{(n_1-1)m}$, and hence $\omega(\tilde{u})<0$ whenever $n_1\geq3$.
\end{alphalist}
\end{const}

The above examples can be extended as follows.

\begin{const}\label{E-MSRD111}
Let $\F=\F_q$ and consider 
\[
     \Pi=\Pi_q\big(n_1\times m_1\mid\cdots\mid n_{t_1}\times m_{t_1}\mid \underbrace{1\times 1\mid\cdots\mid1\times1}_{t_2\text{ blocks}}\big),
\]
with the usual assumption $m_1\geq\ldots\geq m_{t_1}$ and $n_j\leq m_j$. 
Assume furthermore that $t_2\geq m_{t_1}$.
For each $j\in[t_1]$ let $C_j$ be an $[n_j\times m_j;n_j]_q$-MRD code and let $A_{j,1},\ldots,A_{j,m_j}$ be a basis of~$C_j$.
Furthermore, let $G=(g_{ij})\in\F^{m_{t_1}\times t_2}$ be the generator matrix of  an MDS code.
Define the sum-rank metric code $C\leq\Pi$ generated by
\[
    \big(A_{1,i},\ldots,A_{t_1,i},(g_{i1}),\ldots,(g_{it_2})\big)\ \text{ for } i=1,\ldots,m_{t_1}.
\]
Clearly $\dim(C)=m_{t_1}$ and $\srk(C)=\sum_{j=1}^{t_1} n_j+t_2-m_{t_1}+1$.
Thus Theorem~\ref{th:singl} implies that~$C$ is an MSRD.
If the MDS code is non-trivial, this construction requires $t_2\leq q+1$ (up to some exceptions in the MDS conjecture), 
while the trivial MDS codes of length~$t_2=m_{t_1}$ or $t_2=m_{t_1}+1$ work for every field.
For instance, the $2$-dimensional code
\[
  C=\big\langle \big((1,0),\ldots,(1,0),\,(1),\,(1),\,(0)\big),\  \big((0,1),\ldots,(0,1),\,(0),\,(1),\,(1)\big)\big\rangle
\]
in $\Pi(1\times2\mid\cdots\mid1\times2\mid 1\times1\mid1\times 1\mid1\times1)$
is MSRD with sum-rank distance~$t_1+2$ and exists over any finite field~$\F$.
\end{const}

\begin{const}\label{E-Combine}
We can also combine the constructions of Examples~\ref{E-CompMat} and~\ref{E-MSRD111}:
Consider the data from Example~\ref{E-MSRD111} and assume $m_{t_1}$ can be factored as $m_{t_1}=\hat{m}a$, where $a\leq t_2$.
Choose an MDS code in $\F_{q^{\hat{m}}}^{t_2}$ of Hamming distance $t_2-a+1$. 
Hence its $\F_q$-dimension is~$m_{t_1}$.
Then combining Examples~\ref{E-CompMat} and~\ref{E-MSRD111} we can construct a code 
$C\leq\bigoplus_{j=1}^{t_1}\F^{n_j\times m_j}\oplus\bigoplus_{j=1}^{t_2}\F^{1\times\hat{m}}$ of dimension~$m_{t_1}$ and sum-rank distance $\sum_{j=1}^{t_1} n_j+t_2-a+1$.
Again, the Singleton Bound shows that the code is MSRD.
\end{const}

In the special case where $\F^{n_j\times m_j}=\F^{1\times m}$ for all~$j$, Example~\ref{E-MSRD111} can be extended further into a different direction.

\begin{const}\label{E-MSRD111Ext}
Consider Example~\ref{E-MSRD111}. 
Set $s=t_1$ and let $n_j=1$ for all $j\in[s]$, $m:=m_1=\ldots=m_s$, and 
$t_2=m$.
Hence $\srk(C)=s+1$.
We may choose the same $[1\times m;1]_q$-MRD codes~$\hat{C}$ for the components.
Let $A_1,\ldots,A_m$ be a basis of~$\hat{C}$ (e.g., the standard basis vectors).
Choosing $G=I_m$ as a generator matrix of the trivial $[m,m]$-MDS code, the code~$C$ from Example~\ref{E-MSRD111} reads as 
\[
  C=\Big\langle \big(\underbrace{A_j,\ldots,A_j}_{s\text{ entries}},\,(0),\ldots,(0),(1),(0),\ldots,(0)\big)\,\Big|\, j\in[m]\Big\rangle 
  \leq\bigoplus_{j=1}^{s}\F^{1\times m}\oplus\bigoplus_{j=1}^m\F^{1\times1},
\]
and where $(1)$ is at the $j$th position.
Assume now $s\leq m+\binom{m}{2}+1$.
We will now extend~$C$ to an MSRD code $\tilde{C}\leq\Pi':=\bigoplus_{j=1}^{s+1}\F^{1\times m}\oplus\bigoplus_{j=1}^{m+1}\F^{1\times 1}$
of sum-rank distance $s+2$. It thus has dimension~$m+1$.
If $s>m$, choose $s-m-1$ distinct 2-subsets $\{\alpha_k,\beta_k\}$ of~$[m]$.
This is possible thanks to $s\leq m+\binom{m}{2}+1$.
Set $B_k=A_{\alpha_k}+A_{\beta_k}$ for $k\in[s-m-1]$.
Then $B_1,\ldots,B_{s-m-1}$ are distinct.
Define
\begin{align*}
  Z_j&= \big(A_{1},A_{j},\ldots\ldots\ldots\ldots\ldots\ldots\ldots\ldots\ldots,A_{j}\,,\,(0),\ldots,(0),(1),(0),\ldots,(0),(0)\big)\ \text{ for }j\in[m],\\
  Z_{m+1}&= \big(A_{2},A_{1},\ldots,A_{m},B_1,B_1,B_2,\ldots,B_{s-m-1},(0),\ldots,(0),(0),(0),\ldots,(0),(1)\big),
\end{align*}
where again~$(1)$ is at the $j$th position. 
Note that in ~$Z_{m+1}$ the block $(A_{1},\ldots,A_{m},B_1,B_1,B_2,\ldots,$ $B_{s-m-1})$ is of length~$s$. 
If $s\leq m$, only the matrices $A_1,\ldots,A_{s}$ occur.
If $s>m+1$, then $B_1$ occurs twice whereas all other $B_k$ occur once.
Let $\tilde{C}=\langle Z_1,\ldots,Z_{m+1}\rangle$. Then $\tilde{C}\leq\Pi'$ and $\dim(\tilde{C})=m+1$.
We show that~$\tilde{C}$ has sum-rank distance $s+2$. 
This is clearly the case for the subcode $\langle Z_1,\ldots,Z_{m}\rangle$.
Furthermore $\srk(Z_{m+2})=s+2$.
Hence it remains to consider an arbitrary linear combination $X=Z_{m+1}+\sum_{j=1}^m x_j Z_j$, where $x_j\in\F$ are not all zero.
It is of the form 
\[
       X=(X_1,\ldots,X_{s+1},(x_1),\ldots,(x_m),(1)).
\]
In the following we restrict ourselves to the case where $s\geq m$. 
The other case is similar, and actually much simpler.
Note first that $X_1\neq0$.
For the other components we consider the following cases.
\\
\underline{Case 1:} $(B_1,B_1)$ is eliminated by the linear combination $\sum_{j=1}^m x_j Z_j$.
In this case $x_j\neq0$ for at least two values of~$j$. 
Furthermore, none of the entries $A_1,\ldots,A_m$ of $Z_{m+1}$ is cancelled and likewise no $B_k,\,k>1,$ is cancelled. 
All of this leads to $\srk(X)\geq1+m+(s-m-2)+3=s+2$.
\\
\underline{Case 2:} There exists $k>1$ such that $B_k$ is cancelled. 
Again, $x_j\neq0$ for at least two values of~$j$ and no $A_1,\ldots,A_m$ and likewise no $B_j,\,j\neq k$, is cancelled.
Thus  $\srk(X)\geq1+m+(s-m-1)+3=s+3$.
\\
\underline{Case 3:} There exists $j$ such that $A_j$ is cancelled. 
Then no $B_k$ is cancelled and moreover $x_j\neq0$. Thus $\srk(X)\geq1+(m-1)+(s-m)+2=s+2$.
\\
The case where no cancellation occurs is obvious.
We conclude that~$\tilde{C}$ has sum-rank distance~$s+2$ and therefore is MSRD due to the Singleton Bound.
The code in  Example~\ref{E-MSRD8} is a special instance of the above construction. 
\end{const}

\subsection{Lifting Construction}

Finally, we present a construction of sum-rank metric codes that combines a Hamming-metric and a rank-metric code
with each other. We call this the \textit{lifting construction} because the Hamming-metric code serves
to \textit{lift} the rank-metric code to a sum-rank metric code.

\begin{theorem}[\textbf{Lifting Construction}]\label{th:lifting}
Consider~$\Pi$ as in~\eqref{e-Pi}, \eqref{e-nimi}.
For $i\in[t]$ choose integers  $\delta_i\in[n_i]$ and $[n_i\times m_i;\delta_i]_q$-MRD codes~$\mC_i$.
Let $h$ be an integer with $1 \le h \le \min\{m_i(n_i-\delta_i+1) \mid i\in[t]\}$ and let
$H \le \F_{q^{h}}^t$ be a nonzero $\F_q$-linear code of minimum Hamming distance $\Delta$.
Fix $\F_q$-linear injections $\varphi_i :\F_{q^h} \to \mC_i$, $i \in [t]$, and let
\[
   C:=\{(\varphi_1(\alpha_1),...,\varphi_t(\alpha_t)) \mid (\alpha_1,...,\alpha_t) \in H\}.
\]
Then~$C$ is a linear sum-rank metric code in~$\Pi$ with 
\[
          \dim(C)=\dim_{\F_q}(H)\ \text{ and }\ \srk(C)\geq \min\Big\{\sum_{i \in I} \delta_i\,\Big|\,  I\subseteq [t],\, |I|=\Delta\Big\}.
\]
\end{theorem}

\begin{proof}
The map $H \longrightarrow C$ given by $(\alpha_1,...,\alpha_t) \longmapsto
(\varphi(\alpha_1),...,\varphi(\alpha_t))$ is $\F_q$-linear and injective, which shows that $C$
and $H$ have the same dimension over $\F_q$.
Now fix a nonzero codeword $(\varphi(\alpha_1),...,\varphi(\alpha_t)) \in C$. Since $\varphi_1,...,\varphi_t$ are injections,
we have either $\rk(\varphi_i(\alpha_i)) \ge \delta_i$ or $\alpha_i=0$.
Since $C$ has minimum Hamming distance $\Delta$, this gives the desired lower bound 
on~$\srk(C)$.
\end{proof}

Note that the construction of Theorem~\ref{th:lifting} generalizes Example~\ref{E-CompMat}. 
Indeed, let $m_i=m$ and $n_i=\delta_i=1$ for all $i\in[t]$, and hence $C_i=\F_q^m$ for $i\in[t]$.
Then we may choose $h=m$ so that~$H$ is an MDS code in $\F_{q^{m}}^t$ with Hamming distance~$\Delta$. 
Then the above construction returns the MSRD code from Example~\ref{E-CompMat}. 

\begin{remark}
We can use Theorem~\ref{th:lifting} to construct sum-rank metric codes meeting the Induced Plotkin Bound of \cref{th:induced}.
Let $m_i=m$ and $n_i=\delta_i=n$ for all $i\in[t]$.
Let $H \le \F_{q^m}^t$ be an $\F_q$-linear code of minimum distance $\Delta$
meeting the Plotkin bound for the Hamming metric, i.e.,
\[
     q^m \Delta > (q^m-1)t \qquad \mbox{and} \qquad |H| = \Big\lfloor\frac{q^m \Delta}{q^m\Delta - (q^m-1)t}\Big\rfloor.
\]
Choosing $[n\times m;n]$-MRD codes~$\mC_i$ and injections~$\varphi_i$ as in Theorem~\ref{th:lifting}
we obtain a code $C\leq\Pi$ of cardinality~$|H|$ and sum-rank distance $n\Delta$.
Since $q^m \Delta > (q^m-1)t$, we have $q^m n\Delta > (q^m-1)nt=(q^m-1)N$.
Therefore the Induced Plotkin Bound from \cref{th:induced} reads as
\[
   |C| \le \Big\lfloor\frac{q^m n\Delta}{q^m n \Delta - (q^m-1)N}\Big\rfloor
    = \Big\lfloor\frac{q^m \Delta}{q^m\Delta - (q^m-1)t}\Big\rfloor
\]
and is met with equality.
We illustrate this more explicitly in the following example, where we lift the simplex code.
\end{remark}

\begin{example}
Let $q=2$, $m=4$ and $n=3$. Let $r=3$ and $t=(16^3-1)/15 = 273$. 
Let $H \le \F_{16}^{273}$ be the simplex code of dimension $r$ over $\F_{16}$. 
The code is linear over $\F_{16}$, but we only need its linearity over $\F_2$. 
The minimum distance of $H$ is $\Delta=16^2=256$. 
Therefore \cref{th:lifting} gives us a code
$C$ with sum-rank distance at least 
$256 \cdot 3=768$ (in fact, exactly $768$). 
The code meets the Induced Plotkin Bound with equality, as 
$|C|=|H| = 16^3=\textnormal{4,096}=q^m\Delta/(q^m\Delta-(q^m-1)t)$.
\end{example}


\bigskip

\bigskip

\section*{Summary} 

We have investigated the fundamental properties of sum-rank metric codes with respect to 
bounds, their asymptotic versions, duality, and existence/optimality.
In contrast to the previous literature, we consider codes with possibly variable block sizes, which therefore do not necessarily have a representation as a linear code over an extension field.

We have seen that the theory of sum-rank metric codes touches aspects of both rank-metric and Hamming-metric codes.
Our approach has led us to various new bounds for sum-rank metric codes. 
In our asymptotic analysis of these bounds, we have considered the case that the number of blocks~$t$ goes to infinity and have observed
a behaviour that reverts to the case of all block lengths $m_i$ being equal.

With respect to duality considerations, we have shown that MacWilliams identities hold for the support and rank-list distributions, 
but not for the sum-rank distribution of a code: two codes with the same sum-rank distribution may have dual codes whose sum-rank distributions are different.

A substantial part of the paper is devoted to MSRD codes, the sum-rank analogues of MDS and MRD codes. 
We have shown however, unlike their Hamming and rank-metric counterparts, that the MSRD property is not an invariant of 
duality, except in the case that all blocks have the same number of columns.
In that case,  we have derived an existence criterion for MSRD codes via duality considerations.
An interesting question on MSRD codes relates to the maximum number of blocks $t$ that such a code can have.
In the instance that all block sizes are equal, we have derived an upper bound on this value.

Finally, we have provided some constructions of optimal codes for certain sets of parameters, illustrating the possible behaviours of these objects. 

There are many  aspects of sum-rank metric codes that are currently open for further research. 
The most obvious of these are general constructions of families of optimal and extremal codes.

\bigskip 

\bigskip

\bibliographystyle{abbrv}
\bibliography{sr}
\end{document}